\documentclass[aps,prl,reprint,superscriptaddress,nobibnotes, nofootinbib]{revtex4-2}
\usepackage[normalem]{ulem}
\usepackage{xcolor, graphicx}
\usepackage{amsthm, amsfonts, amsmath, amssymb, bbding}
\usepackage{mathtools, physics, bm}
\usepackage{hyperref}

\newcommand{\dup}{\mathrm{d}}
\newcommand{\1}{\text{\usefont{U}{bbold}{m}{n}1}}
\newcommand{\Hil}{\mathcal{H}}
\newcommand{\N}{\mathbb{N}}
\newcommand{\R}{\mathbb{R}}
\newcommand{\C}{\mathbb{C}}
\newcommand{\States}{\mathcal{S}}
\newcommand{\Exp}{\mathcal{E}}
\newcommand{\measurements}{\mathcal{M}}
\newcommand{\outcomes}{\mathcal{O}}
\newcommand{\channels}{\mathcal{C}}
\DeclareMathOperator\relu{R}

\theoremstyle{plain}
\newtheorem{theorem}{Theorem}
\newtheorem{proposition}[theorem]{Proposition}
\newtheorem{lemma}[theorem]{Lemma}
\newtheorem{corollary}[theorem]{Corollary}
\newtheorem{definition}{Definition}

\allowdisplaybreaks[1]
\begin{document}


\title{On the Dynamics of Local Hidden-Variable Models}

\author{Nick von Selzam}
\email{nick.von-selzam@mpl.mpg.de}
\affiliation{Max Planck Institute for the Science of Light, 91058 Erlangen, Germany}
\affiliation{Department of Physics, Friedrich-Alexander-Universität Erlangen-Nürnberg, 91058 Erlangen, Germany}

\author{Florian Marquardt}
\email{florian.marquardt@mpl.mpg.de}
\affiliation{Max Planck Institute for the Science of Light, 91058 Erlangen, Germany}
\affiliation{Department of Physics, Friedrich-Alexander-Universität Erlangen-Nürnberg, 91058 Erlangen, Germany}

\date{\today}

\begin{abstract}

Bell nonlocality is an intriguing property of quantum mechanics with far-reaching consequences for information processing, philosophy, and our fundamental understanding of nature. However, standard Bell nonlocality is a statement about correlations at a fixed point in time only. It does not take into account dynamics, i.e.,~time evolution of those correlations. Consider a dynamic situation where the correlations remain local for all times. Then, at each moment in time, there exists a local hidden-variable (LHV) model reproducing the momentary correlations. Can the time evolution of the correlations then be captured by evolving the hidden variables? We make this question precise by defining dynamical LHV models and motivating additional physical attributes such as locality and continuity of the dynamics. We then establish a series of complementary go and no-go theorems that together map out which combinations of these properties LHV dynamics can and cannot have. In particular, noninteracting quantum dynamics \emph{can} be represented on the hidden-variable level in a local way, but LHV dynamics can fail to exist for interacting systems: Even though the observed correlations are local at all times, their time evolution may nevertheless be incompatible with local hidden-variable dynamics. 

\end{abstract}

\maketitle

\section{Introduction}

Quantum mechanics allows for nonlocal correlations: it predicts measurement statistics that are incompatible with any local hidden-variable (LHV) model~\cite{Bell_1964}. Nonlocality is a subtle and not yet fully understood concept. For example, while pure entangled states always violate a Bell inequality~\cite{Gisin_1991, Gisin_1992, Popescu_1992}, entanglement and nonlocality are generally inequivalent~\cite{Werner_1989, Augusiak_2015}. Beyond the foundational and philosophical significance, nonlocality has practical applications such as device-independent quantum cryptography~\cite{Barrett_2005, Acin_2007} (see~\cite{Ghoreishi_2025} for a recent review). All of this concerns correlations at a fixed moment in time. In this work, we investigate their time evolution. 

The original idea of LHV models is that there exist some fundamental microscopic degrees of freedom, the so-called hidden variables~$\lambda$, that are supposed to provide a more informative description of the physical state of a system than the quantum state by locally determining the results of arbitrary measurements. Bell's theorem implies that such a description cannot exist for all quantum states. However, many states \emph{are} local and can be described in this manner. In such situations, an LHV model may be viewed as a potential genuine description of the microscopic physics. 

We argue that a complete microscopic model must also include the dynamics. Consider a situation where some initial local quantum states evolve but remain local for all times. Then, at each instant there exist LHV models that reproduce the momentary measurement statistics. But the LHV models corresponding to different times or states may be entirely distinct, for example, featuring different hidden-variable spaces. 
A more desirable scenario involves only a single hidden-variable space. Then, the time evolution of the correlations should emerge from some evolution of the hidden variables (see Fig.~\ref{fig: 1}). The question arises: Can the quantum dynamics of local correlations be captured by evolving the hidden variables of an underlying LHV model?

\begin{figure}[t]
    \centering
    \includegraphics[width=0.45\textwidth]{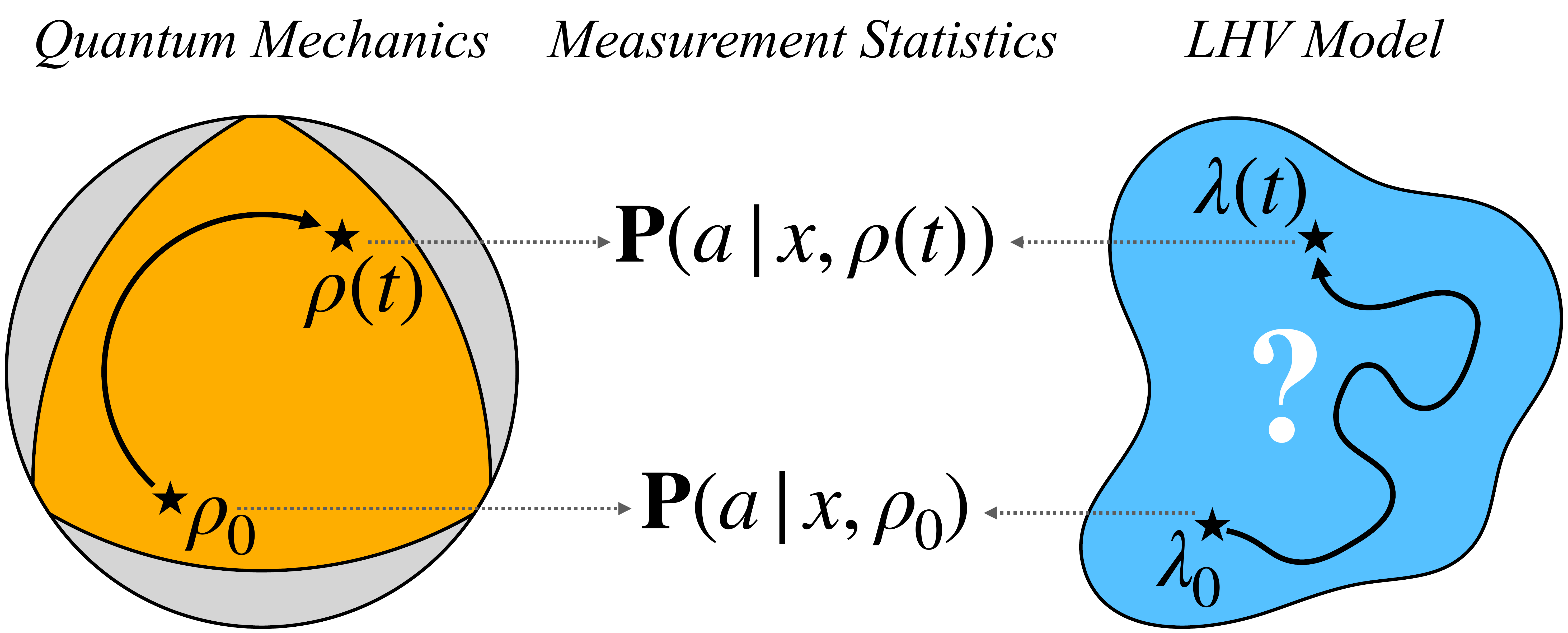}
    \caption{ 
    \textbf{LHV dynamics.} The grey disc represents all quantum states and the orange subset represents the local states. Suppose we have some local quantum states which remain local during time evolution. Then, the quantum measurement statistics at each moment in time, that is, the probabilities~$\mathbf{P}(a|x, \rho(t))$ to observe outcome~$a$ upon measuring~$x$, can be reproduced by LHV models. Can the time evolution of the measurement statistics then be captured by evolving the hidden variables?}
    \label{fig: 1}
\end{figure}

Either answer to this question would be interesting.
Suppose this were possible. In the best scenario this could then allow for efficient simulations of noisy quantum systems. For example, in the context of NMR computing the question has been raised and investigated whether unitaries acting on very noisy qubits can be regarded as true quantum computation~\cite{Braunstein_1999, Schack_Caves_1999, Menicucci_Caves_2002, Milburn_Laflamme_Sanders_Knill_2002}. The existence of local hidden-variable dynamics could answer this in the negative. 
On the other hand, suppose this were generally not possible. Then, we would have identified a new form of nonlocality linked to quantum dynamics, which is conceptually different from temporal Bell inequalities~\cite{Leggett_Garg_1985} and the \lq dynamical nonlocality\rq\ interpretation of the Aharonov-Bohm effect~\cite{Popescu_2010, Aharonov_Bohm_1959}. 

To address these questions, we rigorously define \emph{dynamical LHV models} (Sec.~\ref{sec: lhv models for sets of states} and~\ref{sec: lhv dynamics}): The idea is to model every physical process that corresponds to some quantum channel by a (possibly stochastic) transformation of the hidden variables. Most importantly, this must lead to the same change in measurement statistics as predicted by quantum mechanics. 
Then, we motivate a range of possible additional physical and mathematical assumptions (Sec.~\ref{sec: lhv dynamics: attributes} and~\ref{sec: lhv dynamics for a hamiltonian}). Among other aspects, this includes the notion of \emph{local LHV dynamics} whenever the hidden-variable dynamics respects the interaction structure of the modeled quantum dynamics\footnote{Note that the word \lq local\rq\ becomes slightly overloaded: \emph{Correlations} can be local in the sense of Bell, and \emph{dynamics} can be local in the sense of respecting the physical interaction structure.}. 

Mostly we are interested in the \lq continuum\rq\ setting, considering full-dimensional subsets of states, continuous sets of measurements such as all projective measurements, and continuous sets of quantum channels such as those given by all unitaries. Our framework also covers discrete settings with finite numbers of states, measurements and channels, though.

Based on this framework, we are able to both construct dynamical LHV models and prove no-go theorems in a variety of scenarios. Our results in this article include the following (see also Table~\ref{tab:results_overview} at the end of the article). 
We start by pointing out that for noninteracting quantum dynamics dynamical LHV models exist, even with many of the possible additional properties (Thm.~\ref{thm: local_unitaries} in Sec.~\ref{sec: noninteracting time evolution}). 
Then, we focus on interacting systems. 

Already for two interacting qubits obstructions can appear: We describe a particular simple LHV model for (sufficiently noisy) two-qubit states subject to a Heisenberg interaction. Valid time-dependent hidden-variable distributions exist in this case. Yet no smooth deterministic dynamics of the hidden variables can reproduce that evolution (Thm.~\ref{thm: heisenberg_counterexample} in Sec.~\ref{sec: a counterexample}). 

From this one might conclude that LHV dynamics in general does not exist for interacting quantum systems. However, we then demonstrate that even in the interacting case one can always \lq cheat\rq\ and get around such obstructions (Thm.~\ref{thm: cheat LHV dynamics} in Sec.~\ref{sec: cheat dynamics}, see also~\cite{Menicucci_Caves_2002}). There are two caveats: This cheat dynamics is highly nonlocal, and the construction features a hidden-variable space that is exponentially large in the number of particles. The reason is that, similar to Bohmian mechanics~\cite{Bohm_1952_1, Bohm_1952_2}, the global quantum state becomes part of the hidden variables.

In fact, it turns out that the single-particle hidden-variable space \emph{must} be exponentially large in the total number of particles, under the assumption that the physical processes are invertible and that the dynamics is continuous (Thm.~\ref{thm: dim-constraint for no go} in Sec.~\ref{sec: dimensionality constraint}). This leads to a no-go theorem excluding continuous reversible LHV dynamics for a sufficiently large number of particles (Cor.~\ref{cor:physical_nogo}). Even nonlocal LHV dynamics with unphysical all-to-all interactions is ruled out by this result. Without continuity no such dimensionality constraints can occur (Thm.~\ref{thm: simple hidden-variable space} in Sec.~\ref{sec: non-continuous lhv dynamics}).

Also the assumption of local LHV dynamics turns out to be constraining: In general, the considered quantum states must be more noisy than what is required for standard Bell-locality alone (Sec.~\ref{sec: local LHV dynamics}).

Finally, we observe that the hidden-variable distributions must be fine-tuned in conjunction with the dynamics: Demanding correct transformation of the measurement statistics for arbitrary initially valid hidden-variable distributions (what we will call \lq free gauge\rq) leads to a contradiction (Thm.~\ref{thm: no free gauge LHV dynamics} in Sec.~\ref{sec: no free gauge}).

We conclude by comparing and distinguishing our setting from existing notions such as \lq temporal\rq\ or \lq dynamical\rq\ nonlocality. We also explain the connection to some earlier works on dynamics of hidden-variable models, in contrast to dynamics of \emph{local} hidden-variable models. Finally, we highlight some known results that can now be viewed as special cases of our general framework. This includes a discussion of how Bohmian mechanics fits into the picture.

Throughout we aim to be as general and rigorous as possible while also extensively motivating our definitions and assumptions. For most proofs we describe the key ideas in the main text and defer the more technical parts to the appendix. There, we also give a brief introduction to stochastic maps and their properties (App.~\ref{app: stochastic_maps}), as much of the paper relies on them.

\section{LHV Models for Sets of States}
\label{sec: lhv models for sets of states}

A quantum state is said to be local with respect to some set of measurements if its measurement correlations can be explained by an LHV model. In this work, we want to think of LHV models not just in terms of their mathematical definition but as genuine physical models. 

The single-particle physics should not depend on the preparation procedure or the total number of particles\footnote{When we speak of \lq particles\rq\ in the following, we are referring to individual \lq parties\rq\ in the sense of quantum communication (which still can be composite objects), and measurements and hidden variables are local to these parties.}.
Imagine the following generic situation: We can prepare any number of some particular type of particles in many different ways (states) and then measure them individually, also in many different ways. Suppose the observed correlations for any one of those preparation procedures are local. Then, in principle, it could still be that the LHV models for two different states are wildly different, including the hidden-variable spaces. 
We argue that there should be a fixed single-particle hidden-variable space~$\Lambda_1$ that describes the possible microscopic \lq states\rq\ (hidden variables)~$\lambda$ of single particles, and a fixed local measurement rule~$q$ that determines the local measurement outcomes when measuring individual particles. Both of these may depend on the type of particles but not on the preparation procedure, i.e., the global quantum state. The latter should be encoded in the hidden-variable distribution only. 
See e.g.~Ref.~\cite{von_Selzam_Marquardt_2025} for this particular formulation of LHV models based on hidden variables assigned to individual particles. It is mathematically equivalent to the standard formulation in terms of a shared hidden variable and potentially different local response functions for each particle. 

In this spirit, below we will define \lq LHV base models\rq\ (specifying the single-particle hidden-variable space and local measurement rule) and \lq LHV models for sets of states\rq\ (all of them employing the same base model). While these particular considerations are probably not completely new, we explicitly include them here as they are required for defining \emph{dynamical} LHV models in the next section.

We consider a set of single-particle measurements~$\measurements_1$ together with a finite set of possible outcomes~$\outcomes_1$. Given a single-particle hidden-variable space~$\Lambda_1$, a measurement rule~$q$ assigns to each outcome~$a\in\outcomes_1$ a probability~$q(a|x, \lambda)$ for observing that outcome conditioned on the measurement~$x\in\measurements_1$ and the hidden variable~$\lambda\in\Lambda_1$. We have normalization over the outcomes,~$\sum_{a\in\outcomes_1} q(a|x, \lambda) = 1$.
Given~$\Lambda_1$ and~$q$ one can obtain many different LHV models by specifying a number of particles~$N$ and a joint hidden-variable distribution (a probability measure)~$\mu\in\mathbb{P}(\Lambda)$ on the~$N$-particle hidden-variable space~$\Lambda = \Lambda_1^N$. Then, the measurement statistics predicted by this LHV model, that is, the probabilities~$\mathbf{P}(a|x)$ to obtain outcomes~$a = (a_1,\ldots,a_N)\in \outcomes_1^N=\outcomes$ upon measuring~$x = (x_1,\ldots,x_N)\in\measurements_1^N=\measurements$ are given by
\begin{align}
    \mathbf{P}(a|x) = \int_{\Lambda}\dup \mu(\lambda) \prod_{k=1}^N q(a_k|x_k, \lambda_k).
\end{align}

Without loss of generality, we restrict to deterministic measurement rules,~$q(a|x, \lambda)\in\{0, 1\}$, in the usual way (see e.g.~\cite{Brunner_2014}). In this case, the measurement rule is the indicator function of some measurable subsets~$L^a_x \subseteq \Lambda_1$, i.e.,~$q(a|x, \lambda) = \1_{L^a_x}(\lambda)$. Normalization over outcomes means that these sets form a partition,~$\bigsqcup_a L^a_x = \Lambda_1$. Such a set~$L^a_x\subseteq\Lambda_1$ encompasses exactly those hidden variables~$\lambda$ that lead to outcome~$a$ upon measuring the observable~$x$. Therefore, (deterministic) local measurement rules~$q$ are in one-to-one correspondence with mappings
\begin{align}
    \measurements_1\ni x \mapsto L(x)\equiv (L^a_x)_{a\in\outcomes_1} \in \mathfrak{p}(\Lambda_1),
\end{align}
where~$\mathfrak{p}(\Lambda_1)$ denotes the measurable partitions of~$\Lambda_1$\footnote{We allow~$L^a_x = \emptyset$ for some~$a\in\outcomes_1$. This is needed in the case that a measurement~$x\in\measurements_1$ has fewer than~$\# \outcomes_1$ distinguishable outcomes.}. We will use both forms interchangeably.

For notational simplicity, we introduce the space of \emph{measurement events}~$\mathbf{M}$ combining measurements and outcomes,~$\mathbf{M}_1 \equiv \outcomes_1\times \measurements_1$ and~$\mathbf{M} \equiv \mathbf{M}_1^N$. The product of the local measurement rules becomes the indicator function of the corresponding product of sets: With~$L_m \equiv \bigtimes_{k=1}^N L^{a_k}_{x_k}\subseteq \Lambda$,~$m=(m_1,\ldots,m_N)\in\mathbf{M}$, we have
\begin{align}
    Q_m(\lambda)
    \equiv \prod_{k=1}^N q(m_k|\lambda_k)
    = \prod_{k=1}^N \1_{L^{a_k}_{x_k}}(\lambda_k)
    = \1_{L_m}(\lambda).
\end{align}
The local measurement statistics simplify to the probability of~$\lambda\in\Lambda$ lying in~$L_m\subseteq \Lambda$
\begin{align}
    \mathbf{P}(m) 
    &= \int_{\Lambda}\dup \mu(\lambda)\,Q_m(\lambda)
    = \int_{L_m}\dup\mu(\lambda) 
    = \mu(L_m).
\end{align} 

For the hidden-variable space~$\Lambda_1$ we do not want to allow arbitrary measurable spaces since those can have undesirable pathological properties: For example, conditional probability distributions may not exist~\cite{Fritz_2020}. We will instead require~$\Lambda_1$ to be a standard Borel space, a particular kind of measurable space that is sufficiently well-behaved for our purposes\footnote{A measurable space is a standard Borel space if it can be equipped with a metric such that it becomes a complete separable metric space and such that the Borel~$\sigma$-algebra coincides with the original~$\sigma$-algebra. It has been argued in the mathematical literature that standard Borel spaces are more suitable for the study of probability theory compared to general measurable spaces~\cite{Fritz_2020}.}.

This leads to the following definition.
\begin{definition}
    An \textbf{\emph{LHV base model}} for a set of single-particle measurements~$\measurements_1$ is specified by a \emph{single-particle hidden-variable space} (a standard Borel space)~$\Lambda_1$ and a \emph{local measurement rule}~$L:\measurements_1\to \mathfrak{p}(\Lambda_1)$ that assigns to each measurement~$x\in\measurements_1$ a measurable partition~$(L^a_x)_{a\in\outcomes_1}$ of~$\Lambda_1$.
\end{definition}
We will focus on systems of identical particles, such as~$N$ spins of the same type. Therefore, only one single-particle hidden-variable space~$\Lambda_1$ is required. This can be generalized easily.

On the quantum side, we consider an~$N$-qudit Hilbert space~$\Hil = (\Hil_1)^{\otimes N}$,~$\Hil_1 = \C^{D}$. The set of single-particle measurements~$\measurements_1$ can, for example, correspond to all projective or all POVM qudit measurements. The corresponding set of~$N$-qudit measurements~$\measurements = \measurements_1^N$ is given by~$N$-fold tensor products of single-particle measurements. Quantum mechanics describes a single-particle measurement event~$m_k\in\mathbf{M}_1$ by a positive-semidefinite operator~$M(m_k)\in\mathbb{B}_{\ge 0}(\Hil_1)$ (normalized as~$\sum_{a_k}M(a_k|x_k)=\1$) and a measurement event~$m\in\mathbf{M}$ by the tensor product~$M_m \equiv \bigotimes_{k=1}^N M(m_k)$. For example, for projective measurements the measurement operators~$M(m_k)$ are the eigenprojections of the Hermitian operator corresponding to the measured observable. The quantum measurement statistics are given by Born's rule: For a quantum state (density matrix)~$\rho \in \mathbb{D}(\Hil)$
\begin{align}
    \mathbf{P}(m|\rho) = \Tr(\rho M_m).
\end{align}

We denote by~$\mathcal{L}_\measurements\subseteq \mathbb{D}(\Hil)$ the subset of states that are local with respect to the measurements in~$\measurements$. By definition, the measurement statistics of any state~$\rho\in\mathcal{L}_\measurements$ can be reproduced by an LHV model. For a subset of local states~$\States\subseteq \mathcal{L}_\measurements$ we want the potentially different LHV models to be compatible in the following sense. 
\begin{definition}
    An \textbf{\emph{LHV model for a set of states}}~$\States$ and local measurements~$\measurements = \measurements_1^N$ consists of an LHV base model~$(\Lambda_1, L)$ for the single-particle measurements~$\measurements_1$ together with a \emph{hidden-variable distribution} (a probability measure)~$\mu_\rho \in \mathbb{P}(\Lambda)$ on the~$N$-particle hidden-variable space~$\Lambda = \Lambda_1^N$, for every state~$\rho\in\States$, such that the probabilities for measurement events~$m=(m_1,\ldots,m_N)\in\mathbf{M}$ are given by~$\mathbf{P}(m|\rho) = \mu_\rho(L_m)$. 
\end{definition}
We denote the mapping that assigns to each state a hidden-variable distribution by~$\mu: \States \to \mathbb{P}(\Lambda),\, \rho\mapsto \mu_\rho$. Any distribution~$\mu_\rho\in\mathbb{P}(\Lambda)$ that reproduces the measurement statistics of~$\rho\in\States$ correctly is called \emph{valid for}~$\rho$. For sets of states~$\States$ that contain only a single state,~$\States = \{\rho\}$, the above reduces to the standard notion of \lq an LHV model for the quantum state~$\rho$\rq. 

The main reason for the distinction between LHV base models and LHV models for states is the following: For a given LHV base model and state~$\rho$ there are potentially many different valid hidden-variable distributions~$\mu_\rho$. Any LHV model for a state implicitly makes a particular choice. We call such a choice a \emph{gauge} since it does not affect the measurement statistics (at least not the instantaneous ones, see Sec.~\ref{sec: no free gauge}). Working with LHV base models allows us to not commit to a particular gauge while still specifying the hidden-variable space and measurement rule.

One may ask whether it is a restriction to consider such LHV models for sets of states with a shared hidden-variable space. For a finite number~$K\in\N$ of local measurement settings,~$\measurements_1 \cong \{1,\ldots,K\}\equiv [K]$, our requirements are not a restriction. Indeed, the standard formulation of the local polytope is already of the desired form (see e.g.~\cite{Brunner_2014}): The single-particle hidden variables are lists of measurement results for all possible local measurements,~$\Lambda_1 = \outcomes_1^K$, and the local measurement rule selects the specified outcome,~$q(a|x, \lambda)=\delta_{a, \lambda_{x}}$ for~$a\in\outcomes_1,\,x\in [K]$, and~$\lambda = (\lambda_1,\ldots,\lambda_K)\in\outcomes_1^K$. This is an LHV base model that can represent \emph{all} states that are local with respect to the measurements in~$\measurements_1^N$, for any number of particles.

For continuous sets of measurements, such as all projective measurements, there is in general no guarantee that a fixed LHV base model, given by a single-particle hidden-variable space and local measurement rule, is sufficiently \lq expressive\rq\ to represent all states in~$\mathcal{L}_\measurements$. The generalization of the construction for the discrete setting would require a function space for~$\Lambda_1$. However, every given LHV base model yields an LHV model for all the states whose correlations can be captured by \emph{some} hidden-variable distribution. We refer to these states as \emph{representable}. 
We demonstrate in App.~\ref{app: general lhv models} how to construct finite-dimensional hidden-variable spaces and measurement rules that can be made arbitrarily expressive in a systematic manner, following our previous work~\cite{von_Selzam_Marquardt_2025}.

\section{LHV Dynamics}
\label{sec: lhv dynamics}

We now introduce a precise notion of LHV dynamics. We will distinguish between the underlying physical processes themselves, and the corresponding operations on the level of quantum states and on the level of the hidden variables. Instead of trying to directly map quantum channels to transformations of the hidden variables, we take the physical processes to be fundamental, such that everything (on the quantum and LHV side) derives from them. Adopting this point of view naturally gives rise to the following definitions.

We consider a set of experimental procedures or physical \emph{processes} which we denote by~$\Exp$. According to quantum mechanics every process~$g\in\Exp$ can be described by a quantum channel~$\Phi_g \in \mathrm{C}(\Hil)$\footnote{For simplicity we only consider particle-number preserving processes. Then quantum channels are completely positive trace-preserving linear maps~$\Psi: \mathbb{D}(\Hil) \to \mathbb{D}(\Hil)$.}.
A natural example to keep in mind is time evolution with respect to a fixed Hamiltonian~$H$. In that case the physical processes~$g\in \Exp$ simply correspond to different choices of evolution times~$t\in[0, \infty)$, and the quantum dynamics for each such choice is given by the unitary time-evolution operator,~$\Phi_t \cong U_t = \exp(-iHt)$. 

We argue that every such process also transforms the hidden variables, which then leads to a change of the measurement statistics. Crucially, we do \emph{not} want to allow arbitrary transformations of the hidden-variable distributions~$\mu_\rho$, only those that are induced by transforming the hidden variables~$\lambda$. The reason is that we assume the statistical distribution of the hidden variables to be an emergent effect due to our ignorance or potentially due to our fundamental inability to determine their precise value. Consequently, dynamics must occur on the level of individual hidden variables.

The most general way to transform a hidden variable~$\lambda\in\Lambda$ is stochastic, that is, by mapping it to a distribution,~$\lambda \mapsto \lambda' \sim T(\lambda) \in \mathbb{P}(\Lambda)$. Mathematically, this means for every process~$g\in\Exp$ we require a stochastic map~$T_g:\Lambda \to \mathbb{P}(\Lambda)$. Here,~$T_g(A|\lambda)\equiv (T_g(\lambda))(A)$ describes the probability for~$\lambda\in\Lambda$ being transformed into something in the subset~$A\subseteq \Lambda$ during the process~$g$. Just like probability measures are a generalization of probability densities, stochastic maps can be seen as a generalization of transition densities. See App.~\ref{app: stochastic_maps} for a short introduction to stochastic maps including the notation and particular notions and properties that we will need in the following. 

We denote the set of all stochastic maps on~$\Lambda$ by~$\mathrm{End}(\Lambda)$. If the hidden variables are initially distributed according to~$\mu\in\mathbb{P}(\Lambda)$ the transformation of the hidden variables described by a stochastic map~$T_g$ leads to a new distribution. The new distribution~$\mu'\in\mathbb{P}(\Lambda)$ is the push-forward of~$\mu$ with respect to~$T_g$, given by the chain rule of probabilities:
\begin{align}
    \mu'(A) = (T_g)_\# \mu(A) = \int_\Lambda T_g(A|\lambda) \dup \mu(\lambda).
\end{align}

The fundamental compatibility constraint is that the measurement statistics transform exactly as predicted by quantum mechanics. That is, for all considered measurement events~$m\in\mathbf{M}$, all states~$\rho\in\States$, and all processes~$g\in\Exp$ we must have
\begin{align}
    \mathbf{P}(m|\Phi_g[\rho]) \stackrel{!}{=} (T_g)_\#\mu_\rho(L_m).
\end{align}

Processes can be composed, and the dynamics must respect this. Mathematically, the set of processes~$\Exp$ is a monoid\footnote{A monoid is a set equipped with an associative binary composition operation and a two-sided identity element. The structure is very similar to that of a group except that inverses do not necessarily exist. A monoid where every element has a two-sided inverse is a group.}: There is a trivial process~$\mathrm{Id}\in\Exp$ that does nothing, e.g., time evolution over a time interval of length zero. Also, multiple processes can be composed in an associative manner, e.g., evolve for a time interval of length~$t_1$ and then for a time interval of length~$t_2$. Similarly, both the set of quantum channels~$\mathrm{C}(\Hil)$ and the set of stochastic maps~$\mathrm{End}(\Lambda)$ are monoids.

The mappings of processes to their quantum and LHV counterparts,~$\Exp\ni g \mapsto \Phi_g \in \mathrm{C}(\Hil)$ and~$\Exp\ni g \mapsto T_g \in \mathrm{End}(\Lambda)$, must be monoid homomorphisms, respecting that structure (identity and composition): The trivial process~$\mathrm{Id}\in\Exp$ cannot lead to any physical change, hence~$\Phi_\mathrm{Id} = \mathrm{Id}_{\mathbb{D}(\Hil)}$ (the identity channel), and~$T_\mathrm{Id} = \mathrm{Id}_\Lambda$ (the deterministic identity map on~$\Lambda$). Moreover, if a process~$g_1\in \Exp$ is followed by a process~$g_2\in\Exp$, the overall physical process \emph{is}~$g_2\circ g_1 \in \Exp$. Therefore,~$\Phi_{g_2}\circ \Phi_{g_1} = \Phi_{g_2\circ g_1}$ and~$T_{g_2}\circ T_{g_1} = T_{g_2\circ g_1}$. 

By writing~\lq$T_g$\rq\ we implicitly assume that the dynamics \emph{only} depends on the process~$g\in\Exp$, and not on other things, such as the history of previous processes or states. The reasoning is that we assume all such potential additional information to already be part of the genuine microscopic physical states, i.e., the hidden variables~$\lambda\in\Lambda$. In other words, following the original spirit of Einstein, Podolsky and Rosen~\cite{Einstein_Podolsky_Rosen_1935}, we assume the hidden variables to be a \emph{complete} description of the microscopic physical state such that dynamics becomes Markovian in the above sense by definition.

In particular, the hidden-variable dynamics, given by the stochastic maps~$T_g$, does not depend on the quantum state~$\rho$ or the distribution~$\mu_\rho$. It is \emph{state-independent}. Dependence on the quantum state~$\rho$ would imply access to global information, which is unavailable: A particular instance~$\lambda\in\Lambda$ of the hidden variables \lq does not know\rq\ about the quantum state~$\rho$ or the distribution~$\mu_\rho$ it belongs to. In fact, the same~$\lambda$ may occur for different quantum states. 
This assumption is what makes the problem nontrivial: State-dependent hidden-variable dynamics generically exists, see Sec.~\ref{sec: cheat dynamics}.

\begin{figure} [t]
    \centering
    \includegraphics[width=\linewidth]{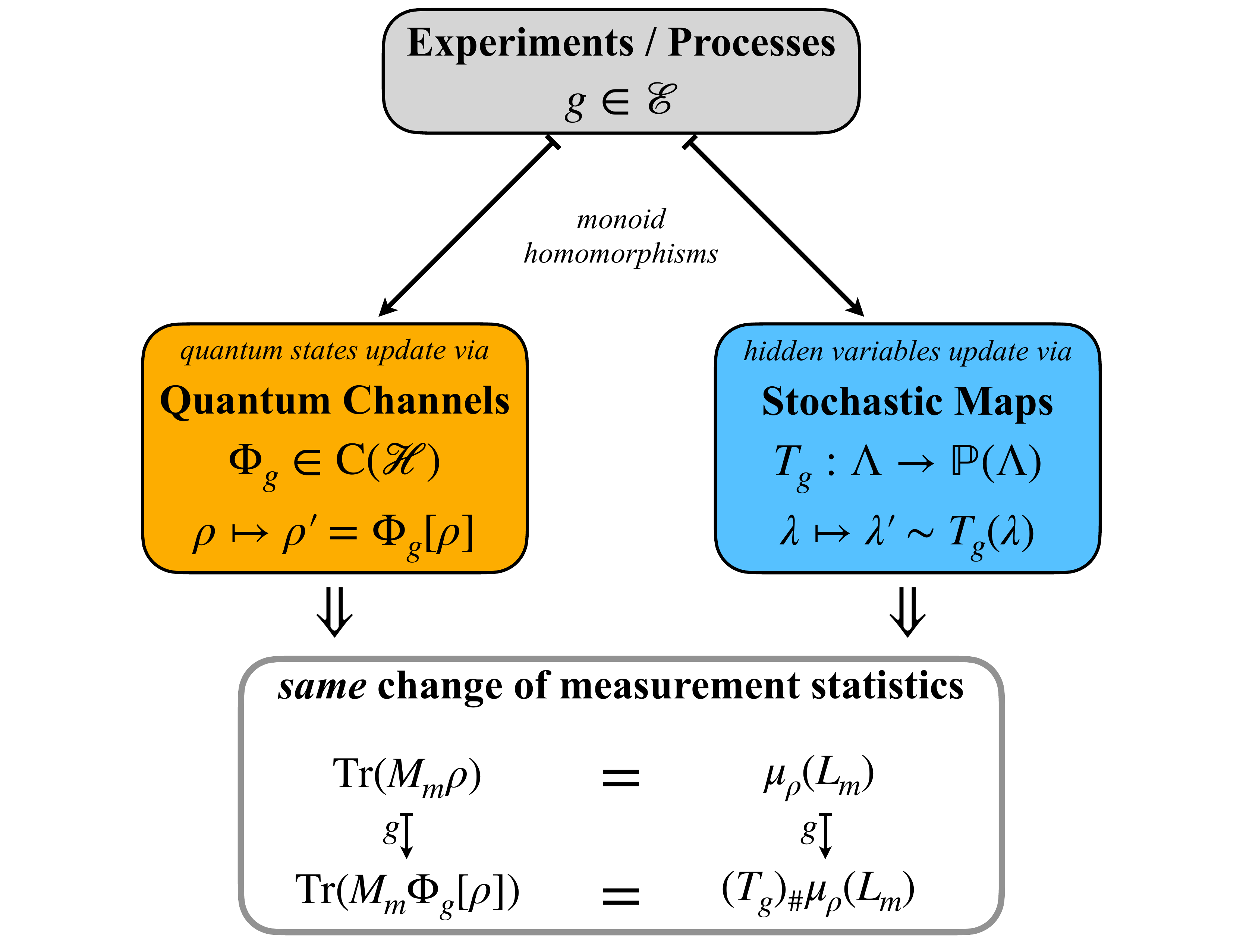}
    \caption{\textbf{LHV dynamics} means that for every physical process~$g\in\Exp$ there is a stochastic map~$T_g$ on the hidden-variable space~$\Lambda$ that induces the same change of measurement statistics as the quantum channel~$\Phi_g$ that is associated with the process~$g$, for all considered initial states~$\rho\in\States$ and all measurement events~$m\in\mathbf{M}$.}
    \label{fig: 2}
\end{figure}

Note that there may be nontrivial processes~$g\in\Exp,\ g \neq \mathrm{Id}$, that are nonetheless described by the identity channel on the quantum side,~$\Phi_g = \mathrm{Id}_{\mathbb{D}(\Hil)}$: for instance, periodic time evolution over one period. However, in general we \emph{cannot} demand~\lq$T_g = \mathrm{Id}_\Lambda$\rq\ in that case. The only constraint is that~$T_g$ leaves all measurement statistics invariant.

The implications of the composition property depend on the algebraic structure of the process monoid~$\Exp$. What always follows is that the measurement statistics transform correctly also under sequences of processes
\begin{align}
    \mathbf{P}(m|\Phi_{g_K}\circ \cdots \circ \Phi_{g_1} [\rho]) = (T_{g_K})_\# \cdots (T_{g_1})_\# \mu_\rho(L_m).
\end{align}
The least restrictive case is the following: Suppose there is some set of basic processes~$\Exp_0$ that we can realize experimentally, for instance, the physical gates available on a particular quantum computer. Then the full set of processes~$\Exp$ consists of arbitrary finite sequences of basic processes in~$\Exp_0$. Formally, processes~$g\in\Exp$ are finite words~$g=g_K\circ \cdots \circ g_1$ of basic processes~$g_1,\ldots,g_K\in\Exp_0$. In this case, there are no algebraic relations. Composing two processes corresponds to concatenating the words. Mathematically, these properties mean that~$\Exp$ is the \emph{free monoid} generated by~$\Exp_0$. As a set,~$\Exp = \bigsqcup_{K\in\N_0} \Exp_0^K$.
If~$\Exp$ is a free monoid then the equation above is \emph{all} that the composition property implies.
Additional structures that one may impose, in particular on the monoid of processes~$\Exp$ or the LHV dynamics~$\Exp \ni g \mapsto T_g \in \mathrm{End}(\Lambda)$, lead to additional constraints connecting different processes (see the next section).

We now formalize the discussion above. By \emph{quantum dynamics} for a monoid of physical processes~$\Exp$ and a submonoid of quantum channels~$\channels\subseteq \mathrm{C}(\Hil)$ we mean a monoid homomorphism~$\Phi:\Exp \to \mathrm{C}(\Hil)$ that assigns to each process~$g\in\Exp$ a quantum channel~$\Phi_g\in\channels$, such that every channel in~$\channels$ can be realized,~$\channels = \Phi(\Exp)$. 
By \emph{LHV dynamics} on a hidden-variable space~$\Lambda$ for a monoid of physical processes~$\Exp$ we mean a monoid homomorphism~$T: \Exp \to \mathrm{End}(\Lambda)$ that assigns to each process~$g\in\Exp$ a stochastic map~$T_g\in \mathrm{End}(\Lambda)$. 
We make the following definition (see also Fig.~\ref{fig: 2}).
\begin{definition} \label{def:LHV-dynamics}
    A \textbf{\emph{dynamical LHV model}} for a set of states~$\States$, a set of measurements~$\measurements=\measurements_1^N$, and a submonoid of quantum channels~$\channels$ is given by an LHV model~$(\Lambda_1, L, \mu)$ for~$\States$ and~$\measurements$ together with a monoid of physical processes~$\Exp$, LHV dynamics~$T: \Exp \to \mathrm{End}(\Lambda)$ on~$\Lambda$, and quantum dynamics~$\Phi:\Exp\to \mathrm{C}(\Hil)$ for~$\channels$ such that the measurement statistics transform correctly: For all states~$\rho\in\States$, processes~$g\in\Exp$ and measurement events~$m\in\mathbf{M}$
    \begin{align}
        \mathbf{P}(m|\Phi_g[\rho]) 
        &= (T_g)_\#\mu_\rho(L_m).
    \end{align}
\end{definition} 

For a dynamical LHV model we can require, without loss of generality, that the set of states~$\States$ is invariant under the set of quantum channels~$\channels$: If there were a state~$\rho\in\States$ and a channel~$\Phi_g \in \Phi(\Exp)=\channels$ such that~$\Phi_g[\rho] \notin \States$, then we could simply add~$\Phi_g[\rho]$ to the set of states~$\States$. Because the measurement statistics transform correctly, we have the valid hidden-variable distribution~$\mu_{\Phi_g[\rho]} \equiv (T_g)_\# \mu_\rho$ for~$\Phi_g[\rho]$. And because the measurement statistics of~$\mu_\rho$ transform correctly under compositions of processes, the statistics of~$\mu_{\Phi_g[\rho]}$ transform correctly under any process~$g'\in\Exp$, i.e.,
\begin{align} \notag
    (T_{g'})_\# \mu_{\Phi_g[\rho]}(L_m) 
    &= (T_{g'})_\# (T_g)_\# \mu_\rho(L_m) \\
    &= \mathbf{P}(m|\Phi_{g'}[\Phi_g[\rho]]).
\end{align}

The set of states~$\States$ being invariant under the channels~$\channels$ is already a notable constraint. For example, for the full (projective) unitary group,~$\channels = \mathrm{PU}(\Hil)$\footnote{We will speak of \lq unitaries~$U\in\mathrm{PU}(\Hil)$\rq\ in the following, implicitly ignoring global phases. Unitaries are regarded as quantum channels via their action on states:~$U[\rho]\equiv U \rho U^\dagger = \mathrm{Ad}_U(\rho)$.}, it immediately excludes all pure product states: These are local with respect to any set of measurements but can be mapped to pure entangled, and therefore nonlocal, states by certain unitaries. 

\begin{figure}[t]
    \centering
    \includegraphics[width=0.45\textwidth]{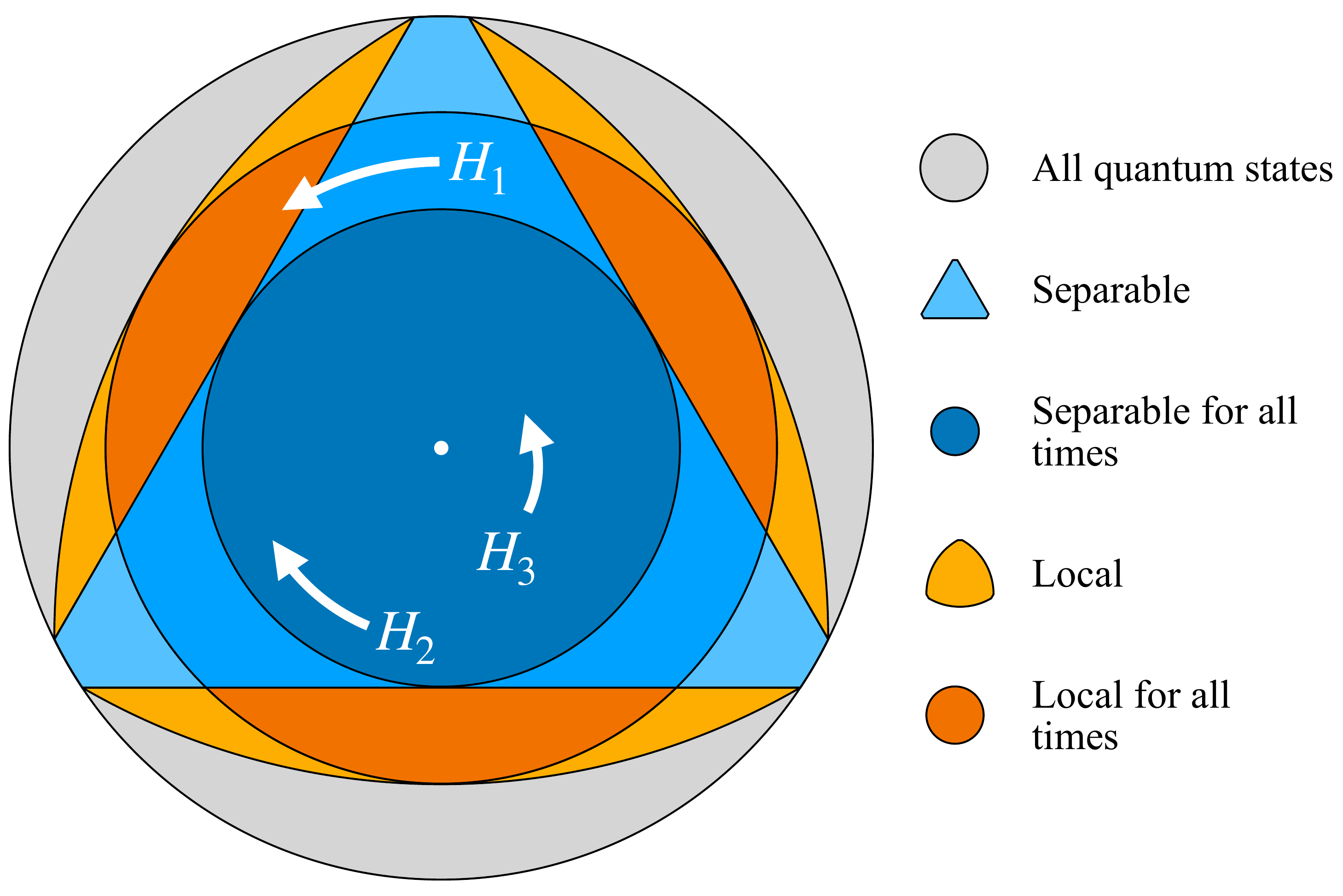}
    \caption{
    \textbf{The space of quantum states (schematic).} In this sketch, unitary evolution is represented by angular motion (indicated by the example Hamiltonians~$H_1,H_2,H_3$). Hence, the (convex) set of all states corresponds to a disc (grey), with the white dot at the center representing the maximally mixed state. All pure states lie on the boundary of this disc. The separable states (blue) are convex combinations of the pure product states. The local states (bright orange) are a convex superset of the separable states. The largest disc contained in the separable states contains all the states that remain separable under arbitrary unitary time evolution (absolutely separable states, dark blue). The local states contain a maximal disc of states that remain local under arbitrary unitary time evolution (absolutely local states, orange). 
    }
    \label{fig: 3}
\end{figure}

There are subsets of the set of all local states that have the required property, even independently of the considered measurements~$\measurements$. Specifically, denote by~$\rho_0 = \1/D^N$ the maximally mixed state. Then, a useful construction for our following arguments is the \emph{noisy ball of states}
\begin{align}
    \States_v \equiv \{v \rho + (1-v) \rho_0 \, |\, \rho \in \mathbb{D}(\Hil) \} \subseteq \mathbb{D}(\Hil).
\end{align}
It is full-dimensional, invariant under conjugation by arbitrary unitaries, and it contains only separable states for sufficiently small \emph{visibilities}~$v>0$~\cite{Braunstein_1999, Rungta_et_al_2001, Gurvits_Barnum_2002}. There are also (potentially larger) visibility thresholds below which all states~$\rho\in\States_v$ become local, see e.g.~\cite{Almeida_2007} and~\cite{Loubenets_2017} for some general bounds. This is visualized in Fig.~\ref{fig: 3}. 

LHV dynamics for all projective measurements and all unitaries requires \emph{absolutely local} states: In analogy to absolutely separable states~\cite{Kus_Zyczkowski_2001} we call a state~$\rho\in\mathbb{D}(\Hil)$ absolutely local if~$U[\rho]\equiv U \rho U^\dagger$ is local with respect to all projective measurements for all unitaries~$U\in\mathrm{PU}(\Hil)$\footnote{This definition is intentionally specific to the case of all projective measurements. The notion would also make sense for other sets of measurements. For example, for the CHSH setting the concept of \emph{absolutely Bell-CHSH local} states was introduced in~\cite{Roy_Bhattacharya_Mukherjee_Ganguly_2018}.}. 
For sufficiently small visibilities~$v>0$, the unitarily invariant noisy ball of states~$\States_v$ contains only absolutely local states. 

We can introduce several different notions for the precise sense in which LHV dynamics does or does not exist for a submonoid of quantum channels~$\channels \subseteq \mathrm{C}(\Hil)$. We say~$\channels$ \emph{has LHV dynamics} for states~$\States$ and measurements~$\measurements$ if there exists a dynamical LHV model for~$\channels$. One can make this progressively more restrictive: We say an LHV base model~$(\Lambda_1, L)$ has LHV dynamics for~$\channels$ if it can be extended to a dynamical LHV model for~$\channels$. Similarly, an LHV model~$(\Lambda_1, L, \mu)$ has LHV dynamics for~$\channels$ if it can be extended to a dynamical LHV model for~$\channels$. 
One may even pre-specify how the hidden-variable distributions should update under the processes~$g\in\Exp$: Our definition of an LHV model for a set of states~$\States$ includes a choice of gauge~$\mu_\rho\in\mathbb{P}(\Lambda)$ for every state~$\rho\in\States$. In the context of dynamical LHV models one may view this as the \lq initial gauge\rq\ since the dynamics then changes these distributions. One could also fix hidden-variable distributions~$\mu_{\rho, g}$ reproducing the statistics of~$\Phi_g[\rho]$ in advance. We refer to this as a \emph{dynamical gauge}. Then we say a dynamical gauge \emph{is compatible with LHV dynamics} if the underlying LHV base model has LHV dynamics with~$(T_g)_\# \mu_\rho = \mu_{\rho, g}$.  
This can be summarized as a hierarchy of implications visualized in Fig.~\ref{fig: 4}. 

This concludes our fundamental definition of LHV dynamics. In the following, we introduce additional properties that one may or may not demand LHV dynamics to have.

\section{LHV Dynamics: Attributes}
\label{sec: lhv dynamics: attributes}

On top of the core definitions above, this section introduces additional attributes that one may demand LHV dynamics to have: We motivate and define notions of local, continuous, deterministic, smooth, reversible and microscopic LHV dynamics, as well as LHV dynamics with consistent gauge and free gauge. These are modeling choices and our results depend on which combination of attributes is assumed. In contrast to the core definitions there is more freedom in the formalization here. 

For each attribute we first give the physical intuition and then motivate the mathematical modeling choices. This leads to formal definitions which are required as the basis for rigorous proofs. We encourage the reader to at least absorb the physical meaning of each notion, which is often apparent from the name itself. A summary of all attributes is provided in Table~\ref{tab:attributes_and_consequences}. 

\begin{figure}[t]
    \centering
    \includegraphics[width=0.80\linewidth]{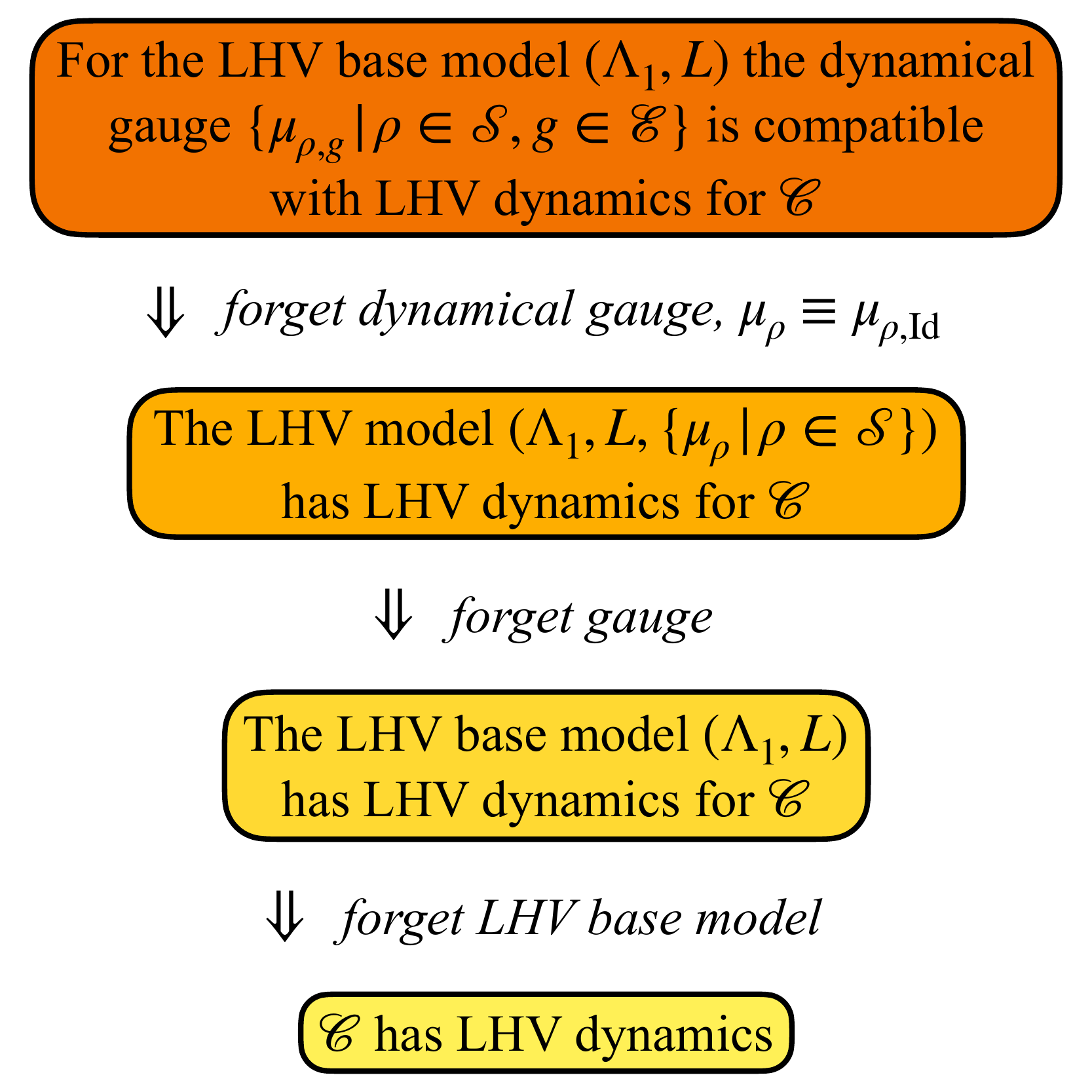}
    \caption{\textbf{Different notions for existence of LHV dynamics.} For fixed states~$\States$ and measurements~$\measurements$ we have a hierarchy of notions of \lq LHV dynamics exists\rq\ for a subset of channels~$\channels\subseteq \mathrm{C}(\Hil)$. Reading from top to bottom, each notion implies the more general notions below it by \lq forgetting\rq\ some of the structure. For no-go results, meaning \lq LHV dynamics does not exist\rq, the hierarchy flips. Then~\lq$\channels$ does not have LHV dynamics\rq\ becomes the strongest statement and the implications run in reverse.}
    \label{fig: 4}
\end{figure}

\paragraph{Local LHV Dynamics.} 
LHV dynamics should respect the physical interaction structure: Consider a process~$g\in\Exp$ during which two particular particles do not interact. Then neither the hidden-variable transformation~$T_g$ nor the quantum channel~$\Phi_g$ can couple them. After all, one could separate the particles space-like before the dynamical process takes place. 
More generally, there are processes~$g\in \Exp$ that only affect a particular subset~$\pi\subseteq [N]$ of the particles. We denote the collection of all such processes by~$\Exp_\pi\subseteq \Exp$. 
If we consider a larger subset of particles~$\pi \subseteq \pi' \subseteq [N]$ then any process that only affects particles in~$\pi$ in particular affects nothing outside of~$\pi'$. This means we have the inclusion~$\Exp_\pi \subseteq \Exp_{\pi'}$.
The trivial identity process does not affect any particles so~$\mathrm{Id}\in \Exp_{\emptyset} \subseteq \Exp_\pi$. If two processes only affect particles in~$\pi\subseteq [N]$, then so does their composition. This means the subsets~$\Exp_\pi\subseteq \Exp$ are actually submonoids. 
Finally, if two subsets of particles~$\pi, \pi'\subseteq [N]$ are disjoint,~$\pi\cap\pi'=\emptyset$, then the processes in~$\Exp_\pi$ cannot influence those in~$\Exp_{\pi'}$ and vice versa, they are independent. This means the order of such processes should not matter:~$gg'=g'g$ for~$g\in\Exp_\pi$ and~$g'\in\Exp_{\pi'}$. We write this commutation property as~$[\Exp_\pi, \Exp_{\pi'}]=0$ (even though, due to the lack of inverses, we cannot actually define such a commutator on~$\Exp$). 

We make the following definition\footnote{The structure introduced above bears formal similarities to some of the Haag-Kastler axioms in algebraic quantum field theory~\cite{Haag_Kastler_1964} if one replaces the monoid of processes by the algebra of quasilocal observables and subsets of particles by regions of spacetime.
We loosely adapt their naming in the following definition. The term \lq net\rq\ in this context seems to have been introduced only later, see e.g.~\cite{Haag_1993}.}
\begin{definition}
    A \textbf{\emph{local net of processes}} for~$N$ particles is a monoid~$\Exp$ together with submonoids~$\Exp_\pi\subseteq \Exp$ for every subset~$\pi\subseteq [N]$ such that we have
    \begin{itemize} \itemsep0em
        \item~$\Exp_\emptyset = \{\mathrm{Id}\}$ and~$\Exp_{[N]} = \Exp$,
        \item \emph{isotony:} if~$\pi\subseteq \pi'\subseteq [N]$ then~$\Exp_\pi \subseteq \Exp_{\pi'}$,
        \item \emph{local commutativity:} if~$\pi, \pi'\subseteq [N]$ and~$\pi\cap \pi' = \emptyset$ then~$[\Exp_\pi, \Exp_{\pi'}]=0$.
    \end{itemize}
\end{definition}

Both the monoid of stochastic maps~$\mathrm{End}(\Lambda)$ as well as the monoid of quantum channels~$\mathrm{C}(\Hil)$ naturally are local nets of processes: For a subset of particles~$\pi\subseteq [N]$, and with~$\#\pi\in \{0,\ldots, N\}$ corresponding to the number of particles in~$\pi$, denote the hidden-variable space for those particles by~$\Lambda_\pi \cong \Lambda_1^{\#\pi}$ and the Hilbert space by~$\Hil_\pi \cong \Hil_1^{\otimes \#\pi}$. Then~$\mathrm{End}(\Lambda_\pi) \subseteq \mathrm{End}(\Lambda)$ as well as~$\mathrm{C}(\Hil_\pi)\subseteq \mathrm{C}(\Hil)$ are submonoids. Here and throughout, we implicitly extend the stochastic maps or quantum channels, respectively, by the identity operation on the particles outside of~$\pi$. The inclusion (isotony) and commutation (local commutativity) properties then hold trivially. 

With this structure at hand, it becomes straightforward to formalize what \lq local dynamics\rq\ should mean. Namely, if a process~$g\in\Exp$ only affects particles in the subset~$\pi\subseteq [N]$, then the hidden variables for particles outside of~$\pi$ cannot be affected, and similarly the quantum channel describing such a process can act nontrivially only on~$\Hil_\pi$. In equations,~$T(\Exp_\pi)\subseteq \mathrm{End}(\Lambda_\pi)$ and~$\Phi(\Exp_\pi)\subseteq \mathrm{C}(\Hil_\pi)$. 

These two conditions alone are too weak: Based only on them, LHV dynamics could formally always be interpreted as being local by postulating that every process (except the identity) is all-to-all interacting,~$\Exp_{\pi}=\{\mathrm{Id}\}$ for all~$\pi\neq [N]$, which is clearly not what we want. To obtain a meaningful notion we must ground it in the interaction structure as given by quantum mechanics: If a channel~$\Phi_g\in\channels$ factorizes with respect to some partition~$(\pi_j)_{j=1}^J \in \mathfrak{p}([N])$ of the particles, that is,~$\Phi_g=\Psi_1\otimes \cdots \otimes \Psi_J$ with~$\Psi_j\in\mathrm{C}(\Hil_{\pi_j})$, then it should be possible to individually and locally realize the channels~$\Psi_j$ as well. This means there should exist some processes~$g_j\in\Exp_{\pi_j}$ with~$\Phi_{g_j}=\Psi_j$.

We summarize the discussion with the following definition.
\begin{definition}
    A dynamical LHV model, given by~$(\Lambda_1, L, \mu, \Exp, T, \Phi)$, has \textbf{\emph{local} LHV dynamics} if~$\Exp$ is a local net of processes and we have
    \begin{itemize} \itemsep0em
        \item \emph{locality (LHV):}~$T(\Exp_\pi)\subseteq \mathrm{End}(\Lambda_\pi)$ for all~$\pi\subseteq [N]$, 
        \item \emph{locality (QM):}~$\Phi(\Exp_\pi)\subseteq \mathrm{C}(\Hil_\pi)$ for all~$\pi\subseteq [N]$, 
        \item \emph{local realizability:} if~$\Phi_g=\bigotimes_{j=1}^J \Psi_j$ with~$\Psi_j\in\mathrm{C}(\Hil_{\pi_j})$,~$g\in\Exp$,~$(\pi_j)_{j=1}^J \in \mathfrak{p}([N])$, then~$\Psi_j\in \Phi(\Exp_{\pi_j})$.
    \end{itemize}
\end{definition}
An immediate consequence of local realizability is that for a subset~$\pi\subseteq [N]$ of the particles and any channel that only affects those particles,~$\Psi \in \channels\cap \mathrm{C}(\Hil_{\pi}) \equiv \channels_\pi$, it follows that~$\Psi\in\Phi(\Exp_\pi)$ by considering the partition~$(\pi, \pi^c)$ and using~$\Psi \widehat{=} \Psi \otimes \1_{\pi^c}$ (here~$\pi^c$ denotes the complement of~$\pi$ in~$[N]$). Together with the locality condition this implies~$\Phi(\Exp_\pi)=\channels_\pi$. 
Local realizability also constrains the set of channels~$\channels$ by itself: The factors of a factorizing channel in~$\channels$ must also lie in~$\channels$. Enlarging the set of channels accordingly is in general not without cost. Natural choices, such as all unitary channels, have this property, though.

\paragraph{Continuous LHV Dynamics.} 
Time evolution is usually expected to be a continuous process. This has two aspects: Similar physical states (in our case, hidden variables) should transform in a similar way, and slightly modifying the physical process, e.g., by letting the system evolve for a short additional time interval, should change the physical system only in a slightly different manner. The first property says that the hidden-variable transformations~$\Lambda \ni \lambda \mapsto T_g(\lambda)\in \mathbb{P}(\Lambda)$ are continuous. In equations,~$T_g\in \mathfrak{C}(\Lambda, \mathbb{P}(\Lambda))\subseteq \mathrm{End}(\Lambda)$, where~$\mathfrak{C}(\Lambda, \mathbb{P}(\Lambda))$ denotes the continuous maps from~$\Lambda$ into~$\mathbb{P}(\Lambda)$. The second property says that the mapping~$\Exp \ni g \mapsto T_g \in\mathfrak{C}(\Lambda, \mathbb{P}(\Lambda))$ is continuous. 
For this to make sense, we require topologies on the hidden-variable space~$\Lambda_1$, the space of probability measures~$\mathbb{P}(\Lambda)$, the set of processes~$\Exp$, and on the space of continuous stochastic maps~$\mathfrak{C}(\Lambda, \mathbb{P}(\Lambda))$. 

For continuous LHV dynamics we will require the hidden-variable space~$\Lambda_1$ and the set of processes~$\Exp$ to be topological manifolds. This means, every particular hidden variable~$\lambda \in \Lambda_1$ or process~$g\in\Exp$ can be specified by finitely many discrete and continuous parameters. In particular, this allows us to talk about the dimensions of these two spaces. More formally, by a topological manifold we mean a second countable locally Euclidean Hausdorff space. 
In this manuscript we use the following non-standard notion of \lq locally Euclidean\rq: We say a topological space is locally Euclidean if every point has an open neighborhood homeomorphic to an open subset of \emph{a} half-space~$\R^n_+=\{(x_1,\ldots,x_n)\in\R^n | \,x_n\ge 0\}$, where~$n\in\N_0$ may depend on the point. That is, we allow manifolds of non-pure dimension and with boundary. 
The reason is that many natural examples are of this form. For instance, nonnegative evolution times~$[0, \infty)$ or the set of all quantum channels~$\mathrm{C}(\Hil)$ are (topological) manifolds with boundary, and the free monoid of some at least one-dimensional set of generating processes leads to a manifold of non-pure (even unbounded) dimension. 
In the following, we will simply speak of \lq manifolds\rq, implicitly referring to \lq topological manifolds\rq\ in the above sense unless stated otherwise.

We require the hidden-variable space~$\Lambda_1$ to be finite-dimensional and write~$\dim \Lambda_1 < \infty$ for the largest dimension of any connected component. This is a necessary ingredient for the existence of a dynamical LHV model to imply that the dynamics can be simulated efficiently. Since we are dealing with finite-dimensional quantum systems this assumption seems natural. In particular, we exclude infinite-dimensional spaces, such as function spaces, for the single-particle hidden-variable space~$\Lambda_1$. 
As a manifold~$\Lambda_1$ equipped with its Borel~$\sigma$-algebra becomes a standard Borel space (e.g., via~\cite[Thm.~5.3]{Kechris_1995}), consistent with our general definitions. 

Regarding the set of processes~$\Exp$ we can take an operational point of view. We argue that any experiment can be specified by a finite number of parameters, at least in an approximate sense: In practice even an electric field, that could in principle be any continuous function, has to be specified by a finite amount of data. Note that our definition of a manifold allows the dimension of~$\Exp$ to be unbounded. This means the number of parameters that describe any particular process can be arbitrarily large. 
The topology should be compatible with the monoid structure, i.e., composing processes should be continuous. Taken together, we assume that~$\Exp$ is a topological monoid where the underlying topological space is a manifold.

For the space of hidden-variable distributions~$\mathbb{P}(\Lambda)$ we choose the topology of weak convergence. This means that two distributions are regarded as close if all expectation values are close. Since~$\Lambda_1$ and hence~$\Lambda=\Lambda_1^N$ are manifolds this topology on~$\mathbb{P}(\Lambda)$ is particularly well behaved (see App.~\ref{app: stochastic_maps}).
Instead of dealing with topologies on~$\mathfrak{C}(\Lambda, \mathbb{P}(\Lambda))$ we simply demand that the mapping~$\Exp\times \Lambda \ni (g, \lambda) \mapsto T_g(\lambda)\in\mathbb{P}(\Lambda)$ is \emph{jointly} continuous.

The quantum channels~$\Phi_g:\mathbb{D}(\Hil) \to \mathbb{D}(\Hil)$ are continuous. However, the quantum dynamics~$\Exp \ni g\mapsto \Phi_g \in \mathrm{C}(\Hil)$ being continuous is an additional assumption that we make for the same reasons as in the case of the LHV dynamics. 
\begin{definition}
    A dynamical LHV model has \textbf{\emph{continuous} LHV dynamics} if 
    \begin{itemize} \itemsep-.35em
        \item the single-particle hidden-variable space~$\Lambda_1$ is a finite-dimensional manifold,
        \item the set of processes~$\Exp$ is a topological monoid that is also a manifold, 
        \item the mapping~$\Exp\times \Lambda \ni (g, \lambda) \mapsto T_g(\lambda)\in\mathbb{P}(\Lambda)$ is jointly continuous for the weak topology on~$\mathbb{P}(\Lambda)$, 
        \item the quantum dynamics,~$\Exp \ni g \mapsto \Phi_g \in \mathrm{C}(\Hil)$, is continuous.
    \end{itemize}
\end{definition}
Of course, continuous LHV dynamics as defined above implies that the hidden-variable transformations are continuous,~$T(\Exp)\subseteq \mathfrak{C}(\Lambda, \mathbb{P}(\Lambda))$.

\paragraph{Deterministic LHV Dynamics.} 
Our general definition of LHV dynamics allows for stochastic transformations given by stochastic maps. One could make the additional assumption that the transformations~$T_g$ are actually deterministic, meaning a process~$g\in\Exp$ transforms any hidden variable~$\lambda\in\Lambda$ precisely to some~$\lambda'\in\Lambda$. 
Formally, this means~$T_g(A|\lambda) = \1_{A}(\tau_g(\lambda))$ for some measurable maps~$\tau_g:\Lambda \to \Lambda$. In this case we identify~$T_g$ with~$\tau_g$ and write~$T_g:\Lambda \to \Lambda$.
\begin{definition}
    A dynamical LHV model has \textbf{\emph{deterministic} LHV dynamics} if all hidden-variable transformations~$T_g:\Lambda \to \Lambda$ are deterministic measurable maps.
\end{definition}
For deterministic dynamics we have a monoid action~$\Exp\times \Lambda \ni (g, \lambda) \mapsto T_g(\lambda)\in\Lambda$. If the dynamics is additionally continuous, then this monoid action is continuous as well. If the dynamics is additionally local, and a process~$g\in\Exp_\pi$ affects only particles in~$\pi\subseteq [N]$, then the corresponding hidden-variable transformation only changes the hidden variables of those particles, and the transformation also depends only on them. That is,~$T_g: \Lambda_\pi \to \Lambda_\pi$, implicitly acting as the identity on the remaining particles. 
For deterministic LHV dynamics the measurement statistics transform by literally changing~$\lambda_k \mapsto [\tau_g(\lambda)]_k$, as one would expect:
\begin{align} \notag
    (T_g)_\#\mu_\rho(L_m) 
    &= \int_\Lambda \dup \mu_\rho(\lambda)\, T_g(L_m|\lambda) \\ \notag 
    &= \int_\Lambda \dup \mu_\rho(\lambda)\, \1_{L_m}(\tau_g(\lambda)) \\ 
    \label{eq: deterministic HV transform update}
    &= \int_\Lambda \dup \mu_\rho(\lambda)\, \prod_{k=1}^N q(m_k | \left[\tau_g(\lambda)\right]_k).
\end{align}
Equivalently, this means~$(T_g)_\#\mu_\rho(L_m) = \mu_\rho(\tau_g^{-1}(L_m)) = (\tau_g)_\# \mu_\rho(L_m)$ (the usual push-forward).

\begin{table*}[ht]
\renewcommand{\arraystretch}{1.0}
\setlength{\tabcolsep}{12pt}
\caption{\textbf{Attributes of LHV dynamics.} The table summarizes different attributes one may demand LHV dynamics to have on top of the fundamental definition (Def.~\ref{def:LHV-dynamics}) and some immediate consequences. Here~$\Exp$ is the monoid of physical processes,~$\Lambda = \Lambda_1^N$ is the~$N$-particle hidden-variable space,~$\mathbb{P}(\Lambda)$ are the probability measures (hidden-variable distributions) on~$\Lambda$ and~$\mathrm{End}(\Lambda)$ are the stochastic maps on~$\Lambda$ describing the hidden-variable dynamics. On the quantum side each process is associated with a quantum channel~$\Phi_g\in \mathrm{C}(\Hil)$. By~$\mathrm{PU}(\Hil)\subseteq \mathrm{C}(\Hil)$ we denote the projective unitary group.}
\label{tab:attributes_and_consequences}
\begin{tabular}{|l|p{6.15cm}|p{6.15cm}|}
\hline\hline
\textbf{Attribute} & \textbf{Description} & \textbf{Consequences} \\
\hline\hline 
    \emph{Local} & Preserved interaction structure \newline for LHV dynamics~$\Exp\ni g \mapsto T_g \in \mathrm{End}(\Lambda)$ \newline \& quantum dynamics~$\Exp \ni g \mapsto \Phi_g \in \mathrm{C}(\Hil)$ & See Sec.~\ref{sec: local LHV dynamics} \\
\hline 
    \emph{Continuous} & $\Lambda\ \&\ \Exp$ manifolds \newline Continuous hidden-variable dynamics \newline~$\Exp\times \Lambda \ni (g, \lambda) \mapsto T_g(\lambda)\in \mathbb{P}(\Lambda)$ \newline Continuous quantum dynamics~$g\mapsto \Phi_g$ & Continuous transformations \newline~$T(\Exp)\subseteq \mathfrak{C}(\Lambda, \mathbb{P}(\Lambda))$ \newline Concept of dimension~$\dim \Lambda < \infty$ \\
\hline
    \emph{Deterministic} & Deterministic transformations \newline~$\Lambda \ni \lambda \mapsto T_g(\lambda) \in \Lambda$ & $T$ is monoid action~$\Exp \curvearrowright \Lambda$ \newline $\Exp \times \Lambda \ni (g, \lambda) \mapsto T_g(\lambda) \in \Lambda$ \newline Continuous dynamics~$\Rightarrow$ continuous action \\
\hline
    \emph{Smooth deterministic} & Smooth transformations \newline~$\Exp \times \Lambda \ni (g, \lambda) \mapsto T_g(\lambda) \in \Lambda$ & Smooth monoid action~$\Exp \curvearrowright \Lambda$ \\
\hline
    \emph{Reversible} & Reversible processes: \newline Inverses~$g^{-1} \in \Exp$ for all~$g\in\Exp$ & Deterministic,~$\Exp$ is a group \newline~$T$ is group action~$\Exp\curvearrowright \Lambda$ \newline 
    Invertible channels~$\Phi(\Exp)\subseteq \mathrm{PU}(\Hil)$  \newline Continuous dynamics~$\Rightarrow$~$\Exp$ is Lie group \\
\hline 
    \emph{Microscopic} & Physical processes \emph{are} quantum channels: \newline Faithful~$\Phi: \Exp \to \mathrm{C}(\Hil)$ & If~$\Phi(\Exp) \subseteq \mathrm{PU}(\Hil)$ subgroup \newline~$\Rightarrow$ reversible \& deterministic \\
\hline 
    \emph{Consistent gauge} & Path independence~$(T_g)_\# \mu_\rho = \mu_{\Phi_g[\rho]}$ & \\
\hline 
    \emph{Free gauge} & Correct dynamics for \emph{any} initially valid hidden-variable distribution & See Sec.~\ref{sec: no free gauge} \\
\hline\hline
\end{tabular}
\end{table*}

Note that, in contrast to the measurement rule, it is not obvious how to reduce general stochastic LHV dynamics to deterministic dynamics. For a single stochastic map it is possible (see App.~\ref{app: stochastic_maps}) and requires extending the hidden variables by one dimension. This shifts the randomness from the dynamics into the initial distribution of the hidden variables. Now, we have one stochastic map for each process~$g\in\Exp$. So we would need a corresponding deterministic map~$T_g:\Lambda \to \Lambda$ for each~$g\in\Exp$ and they must still compose correctly. 
Reusing the same initial source of randomness (the \lq seed\rq) at each step can lead to incorrect measurement statistics. An alternative is that each particle carries an infinite sequence of independent seeds, and each process uses and discards the first unused seed. However, algebraic relations between different processes can break the composition property since the number of consumed seeds can become inconsistent.
We do not explore this further here, and work with stochastic dynamics unless stated otherwise.

\paragraph{Smooth LHV Dynamics.} For deterministic dynamics one may require not only continuity but also differentiability in a straightforward sense. This allows us to derive classical time-evolution equations in the next section.
\begin{definition}
    A dynamical LHV model has \textbf{\emph{smooth} deterministic LHV dynamics} if the dynamics is deterministic,~$\Lambda_1$ and~$\Exp$ are smooth manifolds, and the mapping~$\Exp\times \Lambda\ni (g, \lambda) \mapsto T_g(\lambda)\in\Lambda$ is jointly smooth.
\end{definition}

\paragraph{Reversible LHV Dynamics.} 
Standard quantum mechanics, restricting to unitary processes, as well as classical Hamiltonian mechanics suggest that on a fundamental, microscopic level physical processes are reversible. For example, for a process described by a particular sequence of gates~$U_1,\ldots, U_K$ on a quantum computer, one can reverse every gate and implement the inverse process using the sequence~$U_K^\dagger,\ldots, U_1^\dagger$. 
\begin{definition}
    A dynamical LHV model has \textbf{\emph{reversible} LHV dynamics} if the physical processes~$g\in\Exp$ are reversible, i.e., if they have inverses~$g^{-1}\in\Exp$.
\end{definition}
It could be that for a particular process it is not practically feasible to implement its reverse process. For our notion of reversible LHV dynamics it is sufficient if the reverse process exists in principle.

In this case~$\Exp$ becomes a group and the monoid homomorphisms~$\Phi:\Exp \to \mathrm{C}(\Hil)$ and~$T:\Exp \to \mathrm{End}(\Lambda)$ become group homomorphisms. 
This forces~$\Phi(\Exp)$ and~$T(\Exp)$ to be groups. The only invertible quantum channels are given by unitaries,~$\Phi(\Exp)\subseteq\mathrm{PU}(\Hil)$~\cite{Nayak_Sen_2007}. The only invertible stochastic maps are those given by measurable bijections~$T_g:\Lambda \to \Lambda$ (see Prop.~\ref{prop:invertible_stochastic_map_is_deterministic} in App.~\ref{app: stochastic_maps}). Therefore, reversible LHV dynamics is \emph{automatically} deterministic. Moreover,~$T$ becomes a group action of the group of physical processes~$\Exp$ on the hidden-variable space~$\Lambda$. 
If the dynamics is additionally continuous, then so is this group action. In that case,~$\Exp$ is in fact a Lie group (see Thm.~\ref{thm: lie group from monoid} in App.~\ref{app: lie group from monoid}).

\paragraph{Microscopic LHV Dynamics.} 
In some situations it may be possible to identify the set of physical processes with the quantum dynamics itself. The physical meaning would be that the quantum description of the dynamics, say a particular unitary, is already equivalent to the most microscopic description possible: If two processes are described by the same quantum channel, they \emph{are} the same process.
\begin{definition}
    A dynamical LHV model has \textbf{\emph{microscopic} LHV dynamics} if the physical processes can be identified with the quantum channels in question,~$\Exp \cong \Phi(\Exp) \subseteq \mathrm{C}(\Hil)$, i.e., if~$\Phi:\Exp \to \mathrm{C}(\Hil)$ is faithful.
\end{definition}
If the considered channels form a subgroup~$\channels = \Phi(\Exp) \subseteq \mathrm{PU}(\Hil)$, then microscopic LHV dynamics is automatically reversible and hence also deterministic.

\paragraph{LHV Dynamics with Consistent Gauge.} 
The gauge, i.e., a particular choice of valid hidden-variable distributions to describe given measurement statistics, is by definition physically irrelevant. However, some models may have the property that the gauge can be chosen consistently across the dynamics. One may think of this as path independence in the sense that if the same state is reached in two different ways, then the resulting hidden-variable distributions must also match.
\begin{definition}
    A dynamical LHV model has \textbf{LHV dynamics \emph{with consistent gauge}} if the hidden-variable distributions are fixed uniquely for each state. That is, if for any state~$\rho \in\States$ and any process~$g\in\Exp$ we have~$(T_g)_\#\mu_{\rho}=\mu_{\Phi_g[\rho]}$.
\end{definition}
Note that we did \emph{not} require this in the general definition for LHV dynamics. There, it is only implied that~$(T_g)_\#\mu_\rho$ is \emph{a} valid hidden-variable distribution for~$\Phi_g[\rho]$. A consistent gauge is difficult to motivate physically but can simplify the mathematical analysis.

\paragraph{LHV Dynamics with Free Gauge.} 
Instead of restricting the choice of gauge in the above sense, one may also consider the opposite extreme, and argue that the gauge should not be fixed at all. Specifically, our definition for a dynamical LHV model requires a \emph{particular} choice of valid hidden-variable distributions~$\mu_\rho$ for all states~$\rho\in\States$ such that the LHV dynamics reproduces the correct change of measurement statistics. Now suppose a given state~$\rho\in \States$ can be decomposed as
\begin{align}
    \rho = t \rho_1 + (1-t) \rho_2
\end{align}
for~$t\in [0, 1]$ and states~$\rho_1, \rho_2$. This means the state~$\rho$ can be prepared, e.g., by preparing~$\rho_1$ with probability~$t$ and~$\rho_2$ with probability~$(1-t)$. Assume that~$\mu_{\rho_1}, \mu_{\rho_2}\in\mathbb{P}(\Lambda)$ are valid hidden-variable distributions for~$\rho_1, \rho_2$, respectively. Since the measurement statistics are linear it follows that
\begin{align}
    \mu'_\rho \equiv t \mu_{\rho_1} + (1-t) \mu_{\rho_2} \in \mathbb{P}(\Lambda)
\end{align}
is a valid hidden-variable distribution for~$\rho$. However, in general~$\mu'_\rho \neq \mu_\rho$. We argued earlier that the LHV dynamics should be a process on the individual hidden-variable level, i.e., the dynamics is given by maps~$T_g: \Lambda \to \mathbb{P}(\Lambda)$, with no dependence on the global quantum state~$\rho$ or the full distribution~$\mu_\rho$. Therefore it should actually not matter for the dynamics whether~$\lambda\in \Lambda$ is sampled according to~$\mu_\rho$ or~$\mu'_\rho$.
\begin{definition}
    A dynamical LHV model has \textbf{LHV dynamics \emph{with free gauge}} if for all states~$\rho\in\States$, measurement events~$m\in\mathbf{M}$ and processes~$g\in\Exp$ it holds that
    \begin{align}
        (T_g)_\# \mu'_\rho(L_m) = \mathbf{P}(m|\Phi_g[\rho])
    \end{align}
    for \emph{any} distribution~$\mu'_\rho\in\mathbb{P}(\Lambda)$ that is a valid hidden-variable distribution for~$\rho$, i.e., that satisfies
    \begin{align}
        \mu'_\rho(L_m) = \mathbf{P}(m|\rho) \quad \forall m \in \mathbf{M}.
    \end{align}
\end{definition}

\section{LHV Dynamics for a Hamiltonian}
\label{sec: lhv dynamics for a hamiltonian}

An important special case is dynamics generated by a fixed Hamiltonian~$H$. In this case, a process or experiment~$g\in\Exp$ is simply given by letting the system time-evolve according to its internal interactions for some time interval~$t\ge 0$. That is, we have the monoid of processes~$\Exp = (\R_{\ge 0}, +)$. The quantum dynamics is unitary with~$\R_{\ge 0} \ni t \mapsto \Phi_t \cong U_t = e^{-iHt} \in \mathrm{PU}(\Hil)$. Therefore, LHV dynamics corresponds to a monoid homomorphism~$\R_{\ge 0} \ni t \mapsto T_t \in \mathrm{End}(\Lambda)$ such that
\begin{align}
    \mathbf{P}(m|\rho(t)) = (T_t)_\# \mu_\rho(L_m),\quad \rho(t) = e^{-iHt}\rho e^{iHt}
\end{align}
for all times~$t\ge 0$, initial states~$\rho=\rho(0)\in\States$ and measurement events~$m\in\mathbf{M}$.
The homomorphism property says that there are no instantaneous changes on the hidden-variable level,~$T_0 = \mathrm{Id}$, and evolving for a time interval of length~$t_1$ and afterwards for a time interval of length~$t_2$ is the same as evolving for a time interval of length~$t_2+t_1$, in other words,~$T_{t_2}\circ T_{t_1} = T_{t_2+t_1}$.

Now, we would like to have a time-evolution equation for~$\mu_{\rho, t} \equiv (T_t)_\# \mu_\rho$. We expect that requiring continuous LHV dynamics would already suffice for this, given that one works with a suitably weak notion of differentiability. Here, we do not explore this general setting further and instead assume smooth deterministic LHV dynamics. As discussed in the previous section, the dynamics is, for example, automatically deterministic if it is reversible, i.e., if one allows for backwards time evolution,~$\Exp = (\R, +)$. For smooth deterministic LHV dynamics we can derive a time-evolution equation for the hidden variables and make sense of~\lq $\lambda(t)$\rq.

Indeed, in this situation LHV dynamics is given by the flow of a time-independent smooth velocity field: We can define the time evolution of hidden variables,~$\lambda(t) \equiv T_t(\lambda)$ for~$\lambda = \lambda(0)\in\Lambda$. This time evolution can be described by a simple first-order differential equation: Since the mapping~$(t, \lambda) \mapsto T_t(\lambda) = \lambda(t)$ is jointly smooth we can differentiate with respect to time
\begin{align}
    \notag
    \frac{\dup}{\dup t} \lambda(t) 
    &=\frac{\dup}{\dup t} T_t(\lambda)
    = \frac{\dup}{\dup s}\Big|_{s=0} T_{s+t}(\lambda) \\
    &= \frac{\dup}{\dup s}\Big|_{s=0} T_s(T_t(\lambda))
    = \frac{\dup}{\dup s}\Big|_{s=0} T_s (\lambda(t)).
\end{align}
Thus, we obtain a time-independent smooth velocity field~$V$ on~$\Lambda$,
\begin{align}
    V(\lambda)\equiv \frac{\dup}{\dup s}\Big|_{s=0} T_s (\lambda),
\end{align}
for which the trajectories~$\lambda(t)$ satisfy a classical time-evolution equation
\begin{align}
    \frac{\dup}{\dup t} \lambda(t) = V(\lambda(t)),\quad \lambda(0) = \lambda.
\end{align}
See Fig.~\ref{fig: 5} for a visualization.

We can also derive a time-evolution equation for the time-evolved hidden-variable distributions,~$\mu_{\rho, t}$. For simplicity, we consider the special case where the measures~$\mu_\rho$ can be written as densities with respect to a common smooth positive background measure,~$\dup \mu_\rho(\lambda) = p_\rho(\lambda) \dup \lambda$ with~$p_\rho:\Lambda \to [0, \infty)$. Then the same holds for their time evolution,~$\dup \mu_{\rho, t}(\lambda) = p_{\rho, t}(\lambda)\dup \lambda$, and~$p_{\rho, t}(\lambda)$ solves the continuity equation for the velocity field~$V$\footnote{If the densities~$p_\rho$ are not differentiable, the continuity equation holds in a weak sense.}
\begin{align}
    \partial_t p_{\rho, t}(\lambda) + \mathrm{div}_\lambda(p_{\rho, t}(\lambda)V(\lambda)) = 0.
\end{align}
We can recover the hidden-variable dynamics~$T_t$ as the flow of the velocity field~$V$. In particular, a fixed dynamical gauge, given by distributions~$p_{\rho, t}(\lambda)$ valid for the time-evolved states~$\rho(t)$, is compatible with smooth deterministic LHV dynamics if and only if there exists a smooth state-independent velocity field, with globally defined flow, that solves the continuity equation above simultaneously for all states~$\rho\in\States$.

In this infinitesimal setting we argue for a stronger notion of local LHV dynamics. Consider, for example, a spin chain with nearest-neighbor interactions only. Then for time evolution over a time interval of finite length~$t > 0$ the resulting unitary~$U_t$ can in general couple non-nearest-neighbor spins. Hence, the same is allowed for~$T_t$. Thus, with our existing notion of local dynamics the velocity field~$V$ would formally not be constrained to only include nearest-neighbor interactions, although that would seem physically plausible. One may make this an additional assumption. We do not make a formal definition since such a property, though reasonable, is neither used nor required for our results.

Finally, there is the case of time-dependent Hamiltonians~$H(t)$. We will not consider it in the following and discuss it here only for completeness. 
The main observation is that in this case the time-evolution operators~$U_{t_2, t_1}$ are indexed by two times and cannot be arbitrarily composed:~$U_{t_4, t_3} U_{t_2, t_1}$ only describes a valid process if~$t_3=t_2$. This means the time-evolution operators do not form a monoid. From an experimental point of view a time-dependent Hamiltonian would correspond to pre-committing to a particular drive. In practice, one has multiple different options for the drive. Therefore, a more natural formulation in terms of our framework is that the set of physical processes does not correspond to time but to the possible drives in this case. These driving processes \emph{can} be composed arbitrarily and they do form a monoid. Effectively, this just means requiring the quantum dynamics to be a submonoid of all unitaries,~$\Phi(\Exp) \subseteq \mathrm{PU}(\Hil)$.

\begin{figure}[t]
    \centering
    \includegraphics[width=0.40\textwidth]{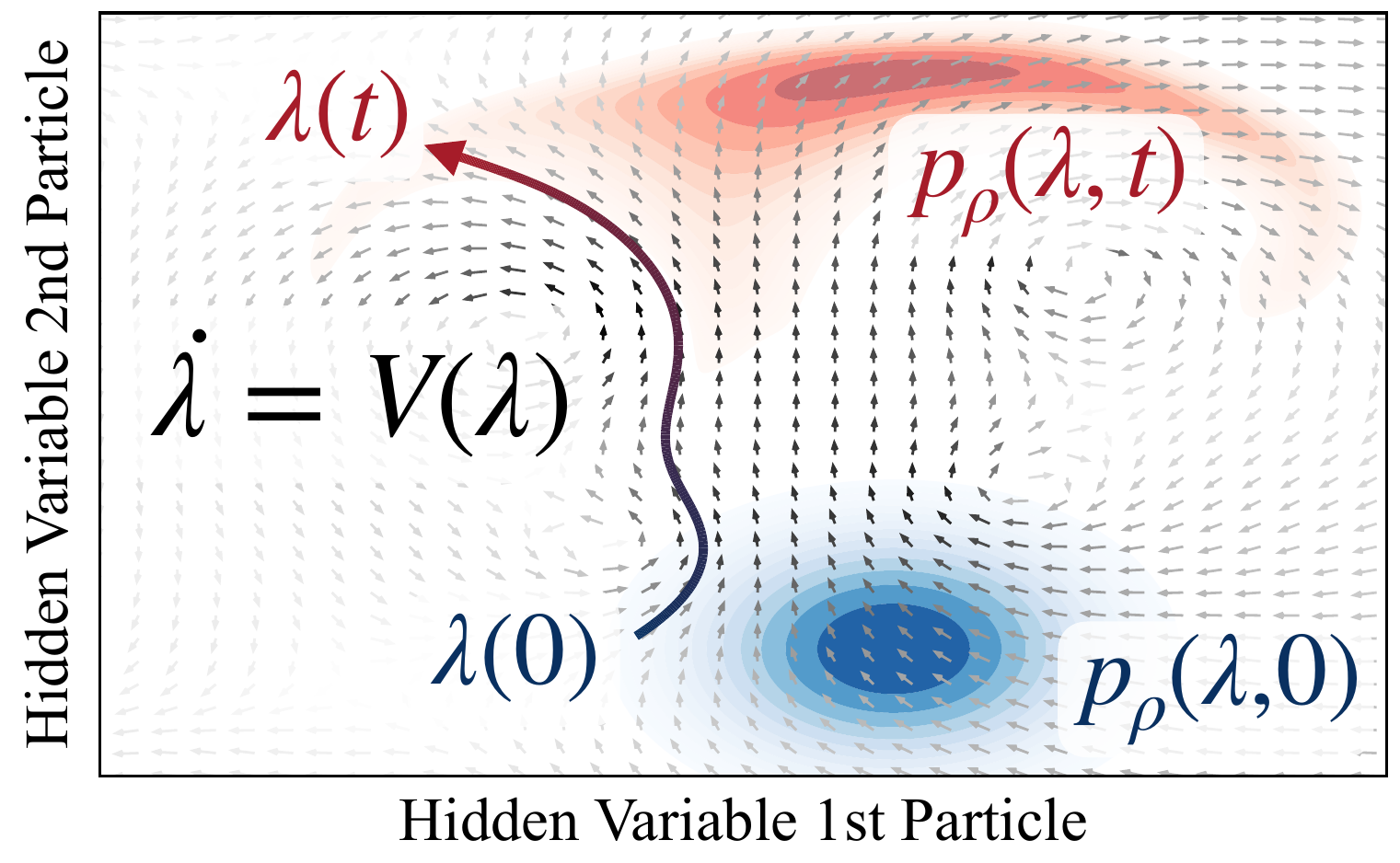}
    \caption{ 
    \textbf{LHV dynamics for a Hamiltonian.} LHV dynamics for a fixed time-independent Hamiltonian~$H$ corresponds to a time- and state-independent velocity field~$V$ on the hidden-variable space. A particular instance of the hidden variables~$\lambda$ evolves according to this velocity field. For every initial quantum state~$\rho$ this leads to time evolution of the hidden-variable distribution~$p_\rho$ according to the continuity equation.}
    \label{fig: 5}
\end{figure}

Alternatively, one could relax the assumption that the set of processes~$\Exp$ must be a monoid and allow that not all processes are composable. This would correspond to the mathematical structure of a \lq small category\rq. In this setting physical processes can be indexed by two times~$\Exp = \{(t_2, t_1) | t_1, t_2 \in \R_{\ge 0}, t_2 \ge t_1\}$ with composition~$(t_4, t_3)\circ (t_2, t_1) = (t_4, t_1)$ only defined for~$t_3 = t_2$. Then the quantum dynamics is given by the time-evolution operators~$\Phi_{t_2, t_1} = U_{t_2, t_1}\in \mathrm{PU}(\Hil)$ (the time-ordered exponential of the time-dependent Hamiltonian~$H(t)$). Assuming smooth deterministic dynamics, one obtains a time-evolution equation for the hidden variables with a time-dependent velocity field
\begin{align}
    \frac{\dup}{\dup t} \lambda(t) = V(\lambda(t), t), \quad V(\lambda, t)\equiv \frac{\dup}{\dup s}\Big|_{s=0} T_{t+s, t} (\lambda).
\end{align}
In this case~$V(\lambda, t)$ should again have the same interaction structure as~$H(t)$ for all times. We suggest \lq LHV dynamics for a time-dependent Hamiltonian~$H(t)$\rq\ to refer to this second construction and otherwise name the submonoid of quantum channels in question explicitly. 

In the following sections we explore in which situations dynamical LHV models do or do not exist. This, of course, depends on the precise setting, that is, the states, measurements, and channels, and on which combination of the additional attributes introduced above one requires.

\section{Noninteracting Time Evolution}
\label{sec: noninteracting time evolution}

We start with the special case of noninteracting time evolution. By \lq noninteracting\rq\ we mean dynamics described by products of single-particle unitaries,~$U\in\mathrm{PU}(\Hil_1)^N$. In particular, this includes arbitrary single-particle dynamics. We demonstrate that there exist \lq nice\rq\ dynamical LHV models in this noninteracting case, featuring many of the additional attributes that we introduced. This demonstrates that \lq non-existence of LHV dynamics\rq\ can only be a property of interacting quantum dynamics. The result also reveals that many of the additional assumptions on LHV dynamics that we have introduced above, including those leading to no-go theorems discussed in later sections, are not fundamentally contradictory from the start. 

The idea is to switch between Schrödinger and Heisenberg pictures: One can express the quantum measurement statistics of a state that has been transformed by a noninteracting unitary as the measurement statistics of the untransformed state but with respect to rotated measurements. Formally, unitaries~$U\in\mathrm{PU}(\Hil)$ act on measurements~$x=(M^a_x)_{a\in\outcomes}\in\measurements$ by conjugation~$U[x] = (U M^a_x U^\dagger)_{a\in\outcomes} = (U [M^a_x])_{a\in\outcomes}$. Then
\begin{align} \notag
    \mathbf{P}(a|x, U[\rho]) &= \Tr(U\rho U^\dagger M^a_x) = \Tr(\rho U^\dagger M^a_x U) \\
    &= \mathbf{P}(a|U^\dagger[x], \rho).
\end{align}
Restricting to noninteracting unitaries means that these rotated local measurements are still valid local measurements. We assume that these correlations can be captured by an LHV model. To obtain valid LHV dynamics we must then be able to replace the rotated measurements by transformed hidden variables. For noninteracting unitaries this should happen in a local manner, i.e., the hidden variables of the different particles must transform independently. This is a requirement on the local measurement rule of the form~$q(a|U^\dagger[x], \lambda) = q(a|x, T_U(\lambda))$ for some transformation maps~$T_U:\Lambda_1 \to \Lambda_1$ depending on single-particle unitaries~$U\in\mathrm{PU}(\Hil_1)$. 

Combining these two steps, moving the dynamics from the state onto the measurement settings and then from the measurement settings onto the hidden variables, one obtains LHV dynamics that transforms the measurement statistics in the correct way. This LHV dynamics is given by measurable maps and it acts on each particle independently. This means the dynamics is deterministic and local by construction. It also has free gauge because the argument only uses that the quantum measurement statistics are correctly reproduced for the untransformed states, hence any choice of valid hidden-variable distributions works. 

This idea leads to the following proposition, which specifies a sufficient condition for local deterministic LHV dynamics to exist for noninteracting unitaries.
\begin{proposition}
    \label{prop: local_unitaries}
    Let~$\mathfrak{L}=(\Lambda_1, L, \mu)$ be an LHV model for a set of states~$\States$ and local measurements~$\measurements=\measurements_1^N$, such that~$\States$ and~$\measurements$ are invariant under the action of noninteracting unitaries:~$\mathrm{PU}(\Hil_1)^N[\States] = \States$ and~$\mathrm{PU}(\Hil_1)^N[\measurements] = \measurements$. Suppose there exist measurable maps~$T_U: \Lambda_1 \to \Lambda_1$ for every single-particle unitary~$U\in\mathrm{PU}(\Hil_1)$ such that~$L^a_{U^\dagger[x]} = T^{-1}_U(L^a_x)$ for~$a \in \outcomes_1$ and~$x\in\measurements_1$, and such that~$U\mapsto T_U$ is a group homomorphism. Then, the LHV model~$\mathfrak{L}$ has local deterministic microscopic (hence reversible) LHV dynamics with free gauge for all noninteracting unitaries~$\mathrm{PU}(\Hil_1)^N$. If~$\Lambda_1$ is a finite-dimensional (smooth) manifold and the mapping~$\mathrm{PU}(\Hil_1) \times \Lambda_1 \ni (U, \lambda) \mapsto T_U(\lambda) \in \Lambda_1$ is continuous (smooth) then so is the LHV dynamics.
\end{proposition}
\begin{proof}
     See App.~\ref{app: action of noninteracting unitaries}.
\end{proof}

There are LHV models that satisfy this sufficient condition. Hence local deterministic microscopic LHV dynamics exists for noninteracting unitaries. The simplest example is to consider separable states with their separable decomposition as an LHV model, single-particle states serving as the hidden variables, and the local measurement rule given by the single-particle Born rule. Born's rule has the required property that a transformed measurement is equivalent to a transformed state.
\begin{theorem} [Dynamical LHV models for noninteracting unitaries]    
    \label{thm: local_unitaries}
    There exist dynamical LHV models for all local projective measurements~$\measurements$, all separable states~$\States\subseteq\mathcal{L}_\measurements$, and all noninteracting unitaries~$\mathrm{PU}(\Hil_1)^N$ with smooth deterministic, reversible, microscopic and local LHV dynamics with free gauge. Moreover, the same base model works for any number of particles. 
\end{theorem}
\begin{proof}
    See App.~\ref{app: action of noninteracting unitaries}.
\end{proof}

This result is not restricted to separable states. Also for the general LHV models from App.~\ref{app: general lhv models} it is possible to translate back and forth between transformed measurement settings and transformed hidden variables. See App.~\ref{app: action of noninteracting unitaries} for the details. These models are expected to be expressive enough to also represent some of the local but entangled states. Indeed, at least for few qubits we have demonstrated this numerically in~\cite{von_Selzam_Marquardt_2025} and we will provide an analytic proof in~\cite{von_Selzam_Marquardt_in_preparation} (in preparation).

\section{A Counterexample: Two Interacting Qubits}
\label{sec: a counterexample}

We have now seen that LHV dynamics exists for noninteracting dynamics, in particular, for single-particle systems. The next natural question is what happens for interacting systems. For this purpose we start by considering the simplest example of that kind: Two qubits coupled by an interacting Hamiltonian.

It turns out that having deterministic LHV dynamics for an interacting Hamiltonian is a very constraining condition. We take the perspective of fixing valid hidden-variable distributions~$p_{\rho_0, t}(\lambda)$ for the time-evolved states~$\rho(t)$, i.e., we fix a dynamical gauge. Then, as pointed out earlier, we can view the continuity equation as an equation for the velocity field~$V$. For a single state there are infinitely many solutions. However, in a particular example we show that there does not exist any \emph{state-independent} solution for the velocity field that works simultaneously for all initial states. This demonstrates that there are LHV models with valid time-dependent hidden-variable distributions which are not compatible with smooth deterministic LHV dynamics.

Specifically, we consider the simple Bell LHV base model for projective qubit (spin-$1/2$) measurements~\cite{Bell_1964}. In this model, the surface of a sphere acts as the single-particle hidden-variable space,~$\Lambda_1 = S^2$. The outcome of a spin measurement is determined by a projection onto the measurement direction,~$q(\uparrow|\hat n,\hat\lambda) = \Theta(\hat n \cdot \hat \lambda)$, where~$\hat n\in S^2$ describes the direction of the spin measurement and~$\Theta$ is the Heaviside step function. We have demonstrated in~\cite{von_Selzam_Marquardt_2025} that the Bell LHV base model is sufficiently expressive to describe all separable~$N$-qubit states and even many local entangled states, such as Werner states up to a visibility of~$v\approx 0.5$. 
We consider evolution under the Heisenberg Hamiltonian for a system of two spins,~$H=\frac{\omega}{4}\sum_{k=1}^3\sigma_k\otimes\sigma_k$, where~$\sigma_k$ are the Pauli matrices and~$\omega > 0$ is a frequency. 
\begin{theorem} 
    \label{thm: heisenberg_counterexample}
    For the Bell LHV base model and the Heisenberg Hamiltonian there is a valid dynamical gauge~$\{p_{\rho_0, t}(\hat\lambda_1,\hat\lambda_2) | \rho_0\in\States_v, t \in \R\}$ reproducing the measurement statistics for the time-evolved states~$\rho(t)$, for initial states~$\rho_0\in \States_v$ for a sufficiently small visibility~$v>0$, all times, and all local projective spin measurements. The hidden-variable distributions~$p_{\rho_0, t}(\hat\lambda_1,\hat\lambda_2)$ are continuous and almost everywhere smooth as a function of time~$t \in \R$ and the hidden variables~$\hat\lambda_1, \hat\lambda_2 \in S^2$. However, for every visibility~$0 < v' \le v$ this dynamical gauge is not compatible with smooth deterministic LHV dynamics for the Heisenberg Hamiltonian.
\end{theorem}
\begin{proof}
    The construction of the hidden-variable distributions~$p_{\rho_0, t}(\hat\lambda_1,\hat\lambda_2)$ is based on hidden-variable distributions that are valid for single-qubit states together with a separable decomposition of sufficiently noisy two-qubit states. The details are presented in App.~\ref{app: Bell LHV counterexample}. There we also demonstrate that no state-independent velocity field~$V$ solves the continuity equation simultaneously for all states~$\rho_0 \in \States$, where~$\States$ can be an arbitrarily small neighborhood of the maximally mixed state. Roughly speaking, reducing the visibility still preserves the structure of the full state space. This translates to many independent conditions on the velocity field and indeed over-determines it. Note that the Hamiltonian is time-independent. However, not even generalizing to a time-dependent velocity field works. In fact the continuity equation cannot even be satisfied at the initial time~$t=0$.
\end{proof}

We suspect that this behavior is not due to our specific choice of hidden-variable distributions, i.e., the choice of gauge. That is, we conjecture that there exists a set of local two-qubit states that remain local under time evolution with respect to the Heisenberg Hamiltonian, such that the Bell LHV base model can describe all the states at each moment in time, but does not have smooth deterministic LHV dynamics. The example above does not prove this conjecture because we assumed some particular evolution of the hidden-variable distributions~$p_{\rho_0, t}(\lambda)$, that is, we fixed a dynamical gauge. In principle, it is possible that there exist different valid hidden-variable distributions~$\Tilde{p}_{\rho_0, t}(\hat\lambda_1, \hat\lambda_2)$ that \emph{are} compatible with a single state- and time-independent velocity field.

\section{Nonlocal Cheat LHV Dynamics}
\label{sec: cheat dynamics}

In the previous section we demonstrated that LHV dynamics can fail to exist for a particular simple model. Here, we show that LHV dynamics essentially always exists \emph{if} one allows for a sufficiently complex underlying LHV model \emph{and} is willing to give up locality of the dynamics. 

The idea is to lift state-dependent dynamics, which generically exists, to formally state-independent dynamics on a larger hidden-variable space. This is achieved by adding the full~$N$-particle quantum state to the hidden variables,~$\tilde\Lambda_1 \equiv \Lambda_1 \times \States$. In this way each particle has access to the global state and can track the quantum time evolution. 
This is similar to what happens in Bohmian mechanics~\cite{Bohm_1952_1, Bohm_1952_2}. Unlike Bohmian mechanics we draw entirely new stochastic values of the hidden variables from a valid hidden-variable distribution of the time-evolved state, though. 
The same idea has also been used in the context of NMR computing~\cite{Menicucci_Caves_2002}. Theorem~\ref{thm: cheat LHV dynamics} below can be seen as a generalization from their separable setting to arbitrary LHV models and measurements. 
The LHV dynamics constructed in this way is not local: in general, it does not respect the underlying interaction structure of the quantum dynamics. 
\begin{theorem} [Cheat LHV dynamics]
    \label{thm: cheat LHV dynamics}
    Let~$\measurements=\measurements_1^N$ be a set of local measurements,~$\States\subseteq \mathcal{L}_\measurements$ a measurable subset of local states, and~$\channels\subseteq\mathrm{C}(\Hil)$ a submonoid of quantum channels that leaves the set of states invariant,~$\channels[\States] \subseteq \States$. Assume there exists an LHV model~$(\Lambda_1, L, \mu)$ for~$\States$ and~$\measurements$ such that~$\States \ni \rho \mapsto \mu_\rho(A)\in [0, 1]$ is measurable for every fixed measurable subset~$A\subseteq \Lambda$. Then~$\channels$ has LHV dynamics with consistent gauge. If~$\Lambda_1$ is a finite-dimensional manifold, and~$\States\subseteq \mathbb{D}(\Hil)$ and~$\channels\subseteq \mathrm{C}(\Hil)$ are submanifolds, and if~$\States \ni \rho \mapsto \mu_\rho \in \mathbb{P}(\Lambda)$ is continuous, then this dynamical LHV model has continuous LHV dynamics.
\end{theorem}
\begin{proof}
    We define the new hidden-variable space by appending the full set of quantum states,~$\tilde \Lambda_1 \equiv \Lambda_1 \times \States$. For~$\tilde\lambda = (\lambda, \sigma)\in\tilde \Lambda_1$ we call~$\sigma\in\States$ the state-component of the hidden variable. The new measurement rule ignores that state-component,~$\tilde L^a_x \equiv L^a_x \times \States$ for~$x\in \measurements_1$,~$a\in\outcomes_1$. The new hidden-variable distributions fix the state-component (of each particle) to equal the initial quantum state,~$\tilde \mu_\rho \equiv \mu_\rho \otimes \delta_\rho^{\otimes N} \in \mathbb{P}(\Tilde{\Lambda})$ where we identify~$\Tilde{\Lambda} = \tilde\Lambda_1^N \cong \Lambda \times \States^N$. 
    Now we can define the LHV dynamics as follows: Map the hidden variables to the hidden-variable distribution of the transformed state independent of their current value (this is a stochastic map), and update the state-component according to the quantum dynamics. Explicitly, for a channel~$\Psi\in\channels$ we set
    \begin{align} \notag
        T_\Psi(\tilde\lambda) 
        &= T_\Psi(\lambda_1,\ldots,\lambda_N, \sigma_1,\ldots, \sigma_N) \\
        &\equiv \mu_{\Psi[\sigma_1]}\otimes \delta_{\Psi[\sigma_1]}^{\otimes N} \in\mathbb{P}(\Lambda \times \States^N).
    \end{align}
    This dynamics is \emph{not} local because, in general,~$\mu_{\Psi[\sigma_1]}\in\mathbb{P}(\Lambda)$ is correlated. Thus, even if~$\Psi$ acts nontrivially only on a subset~$\pi\subseteq[N]$ of the particles,~$T_\Psi$ does not reduce to a stochastic map on~$\tilde\Lambda_\pi$.  
    In App.~\ref{app: cheat dynamics} we provide the remaining details, that is, we prove that the definitions above constitute a valid dynamical LHV model with consistent gauge that correctly reproduces all measurement statistics.
\end{proof}

The new single-particle hidden-variable space~$\tilde\Lambda_1=\Lambda_1 \times \States$ explicitly depends on the total number of particles. Assuming a full-dimensional set of states~$\States\subseteq \mathbb{D}(\Hil)$ the dimension of~$\tilde\Lambda_1$ grows exponentially in~$N$. This defeats the original motivation of simulating the quantum dynamics efficiently via the LHV dynamics.

\section{Dimensionality Constraint}
\label{sec: dimensionality constraint}

The cheat LHV dynamics construction from the previous section features a single-particle hidden-variable space that is exponentially large in the total number of particles. We now demonstrate that the same must hold for a large class of dynamical LHV models, namely, whenever the LHV dynamics is continuous and reversible. This can be reformulated as a no-go theorem excluding the possibility of continuous reversible LHV dynamics for all unitaries and any fixed LHV base model.

The proof relies on three ingredients: First, continuous reversible LHV dynamics corresponds to the continuous action of a Lie group of processes~$\Exp$ on the hidden-variable space~$\Lambda=\Lambda_1^N$. Second, a Lie group can only act faithfully on a finite-dimensional manifold of sufficiently large dimension~\cite{Montgomery_Zippin_1955}. For instance, the group~$SO(3)$ of rotations in three dimensions cannot be faithfully represented as some continuous action on the circle~$S^1$. And third, the set of processes~$\Exp$ must be at least as large as the quantum dynamics~$G\subseteq \mathrm{PU}(\Hil)$ it implements.
Taken together, it follows that~$\dim\Lambda = N \dim\Lambda_1$ must be large, the specific bound being related to the dimension of the considered subgroup of unitaries~$G\subseteq\mathrm{PU}(\Hil)$.

The precise mathematical basis for the dimensionality constraint and physical no-go result is the following theorem, which we prove in App.~\ref{app: dimensionality constraint}. It is a variation of a classical result in the theory of transformation groups (see the book by Montgomery and Zippin~\cite{Montgomery_Zippin_1955}).
\begin{theorem} \label{thm: dim-constraint-maintext}
    Let~$\Exp$ be a Lie group that acts continuously on a manifold~$\Lambda$ via~$T: \Exp \to \mathrm{Homeo}(\Lambda)$. Let~$G$ be a compact simple Lie group and assume there exists a surjective Lie group homomorphism~$\Phi: \Exp \to G$ with~$\ker T \subseteq \ker \Phi$. Then~$\dim G \le \dim \Lambda (\dim \Lambda + 1)/2$.
\end{theorem}
For the main case of interest, all projective measurements and a full-dimensional subset of states, the assumption~\lq$\ker T \subseteq \ker \Phi$\rq\ will be satisfied automatically. Its general meaning is discussed at the end of this section.

We can now derive the following dimensionality constraint on the hidden-variable space.
\begin{theorem}[Dimensionality constraint] 
    \label{thm: dim-constraint for no go}
    For the measurements~$\measurements$ consider all local projective measurements. Let~$G\subseteq \mathrm{PU}(\Hil)$ be a closed simple subgroup of the projective unitary group and let~$\States\subseteq \mathcal{L}_\measurements$ be a subset of local states which contains~$\States_v$ for some~$v>0$. Then, any LHV model for~$\States$, with hidden-variable space~$\Lambda$, that has continuous reversible LHV dynamics for~$G$ must adhere to the following dimensionality constraint:
    \begin{align}
        \dim G \le \frac{\dim \Lambda (\dim \Lambda + 1)}{2}.
    \end{align}
\end{theorem}
\begin{proof}
    Suppose we have an LHV model~$(\Lambda, L, \mu)$ for all local projective measurements for the states~$\States\supseteq \States_v$ and suppose this LHV model has continuous reversible LHV dynamics for~$G$. Then according to our definitions~$\Lambda_1$, and hence~$\Lambda$, are finite-dimensional manifolds. Moreover, as discussed in Sec.~\ref{sec: lhv dynamics: attributes} after the definition of reversible LHV dynamics, we have a Lie group of physical processes~$\Exp$, a continuous Lie group action~$T:\Exp\to\mathrm{Homeo}(\Lambda),\ g \mapsto T_g$, via (deterministic) homeomorphisms, and a continuous surjective homomorphism of Lie groups~$\Phi: \Exp \to G,\ g\mapsto \Phi_g$ such that
    \begin{align}
        \mathbf{P}(m|\Phi_g[\rho]) = (T_g)_\#\mu_\rho(L_m),
    \end{align}
    for all~$\rho\in\States$,~$m\in\mathbf{M}$ and~$g\in\Exp$.
    We apply the previous theorem. Note that~$G$ is a compact simple Lie group since it is a closed simple Lie subgroup of the compact simple Lie group~$\mathrm{PU}(\Hil)$. 
    It remains to show that~$\ker T \subseteq \ker \Phi$. The reason this must be true is that the set of all local projective measurements is tomographically complete (for every product state~$\ket{\psi}$ we have a measurement event~$m=(a|x)\in\mathbf{M}$ with~$M^a_x = \ketbra{\psi}$), i.e., the measurement statistics uniquely determine the state. Similarly, the states~$\States\subseteq \mathbb{B}(\Hil)$ are a sufficiently rich subset of operators such that the action of unitaries becomes faithful. In detail, let~$g\in\ker T \subseteq \Exp$. This means~$T_g = \mathrm{Id}_\Lambda$ so
    \begin{align}
        \mathbf{P}(m|\rho) &= \mu_\rho(L_m) = (T_g)_\#\mu_\rho(L_m)
        = \mathbf{P}(m|\Phi_g[\rho])
    \end{align}
    for all measurement events~$m\in\mathbf{M}$ and states~$\rho\in\States$. Tomographic completeness implies an equality of states,~$\rho = \Phi_g[\rho]$ for all~$\rho\in\States$. Since~$\States\supseteq \States_v$ for some visibility~$v>0$, and because~$\States_v$ spans the space of all Hermitian operators, this implies~$\Phi_g = \1$. Thus~$g\in\ker \Phi$ and therefore~$\ker T \subseteq \ker \Phi$.
\end{proof}

We can reformulate this as a no-go theorem excluding LHV dynamics for the full unitary group. 
\begin{corollary}[No continuous reversible LHV dynamics for all unitaries] \label{cor:physical_nogo}
      Fix any LHV base model~$(\Lambda_1, L)$ for local projective measurements. Then, there exists a critical number of particles~$N_0\in\N$ such that for all~$N\ge N_0$ the LHV base model does not have continuous reversible LHV dynamics for the projective unitary group~$\mathrm{PU}(\Hil)$ and states~$\States \supseteq\States_{v}$ for any~$v>0$.
\end{corollary}
\begin{proof}
    For the full projective unitary group,~$G = \mathrm{PU}(\Hil)$, a compact simple Lie group, the left-hand side of the dimensionality constraint from the previous theorem scales exponentially in~$N$ while the right-hand side only scales quadratically. Specifically, for~$N$ qudits of dimension~$D$, we have~$\dim \mathrm{PU}(\Hil) = D^{2N}-1$ and~$\dim\Lambda = \dim \Lambda_1^N = Nd$ for some fixed single-particle hidden-variable dimension~$d = \dim \Lambda_1 \in\N$. The dimensionality constraint becomes
    \begin{align}
        D^{2N}-1 \le \frac{Nd (Nd+1)}{2},
    \end{align}
    which is a contradiction for sufficiently large~$N$.
\end{proof}

\begin{figure}[t]
    \includegraphics[width=0.45\textwidth]{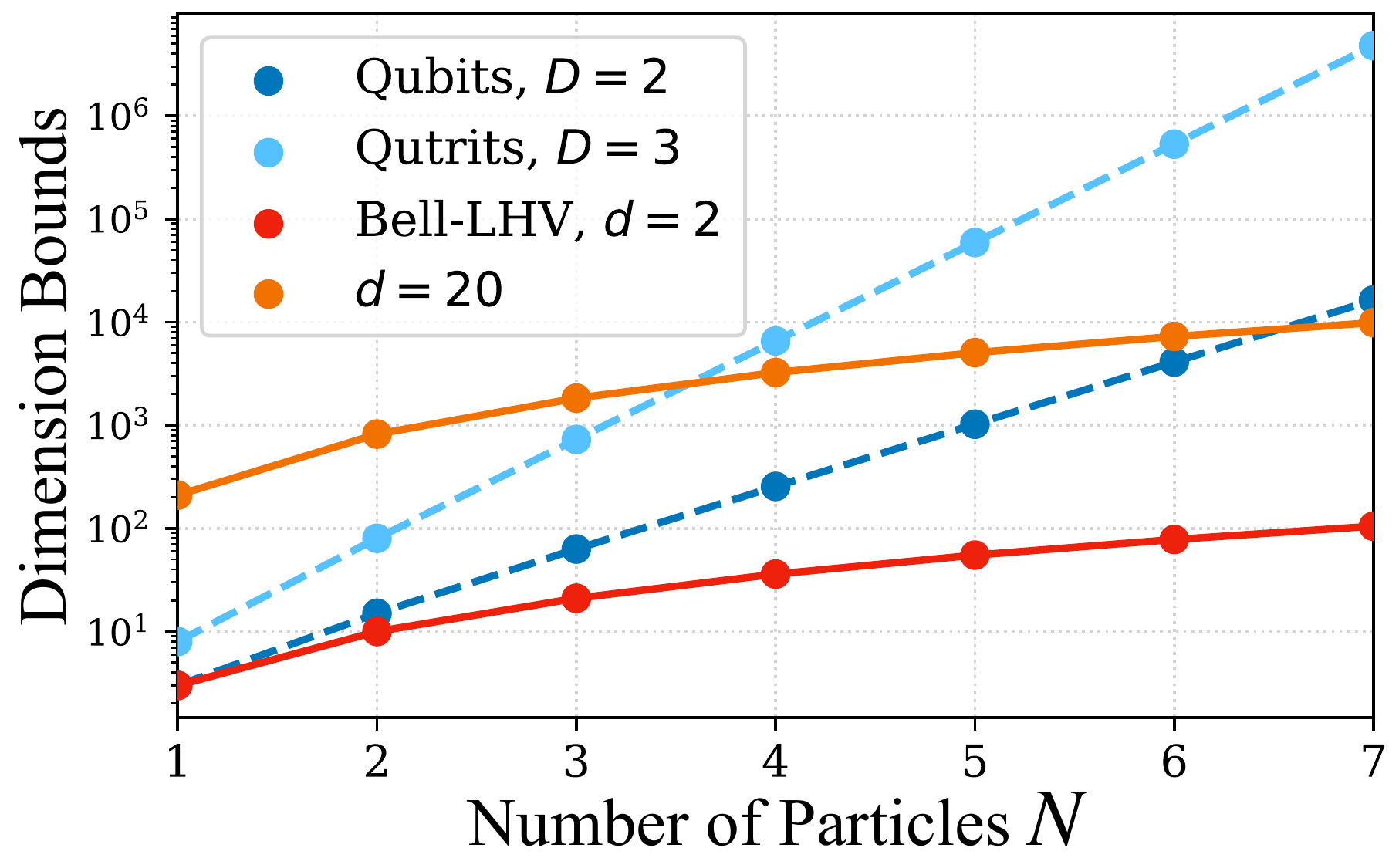}
    \caption{
    \textbf{Dimensionality constraint.} Continuous reversible hidden-variable dynamics requires~$B_{LHV}\ge B_{QM}$, where~$B_{QM} = D^{2N}-1$ for the unitary group of~$D$-dimensional qudits (blue, light blue) and~$B_{LHV} = Nd(Nd+1)/2$ for a~$d$-dimensional single-particle hidden-variable space~$\Lambda_1$ (red, orange).  We observe that the Bell LHV base model ($\Lambda_1 = S^2$) can only accommodate dynamics for a single qubit. A much more expressive LHV base model with~$d=20$ could allow for dynamics for up to~$6$ qubits or~$3$ qutrits, although there is no guarantee that such dynamics actually exists even in these cases.}
    \label{fig: 6}
\end{figure}

Roughly speaking, continuous reversible LHV dynamics requires some of the algebraic structure of unitaries to be represented on the hidden-variable space, and the full unitary group for sufficiently many particles is too big for this to be possible. LHV dynamics under the conditions stated in the theorem above, even restricted to states arbitrarily close to the maximally mixed state, is not possible for arbitrary particle numbers. Clearly, this no-go result extends to closed simple subgroups of the full projective unitary group, as long as their dimension grows strictly faster than quadratically in the number of particles. For non-simple subgroups~$G\subseteq \mathrm{PU}(\Hil)$ the largest closed simple subgroup contained in~$G$ enters the constraint. 
There is no contradiction with our result for noninteracting unitaries since the dimension of that group only scales linearly in the number of particles,~$\dim \mathrm{PU}(\Hil_1)^N = N(D^2-1)$. 

It is important to note that, at first glance, the no-go result is not obvious. After all, the space of all homeomorphisms on the hidden-variable space~$\Lambda$ is infinite-dimensional. Therefore, a priori it is not clear why one would not be able to map a finite-dimensional space of physical processes faithfully into that space of homeomorphisms for any finite number of particles. It is the additional group-action structure, implied by the assumption of reversible LHV dynamics, that leads to this drastic constraint. 

We emphasize that the LHV dynamics was \emph{not} assumed to be local in deriving the theorem. This means the no-go theorem even excludes LHV dynamics with nonlocal all-to-all interactions. Note that such \lq nonlocal dynamics\rq\ still does not allow for signaling or Bell inequality violations since the measurement statistics at each instant in time match quantum mechanics (and all states in question are local by assumption). 

The contradiction arises for sufficiently many particles. Due to the exponential scaling of~$\dim \mathrm{PU}(\Hil)$, the critical~$N$ is not large, growing approximately logarithmically in the single-particle hidden-variable dimension~$d$. To see this in practice, consider qubits,~$D=2$, and the simple Bell LHV base model with~$d=\dim S^2 = 2$ (cf.~Sec.~\ref{sec: a counterexample}). This LHV base model is sufficiently expressive to describe at least all separable states for arbitrary particle numbers. However, already for two particles it cannot accommodate continuous reversible hidden-variable dynamics anymore. Indeed, the minimal hidden-variable dimension required for two qubits would be~$\dim \Lambda_1 = 3$. See Fig.~\ref{fig: 6} for a visualization of the dimensionality constraint. 

We emphasize that our result is only a necessary condition for continuous reversible LHV dynamics but by no means guaranteed to be sufficient. That is, for any fixed number of particles a sufficiently high-dimensional single-particle hidden-variable space does not guarantee the existence of continuous reversible LHV dynamics. Note that the cheat LHV dynamics construction from Theorem~\ref{thm: cheat LHV dynamics} is not reversible.

One may ask how to circumvent this no-go result. One way is to drop one of the assumptions of continuous or reversible LHV dynamics. In both cases the mathematical machinery we used does not apply anymore. We do not know whether continuous but possibly non-reversible LHV dynamics can evade the dimensionality constraint. However, non-continuous LHV dynamics certainly can evade it, see Sec.~\ref{sec: non-continuous lhv dynamics}.
The alternative is to allow a single-particle hidden-variable space~$\Lambda_1$ that depends on the total number of particles. If the dimension of~$\Lambda_1(N)$ grows exponentially in~$N$ there is no contradiction. Of course, this would be strange from a physics perspective because in that case the single-particle physics must change every time one more particle is added to the system. To prevent this one must allow an infinite-dimensional single-particle hidden-variable space from the start, in which case there is obviously no dimensionality conflict and~$\Lambda_1$ can remain independent of~$N$. Whether an infinite-dimensional hidden-variable space would allow for continuous reversible LHV dynamics to exist is an open question.

Finally, we discuss the kernel assumption,~$\ker T \subseteq \ker \Phi$. For the no-go theorem we actually did not need to assume this: instead, we assumed that the considered set of states contains a full-dimensional ball around the maximally mixed state, i.e.,~$\States \supseteq \States_v$ for some visibility~$v>0$, and we considered all local projective measurements. These two assumptions are sufficient to \emph{derive} the condition~$\ker T \subseteq \ker \Phi$. 
One can relax this assumption and consider more restricted sets of local measurements~$\measurements_1$ and local states~$\States$. From a physics perspective one may then simply demand the inclusion of kernels as an additional assumption. The meaning is the following: If a physical process~$g\in\Exp$ leaves the hidden variables invariant~$T_g = \mathrm{Id}_\Lambda$, then the channel describing that process must be the identity~$\Phi_g = \1$. This is indeed expected to hold. After all, the unitaries should be the more coarse-grained point of view: It may reasonably happen that a process~$g$ is described by the identity unitary~$\Phi_g = \1$ while the hidden variables change in a nontrivial way,~$T_g \neq \mathrm{Id}_\Lambda$, but not vice versa. 

We remark that it is possible to replace the algebraic assumption of reversible LHV dynamics with topological assumptions on the process monoid~$\Exp$. For example, any compact connected topological monoid that has an Euclidean neighborhood of the identity is a Lie group~\cite[Cor.~1.3]{Mostert_Shields_1959}. It follows that LHV dynamics for a process monoid with these properties is automatically reversible, and thus the dimensionality constraint and no-go theorem apply.

\section{Non-Continuous LHV Dynamics}
\label{sec: non-continuous lhv dynamics}

We observe the following: Our construction of continuous cheat LHV dynamics (Thm.~\ref{thm: cheat LHV dynamics}) has a single-particle hidden-variable space~$\Lambda_1$ that has exponentially large dimension in the total number of particles~$N$. Conversely, it follows from Thm.~\ref{thm: dim-constraint for no go} that for continuous reversible LHV dynamics the single-particle hidden-variable space \emph{must} have exponentially large dimension in~$N$.

Here we demonstrate that the dimension of~$\Lambda_1$ can \emph{only} play a role for continuous LHV dynamics. Dropping this assumption means we may assume, without loss of generality, that~$\Lambda_1 = (0, 1)$. The reason is that as a standard Borel space any hidden-variable space can be mapped onto~$(0, 1)$ in a one-to-one and measurable manner. 
\begin{theorem} [Simple hidden-variable space]
    \label{thm: simple hidden-variable space}
    Let~$\mathfrak{L}=(\Lambda_1, L, \mu, \Exp, T, \Phi)$ be a dynamical LHV model for measurements~$\measurements$, states~$\States$ and channels~$\channels$. Then there is a dynamical LHV model~$\tilde{\mathfrak{L}} = (\tilde\Lambda_1, \tilde L, \tilde\mu, \Exp, \tilde T, \Phi)$ for~$\measurements, \States, \channels$ with~$\tilde\Lambda_1 = (0, 1)$. Moreover,~$\mathfrak{L}$ has local, deterministic, reversible or microscopic dynamics, or dynamics with consistent or free gauge, if and only if~$\tilde{\mathfrak{L}}$ does.
\end{theorem} 
\begin{proof}
    The main point is that~$\Lambda_1$ is a standard Borel space, and standard Borel spaces are classified only by their cardinality. That is, assuming~$\Lambda_1$ is uncountable, there exists a Borel isomorphism~\cite{Kuratowski_1934}\cite[Thm.~15.6]{Kechris_1995}
    \begin{align}
        J: \Lambda_1 \to (0, 1)\equiv \tilde\Lambda_1.
    \end{align}
    This means~$J$ is Borel measurable with Borel measurable inverse~$J^{-1}: (0, 1) \to \Lambda_1$. Then the extension of~$J$ to all~$N$ particles,~$\boldsymbol{J}\equiv J^{\times N}: \Lambda=\Lambda_1^N \to (0, 1)^N$, is also a Borel isomorphism. All other objects, the measurement rule~$L$, hidden-variable distributions~$\mu_\rho$ and stochastic maps~$T_g$, can be transported from~$\Lambda$ to~$(0, 1)^N$ (and back) using this isomorphism. This process preserves all potential additional attributes, except the topological ones (continuity and smoothness) because~$J$ is in general highly discontinuous.
    
    If~$\Lambda_1$ is countable we can first pad it with~$(0, 1)$. The extra variables are then initially arbitrarily fixed,~$\mu_\rho\mapsto \mu_\rho\otimes \nu^{\otimes N}$ for some~$\nu\in\mathbb{P}((0, 1))$, kept constant under arbitrary processes,~$T_g \mapsto T_g \otimes \mathrm{Id}$, and ignored by the measurement rule,~$L^a_x \mapsto L^a_x \times (0, 1)$. This does not change the measurement statistics and preserves all additional properties, including continuity. In this way we can make~$\Lambda_1$ uncountable and hence can always employ the reasoning above.

    Now,~$\tilde{\mathfrak{L}}$ is explicitly constructed as follows. 
    For a state~$\rho \in \States$, a process~$g\in\Exp$, a single-particle measurement event~$(a|x)\in\mathbf{M}_1$, a hidden variable~$\tilde\lambda \in (0, 1)^N \equiv \tilde\Lambda$, and a measurable subset of hidden variables~$\tilde{A}\subseteq \tilde{\Lambda}$, we define
    \begin{align} \notag
        \tilde L^{a}_x &\equiv J(L^a_x), \\ \notag 
        \tilde \mu_\rho(\tilde A) &\equiv \mu_\rho(\boldsymbol{J}^{-1}(\tilde A)), \\ 
        \tilde T_g(\tilde A |\tilde\lambda) &\equiv T_g\hspace{-2pt}\left(\boldsymbol{J}^{-1}(\tilde A) \big| \boldsymbol{J}^{-1}(\tilde\lambda)\right).
    \end{align}
    In App.~\ref{app: reduce Lambda} we prove in detail that this constitutes a valid dynamical LHV model with the correct measurement statistics and claimed properties.
\end{proof}

Note that for the goal of efficiently simulating quantum dynamics of noisy states it is crucial for the dynamics to be continuous. Without continuity one cannot guarantee that numerical results approximate the simulated physics in any controlled manner. Therefore, this continuity-destroying reduction~$\Lambda_1 \mapsto (0, 1)$ will not be useful in practice.

\section{Consequences of Local LHV Dynamics}
\label{sec: local LHV dynamics}

In the following, we investigate necessary conditions for local LHV dynamics. It turns out that in general the states under consideration must be more noisy than what follows from locality (of the states) alone.

As always we consider some set of local measurements~$\measurements=\measurements_1^N$, a subset of states that are local with respect to those measurements,~$\States \subseteq \mathcal{L}_\measurements$, and quantum channels~$\channels\subseteq\mathrm{C}(\Hil)$. We already explained earlier that LHV dynamics can only exist if all states transformed by all channels are local, i.e.,~$\channels[\States]\subseteq \mathcal{L}_\measurements$. In fact, we may assume that the set of states is invariant under the quantum dynamics,~$\channels[\States] = \States$. 
Now, for \emph{local} LHV dynamics there are more necessary constraints: Effectively, we must deal with a larger set of measurements. Generically, fewer states are local with respect to larger sets of measurements. Therefore, the set of states~$\States$ must shrink towards the maximally mixed state.

Before we formalize this in full generality, consider the following illustrative example. Let~$\measurements_1$ be any subset of the projective measurements and let~$\channels= \mathrm{PU}(\Hil_1)^N$ be the set of all product unitaries. Naively, one might expect that a sufficient condition for LHV dynamics to be feasible is that the set of all transformed states~$\channels[\States]$ must be local with respect to the measurements in~$\measurements = \measurements_1^N$. And since noninteracting unitaries do not change whether a state is local or not, the condition is simply that the states in~$\States$ must be local with respect to~$\measurements$. For instance, if~$\measurements_1$ only contains a single measurement, then \emph{all} states are trivially local. Yet, we will now show that local LHV dynamics effectively expands the set of available measurements and therefore imposes more stringent conditions on the set of states.

Assume that we have local LHV dynamics. The intuition is that one can use noninteracting unitaries to rotate any fixed local measurement to arbitrary local projective measurements. And with local LHV dynamics this results in a valid LHV model for all these transformed measurements. 
More explicitly, we fix a nondegenerate measurement~$y=(y_1,\ldots,y_N)\in \measurements$. Let~$\rho\in\States$ be a state and let~$x=(x_1,\ldots,x_N)$ be any local projective measurement. In particular,~$x$ might not lie in~$\measurements$. Then there are single-particle unitaries~$U_1,\ldots,U_N\in \mathrm{PU}(\Hil_1)$ such that~$U_k[x_k]=y_k$ and~$U_k$ only depends on~$x_k$, for~$k\in [N]$. It follows that~$U[x]=y$ for the product unitary~$U=U_1\otimes \cdots \otimes U_N\in \mathrm{PU}(\Hil_1)^N$. By the assumption that local LHV dynamics exists (using local realizability), there exist processes~$g_k \in \Exp_{ \{k\} }$ where~$g_k=g_k(x_k)$ only depends on~$x_k$, and such that~$\Phi_{g_k} \cong U_k$. Set~$g\equiv g_1\circ \cdots \circ g_N$. Then, the corresponding hidden-variable transformation is of the form~$T_g=T_{g_1}\circ \cdots \circ T_{g_N}=\bigotimes_{k=1}^N T_{g_k} \in \mathrm{End}(\Lambda)$, and similarly~$\Phi_g \cong U$. Thus, we have
\begin{align} \notag 
    \mathbf{P}(a|x, \rho) 
    &= \mathbf{P}(a|U^\dagger[y], \rho)
    = \mathbf{P}(a|y, U[\rho]) \\ \notag 
    &= (T_g)_\# \mu_\rho(L^a_y) 
    = \int_\Lambda \dup \mu_\rho(\lambda)\, T_g^* Q(a|y, \lambda) \\
    &= \int_\Lambda \dup \mu_\rho(\lambda) \prod_{k=1}^N T_{g_k}^*q(a_k|y_k, \lambda_k).
\end{align}
With the response functions~$\tilde q_k(a_k|x_k, \lambda)\equiv T_{g_k(x_k)}^*q(a_k|y_k, \lambda_k)$ this is a valid LHV model for all projective measurements (here using the formulation of LHV models with a shared hidden variable and one response function per particle). It follows that it is not sufficient for the states in~$\States$ to be local with respect to the measurements in~$\measurements$. The states must actually be local with respect to \emph{all} local projective measurements, even if~$\measurements$ only contains a single (nondegenerate) projective measurement. 

This particular case is a simple way to demonstrate that any Bohmian-like dynamical LHV model must have nonlocal dynamics: In Bohmian mechanics one only considers position measurements (a single measurement option), hence all states can be considered in the theory: they are all local \textit{with respect to position measurements only}. However, local dynamics, even just for noninteracting unitaries, would imply all these states must be local with respect to a much larger set of measurements, which is a contradiction.

For the general case of non-separable channels a similar constraint can be derived. Consider any partition~$\pi = (\pi_j)_{j=1}^J \in \mathfrak{p}([N])$ of the~$N$ particles ($1 \le J \le N$). Then, any state~$\rho\in\States \subseteq \mathbb{D}(\Hil_1^{\otimes N})$ can be regarded as a~$J$-particle state in~$\bigotimes_{j=1}^J \Hil_{\pi_j}$ (with~$\Hil_{\pi_j}\equiv \Hil_1^{\otimes \#\pi_j}$) by grouping the particles according to the partition. One can now ask whether the state~$\rho$ is \lq $\pi$-local\rq, i.e., local with respect to that~$J$-particle partition. These new composite particles indexed by~$j\in [J]$ are generally of different dimensionality. 
For~$J=N$ and~$\pi=(\{j\})_{j=1}^N$,~$\pi$-local just means local. For~$J=1$ and~$\pi=([N])$ every state is~$\pi$-local as we are regarding them as single-particle states. The interesting regime is everything in between.

For any subset~$A\subseteq [N]$ let~$\channels_A \equiv \channels \cap \mathrm{C}(\Hil_A)$ be the subset of channels that only affect particles in~$A$. Then we have the set of measurements on~$\Hil_A$ that can be reached from measurements in~$\measurements_A\equiv\measurements_1^{\# A}$ by applying (adjoints of) channels in~$\channels_A$, that is,
\begin{align}
    \channels^\dagger_A[\measurements_A] 
    \equiv \left\{ \Psi^\dagger[x] | x \in \measurements_A, \Psi \in \channels_A \right\}.
\end{align}
Here, with~$x=(x_1,\ldots, x_{\# A}) \in \measurements_A$,~$x_j=(M^a_{x_j})_{a\in\outcomes_1}$, and~$\Psi \in \channels_A$, the action of channels on measurements is defined in the natural way
\begin{align}
    \Psi^\dagger[x] = \left( \Psi^\dagger \hspace{-2pt}\left[ \bigotimes_{j=1}^{\# A} M^{a_j}_{x_j} \right] \right)_{a \in \outcomes_A}.
\end{align}
It turns out that for local LHV dynamics the states in~$\States$ must be~$\pi$-local, for any partition~$\pi$, with respect to these larger sets of measurements.
\begin{theorem}
    \label{thm: local LHV dynamics consequences}
    If a set of channels~$\channels$ has local LHV dynamics for measurements~$\measurements=\measurements_1^N$ and a set of~$N$-particle states~$\States$, then for all partitions~$\pi\in \mathfrak{p}([N])$ the states in~$\States$ are~$\pi$-local with respect to the measurements in~$\measurements[\pi] \equiv \bigtimes_{j=1}^{\# \pi} \channels_{\pi_j}^\dagger[\measurements_{\pi_j}]$.
\end{theorem}
\begin{proof}
    The proof is similar to the case of product unitaries. The details are given in App.~\ref{app: local LHV dynamics}.
\end{proof}

Our initial example concerning product unitaries corresponds to the special case~$\pi=(\{j\})_{j=1}^N$. Then,~$\pi$-locality reduces to standard locality and~$\measurements[\pi] = \{\text{all local PVMs}\}$, assuming there is a nondegenerate PVM in~$\measurements_1$\footnote{More precisely, we only obtain the nondegenerate projective measurements. But then the states in question must be local with respect to all (local) projective measurements by coarse-graining.}. 

The question arises whether different partitions~$\pi\in\mathfrak{p}([N])$ actually yield more constraints. Let us consider local LHV dynamics for the full unitary group,~$\channels = \mathrm{PU}(\Hil)$. The discussion above suggests that we should consider at least all local projective measurements. Now, as discussed in an earlier section, we have the immediate condition that the states must be absolutely local, i.e., they must remain local with respect to all local PVMs under arbitrary global unitaries,~$\mathrm{PU}(\Hil)[\States]\subseteq \mathcal{L}_\measurements$. The constraint arising from the theorem for the action of noninteracting unitaries is now automatically satisfied. What about nontrivial partitions? 

For any partition~$\pi\in\mathfrak{p}([N])$ the possible local measurements on any subset~$\pi_j\subseteq [N]$ are of the form~$U[x_{\pi_j}]$ where~$x_{\pi_j}\in\measurements_{\pi_j}$ is a local projective measurement on the particles in~$\pi_j$ and~$U\in\mathrm{PU}(\Hil_{\pi_j})$ is any unitary acting on the same subset. This means all projective measurements on that subset can be realized:~$\channels^\dagger_{\pi_j}[\measurements_{\pi_j}]$ consists of all projective measurements on~$\Hil_{\pi_j}$. This is a strictly larger set of measurements compared to~$\measurements_{\pi_j}$. It includes collective measurements that are entangling within each block~$\pi_j$.

We believe it to be an open problem whether this is genuinely more constraining. We formalize the problem as follows: A state is called \emph{local under grouping} if it is local with respect to all partitions of the particles and all corresponding partition-local projective measurements. Then the question becomes: Are there absolutely local states that are nonlocal under grouping? In other words, can nonlocality be activated in absolutely local states by grouping?
For local (but not necessarily absolutely local) states, this is possible: \emph{Superactivation} of quantum nonlocality~\cite{Cavalcanti_Almeida_Scarani_Acin_2011, Palazuelos_2012} can be seen as a special case of grouping\footnote{In the simplest setting, nonlocality can be superactivated in a local state~$\rho\in\mathbb{D}(\Hil_A\otimes \Hil_B)$ if~$\rho\otimes \rho$ displays nonlocal correlations when allowing for collective measurements on the~$A$- and~$B$-parts, respectively. This just means grouping the~$A$- and~$B$-parts of the local state~$\rho\otimes\rho$.}. And there are local states in which nonlocality can be superactivated. 
Fully separable states are local under grouping. Thus, any example would have to be an absolutely local but entangled state.

We can summarize the discussion as follows: For local LHV dynamics the states in~$\States$ must be sufficiently noisy, and in general more noisy than what follows from~\lq$\mathcal{C}[\States] \subseteq \mathcal{L}_\measurements$\rq\ alone. Local LHV dynamics for the full unitary group and any subset of projective measurements requires that the states in~$\States$ are both absolutely local as well as local under grouping. Whether this is more constraining than absolute locality alone is an open problem. 

As a final remark, note that the proof of the theorem reveals an even stronger condition: Any particular state does not only have to be partition local, but the local models for all possible partitions must also share the same hidden-variable distributions. Only the response functions may depend on the partition.

\section{No Free Gauge}
\label{sec: no free gauge}

In this final results section we demonstrate that LHV dynamics with free gauge does not exist if one considers all unitary channels. This result is independent of other possible assumptions such as continuous or local LHV dynamics.

As always, we consider states that are local with respect to some set of local measurements~$\measurements$. We also fix some LHV base model~$(\Lambda_1, L)$ and denote by~$\mathcal{R}\subseteq \mathcal{L}_\measurements$ all those states that can be represented using this LHV base model. That is,~$\rho \in \mathcal{R}$ if and only if there exists a hidden-variable distribution~$\mu_\rho \in \mathbb{P}(\Lambda)$ such that~$\mathbf{P}(m|\rho) = \mu_\rho(L_m)$ for all measurement events~$m\in\mathbf{M}$.  

Assume the LHV base model has LHV dynamics with free gauge for some set of states~$\States$. By definition all these states must be representable,~$\States\subseteq \mathcal{R}$. 
It follows that the LHV dynamics is automatically valid on all representable states~$\mathcal{R}$, as we will show in the following.
\begin{proposition}
    \label{prop: extend LHV dynamics to representable states}
    Let~$(\Lambda_1, L, \mu)$ be an LHV model for local measurements~$\measurements$ and states~$\States \subseteq \mathcal{R}\subseteq \mathcal{L}_\measurements$ that has LHV dynamics with free gauge for a monoid of quantum channels~$\channels$, with process monoid~$\Exp$. Suppose there are states~$\rho, \rho_2 \in \States$, a representable state~$\rho_1\in \mathcal{R}$ and~$t\in (0, 1]$ such that~$\rho = t \rho_1 + (1-t) \rho_2$. 
    If~$\mu'_{\rho_1}\in\mathbb{P}(\Lambda)$ is a valid hidden-variable distribution for~$\rho_1$ then the dynamics transforms its statistics correctly, i.e.,
    \begin{align}
        (T_g)_\# \mu'_{\rho_1}(L_m) = \mathbf{P}(m|\Phi_g[\rho_1]) \quad \forall m\in\mathbf{M},\ g\in\Exp.
    \end{align}
    In particular,~$\Phi_g[\rho_1] \in \mathcal{R}\subseteq \mathcal{L}_{\measurements}$ is a representable, and hence local, state for all processes~$g\in \Exp$.
\end{proposition}
\begin{proof}
    The statement follows from linearity of the quantum measurement statistics as well as the push-forward. See App.~\ref{app: no free gauge} for the details.
\end{proof}

This may not seem surprising. However, the consequences are severe: Under the assumption of LHV dynamics with free gauge it does not really matter how we choose~$\States$. The following corollary shows that for LHV dynamics with free gauge the relevant set of states to consider is the set of all representable states~$\mathcal{R}$. And~$\mathcal{R}$ is fixed solely by the LHV base model.
\begin{corollary} [Invariance of representable states]
    \label{cor: invariance of representable states}
    If an LHV base model with representable states~$\mathcal{R}$ has LHV dynamics with free gauge for states~$\States \supseteq \States_v$, for any visibility~$v>0$, and channels~$\channels\subseteq\mathrm{C}(\Hil)$, then~$\mathcal{R}$ must be invariant under~$\channels$, i.e.,~$\channels[\mathcal{R}]\subseteq \mathcal{R}$. Moreover, the dynamical LHV model can be extended from~$\States$ to~$\mathcal{R}$.
\end{corollary}
\begin{proof}
    Let~$\rho\in\mathcal{R}$ be any representable state and let~$\rho_0 \in \States_v$ be the maximally mixed state. Then, by definition of~$\States_v$
    \begin{align}
        \rho' \equiv v\rho + (1-v) \rho_0 \in \States_v.
    \end{align}
    Thus, with~\lq$t=v>0,\ \rho=\rho',\ \rho_1 = \rho$ and~$\rho_2=\rho_0$\rq, the previous proposition implies~$\Phi_g[\rho]\in\mathcal{R}$, for any process~$g\in\Exp$, i.e.,~$\channels[\mathcal{R}] = \Phi(\Exp)[\mathcal{R}] \subseteq \mathcal{R}$. Moreover, the previous proposition ensures
    \begin{align}
        (T_g)_\# \mu_\rho(L_m) = \mathbf{P}(m|\Phi_g[\rho]) \quad \forall m\in\mathbf{M},\ g\in\Exp,
    \end{align}
    for any hidden-variable distribution~$\mu_\rho\in\mathbb{P}(\Lambda)$ that is valid for~$\rho$. Hence, we have LHV dynamics with free gauge for the set of representable states~$\mathcal{R}$ and channels~$\channels$.
\end{proof}

This corollary means that an LHV base model has LHV dynamics with free gauge for states~$\mathcal{S}\supseteq \States_v$ if and only if it has LHV dynamics with free gauge for its full set of representable states~$\mathcal{R}\supseteq \States_v$. 
This, however, leads to a contradiction when considering all unitary channels. 

As an example, consider an LHV base model that is expressive enough to represent all single-particle states, and hence all separable states. In this case~$\mathcal{R}$ contains a pure product state~$\rho$ (in fact, all of them). Suppose we have LHV dynamics for~$\States_v$,~$v>0$, with free gauge. Then, according to the corollary above, it follows that~$\Phi_g[\rho]$ must be local for all processes~$g\in \Exp$. 
However, for at least two particles, there are many channels~$\Psi\in \mathrm{C}(\Hil)$ such that~$\Psi[\rho]$ is nonlocal. For example, any entangling unitary maps~$\rho$ to a pure entangled, hence nonlocal, state. This means~$\Phi(\Exp)$ can contain no such channel, in particular not all unitaries.

While it is reasonable to assume that~$\mathcal{R}$ contains a pure product state it is in fact not needed for the conclusion: At least for~$N\ge 3$ particles the existence of a dynamical LHV model already guarantees that~$\mathcal{R}$ contains almost pure product states, and this is sufficient to arrive at a contradiction.
We have the following helper lemma. The idea is very similar to what is known as \lq algorithmic cooling\rq\ in the NMR quantum computing literature~\cite{Sorensen_1989, Schulman_Vazirani_1999, Boykin_Mor_Roychowdhury_Vatan_Vrijen_2002, Schulman_2005, Park_Rodriguez-Briones_Feng_Rahimi_Baugh_Laflamme_2016}: Entropy is systematically redistributed among many noisy qubits to cool some subset of them, i.e., make them less noisy, while the others are heated. 
\begin{lemma} [Cooling of noisy product states]
    \label{lem: noisy state cooling}
    Let~$\Hil_1$ be a qudit Hilbert space with~$\dim \Hil_1 = D < \infty$. Let~$\omega \in\mathbb{D}(\Hil_1)$ be a pure single-qudit state. For any visibility~$v\in [0, 1]$ define the depolarized state~$\omega_v \equiv v \omega + (1-v) \rho_0$ where~$\rho_0 = \1/D$ is the maximally mixed single-qudit state. 
    Then there exists a three-qudit unitary~$U$ (independent of~$v$) on~$\Hil_1^{\otimes 3}$ such that
    \begin{align}
        \Tr_{2, 3}(U \omega_v^{\otimes 3} U^\dagger) = \omega_{\varphi(v)}
    \end{align}
    for a continuous function~$\varphi: [0, 1] \to [0, 1]$ with the property that
    \begin{align}
        \varphi^n(v) \xrightarrow{n\to \infty} 1 \quad \forall v\in (0, 1].
    \end{align}
\end{lemma}
\begin{proof}
    See App.~\ref{app: no free gauge}.
\end{proof}

We are now in the position to prove a no-go theorem for LHV dynamics with free gauge.
\begin{theorem} [No LHV dynamics with free gauge for all unitaries]
    \label{thm: no free gauge LHV dynamics}
    For~$N\ge 3$ particles there is no LHV dynamics with free gauge for all local projective measurements, states~$\States \supseteq \States_v$, for any visibility~$v>0$, and all unitaries,~$\Phi(\Exp)\supseteq \mathrm{PU}(\Hil)$.
\end{theorem}
\begin{proof}
    Suppose such LHV dynamics does exist. I.e., there is some dynamical LHV model~$(\Lambda_1, L, \mu, \Exp, T, \Phi)$ for all local PVMs, states~$\States\supseteq\States_v$ and quantum dynamics~$\Phi$ with~$\channels = \Phi(\Exp)\supseteq \mathrm{PU}(\Hil)$. Let~$\rho_0 \in \mathbb{D}(\Hil_1)$ be the maximally mixed single-particle state and for~$\omega \in \mathbb{D}(\Hil_1)$ set~$\omega_w = w \omega + (1-w)\rho_0$,~$w\in [0, 1]$. 
    The idea is to show that the set of representable states~$\mathcal{R}$ contains a product state of the form~$\rho = \omega_w^{\otimes N}$ for~$w>0$ sufficiently close to~$1$ and~$\omega\in\mathbb{D}(\Hil_1)$ a pure single-particle state (using the cooling lemma). Then there exists a unitary channel~$\Phi_g \in\channels$ such that~$\Phi_g[\rho]$ is nonlocal, which is a contradiction according to the corollary above.

    The details are given in App.~\ref{app: no free gauge}.
\end{proof}

The assumption of at least three particles is required for the cooling lemma. For two particles representability of a pure, or almost pure, product state must be postulated to arrive at the same contradiction. Of course, this is not a strong assumption since representability of an almost pure single-particle state is sufficient. The assumption that~$\Phi(\Exp)$ contains all unitary channels is required for both the cooling lemma as well as to generate an entangled state from an almost pure product state. We expect that a similar argument goes through even if all channels in question are slightly noisy, while sufficiently strong noise on the channels will break the argument.

\section{Related Work}
\label{sec: related work}

There are three different kinds of related work that we want to highlight. The first kind deals with temporal correlations, that is, correlations between different points in time. The second kind analyzes the dynamics of hidden-variable models in contrast to \emph{local} hidden-variable models. The third kind concerns particular constructions that we can identify as special cases of dynamical LHV models.

\paragraph{Temporal Correlations.} 
We briefly describe the different existing notions of nonlocality in the context of quantum time evolution, and then explain how the concept presented in this work is fundamentally different. Historically, the first example is \emph{temporal Bell inequalities}, rooted in the work by Leggett and Garg~\cite{Leggett_Garg_1985}. Here, instead of measuring different observables on multiple subsystems a fixed observable is measured on a single system at different times. They show that quantum mechanics is incompatible with the assumptions of macroscopic realism and noninvasive measurements. Brukner et al.~introduced the notion of \emph{entanglement in time}~\cite{Brukner_2004}. Here, different observables, again on a single system, can be measured at several fixed points in time. Temporal Bell inequalities are derived under the assumptions of realism and locality in time. Related to this, also \emph{no-signaling in time} is violated by quantum mechanics~\cite{Kofler_Brukner_2013}. 

A different concept is \emph{dynamical nonlocality}, the term being coined by Popescu, based on the Aharonov-Bohm effect~\cite{Aharonov_Bohm_1959, Popescu_2010}. Here, the observation is that the Heisenberg equations of motion for certain observables are nonlocal in space. \emph{Nonclassicality of temporal correlations} is defined by Brierley et al.~as correlations of an~$m$-level system for measurements at different times that require more than~$\log_2(m)$ bits of communication in a classical simulation~\cite{Brierley_2015}. \emph{Entangled histories} are a concept in the quantum history states formalism by Cotler and Wilczek~\cite{Cotler_Wilczek_2016} based on the work by Griffiths on consistent histories~\cite{Griffiths_1984}. They appear when trying to assign a consistent past evolution based on present measurements of a quantum system. 
In the perspectives article \emph{Spooky Action at a Temporal Distance}~\cite{Adlam_2018} Adlam considers many different approaches. In particular, she introduces a notion of \emph{temporal locality} meaning measurement statistics should only depend on the physical state immediately prior to the measurement.

The approach discussed in the present article is different in a fundamental way. Most of the approaches above study correlations between different points in time. We study how the equal-time correlations evolve through time. 
The simplest way to see that the different notions mentioned above are indeed physically inequivalent to the statement that \lq LHV dynamics does not exist\rq\ is the following: Temporal Bell inequalities can be violated by single-particle systems. Likewise, single-particle interference displays dynamical nonlocality and also the entangled histories appear for single-particle systems. In contrast, we have demonstrated that LHV dynamics exists for noninteracting systems, in particular for single particles. Only the evolution of multi-particle measurement statistics under interacting time evolution can be incompatible with LHV dynamics and we have demonstrated that under certain assumptions this indeed happens.

\paragraph{Dynamics of Hidden-Variable Models.} 
Next, there are several works investigating dynamics of hidden-variable models, in contrast to dynamics of \emph{local} hidden-variable models, although not all of them frame it that way. 
Milburn et al.~investigate a simple classical model for the dynamics of two qubits under an interacting Hamiltonian~\cite{Milburn_Laflamme_Sanders_Knill_2002}. They show that, for their specific model, no positive transition propagator (a stochastic map) can reproduce the quantum dynamics correctly, even if the initial and final states are highly mixed. This is similar in spirit to our two-qubit Heisenberg counterexample. Moreover, they argue that the power of quantum computation lies in the quantum dynamics and not (just) entangled states, consistent with our own results. Finally, they raise the open question whether a dynamical hidden-variable model exists for sufficiently noisy states, remarking that what we called cheat dynamics is undesirable due to the exponentially large state space. 
A similar setting is studied by Silveira and Angelo~\cite{Silveira_Angelo_2017}: They construct a hidden-variable model with local dynamics as the classical limit of quantum mechanics for two interacting qubits. Their model qualitatively captures the entanglement dynamics but not the nonlocality dynamics.

Spekkens analyzes \emph{contextuality}, not only of measurements and preparation procedures but also of transformations~\cite{Spekkens_2005}. In particular, he shows that transformations are contextual. This means the stochastic map describing a particular quantum channel cannot only depend on the channel. It must also depend on the actual process that implements the channel. In our definitions we account precisely for this by describing both the quantum as well as the LHV dynamics in terms of a monoid of physical processes~$\Exp$.

Hardy shows that ontological models of quantum mechanics require infinitely many different states~\cite{Hardy_2004}. Of course this says nothing about the dimension of the hidden-variable space. 
Closer to our definitions and results, Montina defines hidden-variable dynamics via state-independent stochastic maps~\cite{Montina_2008, Montina_2011}. He showed that if one considers all unitaries and all pure quantum states, then the hidden-variable space must be exponentially large in the number of particles. In this sense our dimensionality constraint in Thm.~\ref{thm: dim-constraint for no go} is a complementary result addressing the more general situation of arbitrarily mixed states. Montina's result depends on two assumptions that are conceptually very similar to what we call continuous reversible dynamics. However, his precise definitions differ from ours. In particular, his notion of reversibility is weaker.

\paragraph{Dynamics of LHV Models.} 
Finally, we mention constructions that correspond, more or less, to what we call dynamical LHV models. 
The most prominent example is Bohmian mechanics~\cite{Bohm_1952_1, Bohm_1952_2}. In our language, Bohmian mechanics can be regarded as a dynamical LHV model for position measurements: Any state is local with respect to position measurements only and the squared wave function is a valid hidden-variable distribution. Its nonlocal character is only visible in the dynamics. The wave function can be interpreted as a shared hidden variable. The time evolution is then nonlocal in the sense that the interaction structure is not respected. Formally, such a global shared hidden variable always corresponds to all-to-all interactions. While our no-go result does not assume anything about the interaction structure, there is no contradiction since Bohmian mechanics deals with infinite-dimensional systems: Both the quantum system as well as the hidden variables (if one includes the wave function) are infinite-dimensional. Additionally, having the full wave function as a hidden variable corresponds to a particle-number dependent single-particle hidden-variable space.

Aaronson considers qubits and, similar in spirit to Bohmian mechanics, restricts to a single fixed measurement basis~\cite{Aaronson_2005}. His construction can be interpreted as taking one classical bit per particle as the hidden variable. On the resulting discrete and finite hidden-variable space he defines dynamics via stochastic matrices (a special case of stochastic maps). However, due to the particular choice of the underlying LHV model these matrices cannot be state-independent in his case: his dynamics is explicitly state-dependent. Otherwise, he defines similar properties to us: \lq indifference\rq\ and \lq commutativity\rq\ essentially amount to what we call local LHV dynamics, and \lq robustness\rq\ is a stronger version of continuous LHV dynamics. He then proves several no-go theorems, for example, indifference and commutativity cannot both be satisfied. 

Lastly, a dynamical LHV model with local microscopic (hence reversible and deterministic) LHV dynamics with consistent gauge can be extracted from the works of Gross~\cite{Gross_2006}, Veitch et al.~\cite{Veitch_Ferrie_Gross_Emerson_2012}, and Mari and Eisert~\cite{Mari_Eisert_2012}. It covers a discrete setting with finitely many measurements and quantum channels in contrast to the continuum setting we have focused on (hence continuity is automatic). Specifically, for qudits of odd prime dimension~$D$ a discrete version of the Wigner function is nonnegative for all stabilizer states (and convex combinations of them). The discrete single-particle hidden-variable space is given by~$\Lambda_1 = \mathbb{Z}_D^2$. Then, for the so-called phase-space point operators,~$A_\lambda\in\mathbb{B}(\Hil_1),\ \lambda \in \mathbb{Z}_D^2$, we have the discrete hidden-variable distributions~$\mu_\rho(\lambda) = \Tr(\rho A_{\lambda})/D^N$ where~$A_\lambda = A_{\lambda_1}\otimes \cdots \otimes A_{\lambda_N}$ and~$\lambda \in \Lambda=(\mathbb{Z}^2_D)^N$. Together with the (deterministic) local measurement rule given by~$q(a|x, \lambda) = \Tr(M^a_x A_\lambda)\in \{0, 1\}$ this yields a proper LHV model for all stabilizer states and local stabilizer measurements. Finally, the Clifford gates~$U\in \{\text{Clifford Gates}\}$ map stabilizer states to stabilizer states and they correspond to permutations~$T_U:\Lambda \to \Lambda$ of the~$N$-particle phase space~$\Lambda = (\mathbb{Z}_D^{2})^N$ such that~$\mu_{U[\rho]}(\lambda) = \mu_\rho(T^{-1}_U(\lambda))$. Crucially, if~$U$ only affects a subset of the particles, then so does the permutation~$T_U$. There is no contradiction with our dimensionality constraint since the (projective) Clifford group is zero-dimensional as a finite group. Note that some of the Clifford gates are entangling and similarly, some stabilizer states are entangled.

A similar construction by Kocia and Love~\cite{Kocia_Love_2017} that covers the qubit case does not constitute a proper dynamical LHV model in the sense of this work, since qubit stabilizer states can be nonlocal with respect to Pauli measurements~\cite{Mermin_1990}.

\definecolor{lhvyes}{RGB}{138,192,221}
\definecolor{lhvno}{RGB}{253,174,155}
\definecolor{lhvna}{RGB}{212,212,212}
\newcommand{\ticksym}{\Checkmark}
\newcommand{\crosssym}{\XSolidBrush}
\newcommand{\nasym}{\rule{1.2em}{1.2pt}}
\newlength{\lhvcellwd}\setlength{\lhvcellwd}{1.8em}
\newcommand{\lhvrot}[1]{%
  \makebox[\lhvcellwd][c]{\rotatebox{90}{\hspace{4pt}#1\hspace{4pt}}}}
\newlength{\lhvcellht}\setlength{\lhvcellht}{2.3ex}
\newlength{\lhvcelldp}\setlength{\lhvcelldp}{1.2ex}
\newlength{\lhvcellwdS}\setlength{\lhvcellwdS}{1.4em}
\newlength{\lhvcellhtS}\setlength{\lhvcellhtS}{2.0ex}
\newlength{\lhvcelldpS}\setlength{\lhvcelldpS}{0.8ex}
\newcommand{\lhvstrut}{%
  \rule[-\lhvcelldp]{0pt}{\dimexpr\lhvcellht+\lhvcelldp\relax}}
\newcommand{\lhvstrutS}{%
  \rule[-\lhvcelldpS]{0pt}{\dimexpr\lhvcellhtS+\lhvcelldpS\relax}}
\newcommand{\lhvcell}[2]{%
  {\setlength{\fboxsep}{0pt}%
   \sbox0{#2}%
   \colorbox{#1}{\makebox[\lhvcellwd][c]{%
     \lhvstrut
     \raisebox{\dimexpr 0.5\lhvcellht-0.5\lhvcelldp
                        -0.5\ht0+0.5\dp0\relax}{\usebox0}}}}}
\newcommand{\lhvcellS}[2]{%
  {\setlength{\fboxsep}{0pt}%
   \sbox0{#2}%
   \colorbox{#1}{\makebox[\lhvcellwdS][c]{%
     \lhvstrutS
     \raisebox{\dimexpr 0.5\lhvcellhtS-0.5\lhvcelldpS
                        -0.5\ht0+0.5\dp0\relax}{\usebox0}}}}}
\newcommand{\yes}{\lhvcell{lhvyes}{\ticksym}}
\newcommand{\no}{\lhvcell{lhvno}{\crosssym}}
\newcommand{\yesS}{\lhvcellS{lhvyes}{\ticksym}}
\newcommand{\noS}{\lhvcellS{lhvno}{\crosssym}}
\newcommand{\na}{\lhvcell{lhvna}{\nasym}}
\newcommand{\lhvhead}[1]{\makebox[\lhvcellwd][c]{#1}}
\newcommand{\lhvangle}{25}
\newlength{\lhvheadht}
\settowidth{\lhvheadht}{\hspace{2pt}consist.\ gauge} 
\setlength{\lhvheadht}{0.7071\lhvheadht}
\newcommand{\lhvvc}[1]{%
  \raisebox{\dimexpr 0.5\lhvheadht-0.5ex\relax}{#1}}
\newlength{\lhvangoffset}
\setlength{\lhvangoffset}{0.2\lhvcellwd}   

\newcommand{\lhvang}[1]{%
  \makebox[\lhvcellwd][l]{%
    \hspace*{\lhvangoffset}\rlap{\rotatebox{\lhvangle}{#1}}}}
\begin{table*}[ht]
\renewcommand{\arraystretch}{0}
\setlength{\tabcolsep}{6pt}
\caption{\textbf{Overview of results.} Given a set of local measurements~$\measurements$, a set of local states~$\States \subseteq \mathcal{L}_\measurements$, and quantum channels~$\channels=\Phi(\Exp)\subseteq \mathrm{C}(\Hil)$ one can ask whether LHV dynamics with some specified additional properties exists. This table summarizes our constructive existence results as well as no-go results and other constraints. Here, we always consider all local projective measurements for~$\measurements$. Also, we assume that there is an LHV model for~$\measurements$ and~$\States$ since this is always necessary for LHV dynamics to exist. Finally, we assume that the set of states is invariant under the quantum dynamics,~$\channels[\States] \subseteq \States$ (if LHV dynamics exists in the first place, we may enlarge~$\States$ to make this condition true). The set~$\States_v$ refers to a noisy ball of states with positive visibility~$v>0$. The unitary channels, given by projective unitaries, are denoted by~$\mathrm{PU}(\Hil)$. Note that microscopic implies reversible if the channels form a group and reversible always implies deterministic. 
A tick in the consequences means LHV dynamics exists, a cross means it does not exist. 
A ticked attribute is satisfied by the construction or required for the no-go result, a cross means the attribute is not satisfied by the particular construction, a bar means it is not required for the no-go result. 
This summary table cannot represent every statement in full detail, hence we refer to the more precise theorem statements.}
\label{tab:results_overview}
\centering
\begin{tabular}{|p{3.5cm}|@{}c@{}|@{}c@{}|@{}c@{}|@{}c@{}|@{}c@{}|@{}c@{}|@{}c@{}|@{}c@{}|@{}c@{}|p{6.1cm}|}
\hline\hline
\multicolumn{1}{|l}{\lhvstrut Name}
  & \multicolumn{1}{l}{States}
  & \multicolumn{1}{l}{Channels}
  & \multicolumn{1}{@{}c@{}}{\lhvang{local}}
  & \multicolumn{1}{@{}c@{}}{\lhvang{continuous}}
  & \multicolumn{1}{@{}c@{}}{\lhvang{consist.\ gauge}}
  & \multicolumn{1}{@{}c@{}}{\lhvang{free gauge}}
  & \multicolumn{1}{@{}c@{}}{\lhvang{microscopic}}
  & \multicolumn{1}{@{}c@{}}{\lhvang{reversible}}
  & \multicolumn{1}{@{}c@{}}{\lhvang{deterministic}}
  & \multicolumn{1}{c|}{Consequences} \\
\hline\hline
   Cheat LHV dynamics & any & any
  & \no  & \yes & \yes & \no  & \no & \no & \no
  & \yesS\ w/~$\dim\Lambda_1\geq \dim \States$, Thm.~\ref{thm: cheat LHV dynamics} \\ 
\hline
  \hspace{10pt}~$\rotatebox[origin=c]{180}{$\Lsh$}$ drop continuity & any & any
  & \no  & \no  & \yes & \no  & \no & \no & \no
  & \yesS\ Thm.~\ref{thm: cheat LHV dynamics}~$+$ Thm.~\ref{thm: simple hidden-variable space} \\
\hline
   Noninteracting unitaries & any &~$\mathrm{PU}(\mathcal{H}_1)^{N}$
  & \yes & \yes & \no  & \yes & \yes & \yes & \yes
  & \yesS\ Thm.~\ref{thm: local_unitaries} under hyp.~on LHV base model \\
\hline
   Dimension constraint &~$\supseteq\States_v$ &~$\mathrm{PU}(\mathcal{H})$
  & \na  & \yes & \na  & \na  & \na  & \yes & \yes
  & \noS\ Cor.~\ref{cor:physical_nogo} b/c~$\dim\Lambda_1 \geq D^{N}/N$, Thm.~\ref{thm: dim-constraint for no go} \\
\hline
   Bell model + Heisenberg &~$\supseteq \States_v$ &~$e^{-iHt}$ & \na & \yes & \yes & \na & \na & \na & \yes 
& \noS\ for a particular model \& gauge, Thm.~\ref{thm: heisenberg_counterexample} \\ 
\hline 
   No free gauge &~$\supseteq\States_v$ &~$\supseteq\mathrm{PU}(\mathcal{H})$
  & \na  & \na  & \na  & \yes & \na  & \na  & \na
  & \noS\ Thm.~\ref{thm: no free gauge LHV dynamics} for~$N\ge 3$ \\
\hline 
   Local LHV dynamics & any & any
  & \yes & \na  & \na  & \na  & \na  & \na  & \na
  &~$\States$ local w.r.t.~all partitions, Thm.~\ref{thm: local LHV dynamics consequences} \\  
\hline\hline
\end{tabular}
\end{table*}

\section{Conclusion}

Bell's theorem establishes that there are quantum states which produce nonlocal correlations. Our work addresses the fundamentally different question of whether a local description can account for the quantum dynamics, assuming the instantaneous correlations are local. 
More specifically, we raised and investigated the question of whether the quantum dynamics of local correlations that remain local can be captured by evolving the hidden variables of an underlying LHV model. Our results are summarized in Table~\ref{tab:results_overview}. 

For this purpose, we introduced LHV models for sets of states and defined LHV dynamics as state-independent transformations of the hidden variables such that the resulting change of the measurement statistics matches quantum mechanics. We motivated a range of additional physical and mathematical constraints that one may impose on such dynamics, such as locality, continuity, and reversibility. See Table~\ref{tab:attributes_and_consequences} for the complete list.

We demonstrated that for noninteracting quantum dynamics there are indeed \lq nice\rq\ dynamical LHV models, with smooth deterministic reversible microscopic local LHV dynamics with free gauge. In particular, this includes single particles. Therefore, the question of whether LHV dynamics exists reduces to the study of interacting systems. 

For two qubits coupled with a Heisenberg interaction we saw that the particular Bell LHV base model has time-dependent hidden-variable distributions (a dynamical gauge) matching the quantum correlations at each moment in time. Yet, we found it is impossible even in this simple case to construct smooth deterministic dynamics of the hidden variables that would reproduce this particular evolution. We conjectured that any choice of gauge will run into the same problem for this particular base model. 

It turned out that the gauge, i.e., the particular choice of initial hidden-variable distributions, is more important than one might initially expect: it is not free. LHV dynamics for the unitary group (states~$\States_v$ and all projective measurements) with free gauge does not exist for~$N\ge 3$ particles, and also not for~$N=2$ particles under a mild extra assumption. 
This means the initial hidden-variable distributions always have to be fine-tuned in conjunction with the hidden-variable dynamics. In particular, given \emph{any} such dynamical LHV model, there are always states and valid hidden-variable distributions for those states whose measurement statistics do \emph{not} transform correctly under the hidden-variable dynamics.

Ideally we would like to have \emph{local} LHV dynamics, which respects the physical interaction structure. Focusing on LHV dynamics for the full unitary group here, we found two necessary conditions on the considered states: First, they must be absolutely local, i.e., remain local under arbitrary unitary transformations. And second, they must also be local under grouping, i.e., when grouping multiple particles together and enlarging the space of measurements accordingly. It remains open whether this is more constraining than absolute locality alone.

If one drops the assumption of local LHV dynamics, we saw that under relatively weak assumptions one can always construct explicitly nonlocal but continuous \lq cheat\rq\ LHV dynamics from a given LHV model. However, the single-particle hidden-variable space turns out to be exponentially large in the total number of particles. 
In fact, under the additional assumption of reversible LHV dynamics this is unavoidable: We showed that LHV models with continuous reversible dynamics for groups of unitaries must adhere to a dimensionality constraint. 
This implies a general no-go result: No fixed LHV base model can have continuous reversible LHV dynamics for all unitary channels for sufficiently many particles, even for very noisy separable states. The maximal number of particles for which LHV dynamics cannot be excluded by this theorem only grows roughly logarithmically with the dimension of the single-particle hidden-variable space. 

The main technical ingredient for the proof is that we consider reversible LHV dynamics. This means assuming physical processes are, at least in principle, reversible on a fundamental level. That leads to algebraic relations between different processes, which then must be respected on the level of the hidden-variable transformations, and not just on the level of measurement statistics. On the other hand, the no-go result does not assume local LHV dynamics. That is, even LHV dynamics with nonlocal all-to-all interactions is excluded. 
Non-continuous LHV dynamics can never be subject to such dimensionality constraints, independent of whether the dynamics is assumed to be local or not.

One may draw an analogy to the works of von Neumann and Bell~\cite{von_Neumann_1932, Bell_1964, Bell_1966}: Von Neumann excluded hidden-variable extensions of quantum mechanics assuming that algebraic relations between observables are captured on the hidden-variable level, while he did not make any assumptions on locality. Much later, Bell showed that under the additional assumption of locality hidden-variable models for quantum mechanics can even be excluded if one requires the relations between different observables to only hold on the level of expectation values. It would be interesting to figure out whether, in spirit of this analogy, LHV dynamics for the unitary group can be excluded assuming local but not necessarily reversible LHV dynamics.

While there remain many open questions we want to highlight the following: \emph{Is there a dynamical LHV model for a set of states that contains some entangled states, all local projective measurements, and all unitary channels, with local continuous LHV dynamics?} If such a construction exists, is the single-particle hidden-variable space necessarily exponentially large in the number of particles (even for non-reversible dynamics)?

Beyond the questions already raised above, there is room for generalizations and extensions of this work. For instance, one could apply the ideas of this article to continuous-variable systems. Another direction is to investigate whether there are critical noise levels on the set of considered quantum channels above which LHV dynamics becomes possible. One may also consider time evolution of local correlations independent of quantum mechanics and ask in more general settings whether the evolution is compatible with LHV dynamics or not. This would be particularly interesting from an experimental point of view: Our current no-go results rely on the mathematical structure of quantum mechanics. Therefore, the obvious experimental verification is to demonstrate that \lq quantum mechanics is correct\rq. If one could derive constraints restricting \emph{how} local correlations can or cannot change under local LHV dynamics (similar to Bell inequalities), then a simpler experimental verification should become possible.

Finally, we have focused on the \lq continuum setting\rq, with full-dimensional subsets of states, all projective measurements, and all unitary channels. One may ask the same kinds of questions for finite sets of states, measurements, and channels. As explained in the related work section, stabilizer states together with Clifford gates for qudits of odd prime dimension constitute a positive result of this kind. Are there such discrete settings where LHV dynamics does not exist? This regime is particularly suited for numerical and experimental investigations.

\section*{Acknowledgments}

The research is part of the Munich Quantum Valley (K-$4$ and K-$8$), which is supported by the Bavarian state government with funds from the Hightech Agenda Bayern Plus.

\paragraph{AI Statement.} 
Artificial intelligence tools have been used in this work, specifically, Claude Code and Claude Chat based on the models Opus~($4.8, 5.5$) and Fable~($5, 5.1$) by Anthropic. 
The use cases were the following: as a discussion partner and thereby as a source of inspiration, to assist with mathematical proofs, to suggest literature sources, and as a critic providing feedback. 
The outputs generated by the AI tools have been dealt with as follows: The authors actively decided on the final phrasing of all definitions, theorems, proofs, explanations etc. All proofs were checked in detail by the authors, independent of whether the proof idea came from the authors or an AI tool. Similarly, all literature sources and references were checked by the authors, independent of who suggested any particular source. 
All core physical ideas and the motivation for this work originate with the authors. The text has been written by the authors in the sense that no AI tool has directly edited the text. Wherever phrasings suggested by the AI tools have been used, this was actively and deliberately decided by the authors.
The authors take full responsibility for the correctness and integrity of the entire manuscript.

\appendix
\counterwithin{theorem}{section}  

\section{Expressive LHV Models for Continuous Sets of Measurements}
\label{app: general lhv models}

In this appendix, we demonstrate how to unify state-dependent hidden-variable models. Specifically, assume we have a set of different hidden-variable models, given by hidden-variable spaces~$\Lambda_\rho$, (deterministic) local measurement rules~$q_\rho$ and hidden-variable distributions~$\mu_\rho$. We will show how to transform them into a standard form with a single unified single-particle hidden-variable space~$\Lambda_1$ and measurement rule~$q$ that are independent of the state~$\rho$ and particle number~$N$. Although the focus is different, the calculation is essentially the same as in~\cite{von_Selzam_Marquardt_2025}.

We assume that there are only finitely many possible measurement outcomes,~$\#\outcomes_1 = \Delta < \infty$. Then, we write~$\vec q_\rho(x, \lambda)\in\{\hat e_j\}_{j=1}^\Delta$ for the deterministic probability vector with entries corresponding to the different outcomes. In this case,~$\vec q_\rho$ can always be written as the \lq argmax\rq\ of some (non-unique)~$\R^\Delta$-valued function~$\vec f_\rho$ 
\begin{align}
    \vec q_\rho(x, \lambda) = \mathrm{argmax} \vec f_\rho(x, \lambda) = \hat e_j,
\end{align}
for~$\max \vec f_\rho(x, \lambda) = (\vec f_\rho(x, \lambda))_j$.
We now expand the functions~$\vec f_\rho(x,\lambda)$ with respect to the observables~$x$ into~$K$ basis functions~$B_n(x)$ of~$L^2(\measurements_1, \R)$, with vector-valued coefficients depending on the hidden variable~$\lambda$ and the state~$\rho$:
\begin{align}
    \vec f_\rho(x, \lambda) = \sum_{n=1}^K \vec c^{\,(\rho)}_n(\lambda)B_n(x),\quad \vec c^{\,(\rho)}_n(\lambda) \in \R^\Delta.
\end{align}
Finally, we perform a change of variables,~$\boldsymbol{\lambda}_{jn} \equiv (\vec c^{\,(\rho)}_n(\lambda))_j$, such that we obtain new, matrix-valued hidden variables~$\boldsymbol{\lambda} \in \R^{\Delta \times K} \equiv \Lambda_1$. This, of course, modifies the hidden-variable distributions, which is not a problem. Furthermore, we arrange the basis functions~$B_n(x)$ into a vector~$\vec B(x)\in \R^K$. Then, we obtain the state-independent measurement rule
\begin{align}
    \label{eq: general local measurement rule}
    \vec q(x, \boldsymbol{\lambda}) = \mathrm{argmax}\,\boldsymbol{\lambda} \vec B(x).
\end{align}
Since the argmax function does not depend on the norm of its argument one may normalize the hidden-variable matrices to obtain~$\boldsymbol \lambda \in S^{\Delta\times K - 1}$, if a compact hidden-variable space is desired.

The only restriction for a finite-dimensional hidden-variable space is a finite cutoff~$K<\infty$.  
Increasing the cutoff leads to a more expressive model. For example, for qubits projective measurements are just spin measurements, which can be parameterized by unit vectors in~$\R^3$. A natural choice for the basis functions are the spherical harmonics. Then, considering spherical harmonics up to degree~$l=1$ already allows one to represent all separable~$N$-qubit states exactly and considering larger cutoffs (e.g.~up to degree~$l=5$) allows one to represent many local entangled states, for instance, essentially all local two-qubit Werner states (see our previous work~\cite{von_Selzam_Marquardt_2025}).

\section{Stochastic Maps} 
\label{app: stochastic_maps}

In this appendix we collect some standard facts about stochastic maps. This makes the paper more self-contained for readers unfamiliar with stochastic maps. Furthermore, it frames the results from different literature sources in a unified notation consistent with the rest of the paper. Finally, we include short proofs for some statements that we expect are well-known in the mathematics community, for which we, however, could not find an explicit reference.

We assume throughout this appendix that the hidden-variable space~$\Lambda$ is a standard Borel space.
First, what we refer to as a \emph{stochastic map}~$T$ is often called a \emph{Markov kernel} or \emph{Markov transition} in the literature. Formally, a stochastic map~$T\in\mathrm{End}(\Lambda)$ assigns to each point~$\lambda\in\Lambda$ a probability measure~$T(\lambda)\in\mathbb{P}(\Lambda)$ such that the mapping~$\Lambda \ni \lambda \mapsto T(A|\lambda)\equiv T(\lambda)(A)\in[0, 1]$ is measurable for each fixed measurable subset~$A\subseteq \Lambda$~\cite[Chpt.~4]{Fritz_2020}. Intuitively speaking, instead of mapping~$\lambda \in\Lambda$ deterministically to some~$\lambda'\in\Lambda$ we map it to~$\lambda'$ with probability~$T(\dup \lambda'|\lambda)$.

Stochastic maps~$T\in\mathrm{End}(\Lambda)$ transform probability measures~$\mu\in\mathbb{P}(\Lambda)$ via the push-forward (this is the chain rule for probabilities)
\begin{align}
    T_\# \mu(A) = \int_\Lambda T(A|\lambda) \dup \mu(\lambda).
\end{align}

The push-forward is an affine map (preserving convex combinations) on the space of probability measures~$T_\# : \mathbb{P}(\Lambda) \to \mathbb{P}(\Lambda)$.
Similarly, multiple stochastic maps~$T, S \in\mathrm{End}(\Lambda)$ can be composed (this is again the chain rule for probabilities)
\begin{align}
    T\circ S(A|\lambda) = \int_\Lambda T(A|\lambda') S(\dup\lambda'|\lambda).
\end{align}
Equivalently
\begin{align}
    T\circ S(\lambda) \equiv (T\circ S)(\lambda) = T_\# (S(\lambda)) \equiv T_\# S(\lambda).
\end{align}

The two operations, push-forward and composition, commute in the following sense.
\begin{proposition} \label{prop: push-forward and composition commute}
    For stochastic maps~$T, T' \in \mathrm{End}(\Lambda)$ we have~$(T\circ T')_\# = T_\# T'_\#$ as maps on~$\mathbb{P}(\Lambda)$.
\end{proposition}
\begin{proof}
    Probability measures can be regarded as constant stochastic maps: For~$\mu\in\mathbb{P}(\Lambda)$ we have~$S_\mu \in \mathrm{End}(\Lambda)$ where~$S_\mu(A |\lambda) = \mu(A)$ for all~$\lambda \in\Lambda$ and measurable~$A\subseteq \Lambda$. From this perspective the push-forward of measures is simply composition of stochastic maps
    \begin{align} \notag
        S_{T_\# \mu}(A|\lambda)
        &= (T_\# \mu) (A) 
        = \int_\Lambda T(A|\lambda') \dup \mu(\lambda') \\
        &= \int_\Lambda T(A|\lambda') S_\mu(\dup \lambda'| \lambda) = (T\circ S_\mu )(A | \lambda),
    \end{align}
    hence~$S_{T_\# \mu} = T\circ S_\mu$.
    Now use that composition of stochastic maps is associative (this is part of the fact that \lq measurable spaces as objects together with stochastic maps as morphisms form a category\rq, see e.g.~\cite[Chpt.~4]{Fritz_2020}): For all~$\mu\in\mathbb{P}(\Lambda)$ we have
    \begin{align} \notag
        S_{(T\circ T')_\# \mu} 
        &= (T\circ T') \circ S_\mu 
        = T \circ (T' \circ S_\mu) \\
        &= T \circ S_{T'_\# \mu} 
        = S_{T_\# (T'_\# \mu)}.
    \end{align}
    Therefore,~$(T\circ T')_\# \mu = T_\# (T'_\# \mu)$ since~$S_\mu$ determines~$\mu$. Finally, the claim~$(T\circ T')_\# = T_\# T'_\#$ follows because~$\mu\in\mathbb{P}(\Lambda)$ was arbitrary.
\end{proof}

A stochastic map~$T\in \mathrm{End}(\Lambda)$ is \emph{deterministic} if it corresponds to a measurable map~$\tau: \Lambda \to \Lambda$. That is, if
\begin{align}
    T(A|\lambda) = \1_A(\tau(\lambda)) = \delta_{\tau(\lambda)}(A) = \begin{cases}
        1 \text{ if } \tau(\lambda) \in A \\
        0 \text{ else.}
    \end{cases}
\end{align}
In this case the push-forward of measures simplifies to~$T_\#\mu(A) = \mu(\tau^{-1}(A))$ and composition becomes composition of the underlying maps~$T\circ T'(\lambda) = \delta_{\tau \circ \tau'(\lambda)}$.

The identity morphism is deterministic and given by the identity map~$\mathrm{Id}_\Lambda$ on~$\Lambda$
    \begin{align}
        \mathrm{Id}(A|\lambda) = \begin{cases}
            1 \quad \text{if } \lambda \in A, \\
            0 \quad \text{otherwise}
        \end{cases}
        = \delta_{\lambda}(A) = \1_A(\mathrm{Id}_\Lambda(\lambda)).
    \end{align}
Its push-forward is the identity on~$\mathbb{P}(\Lambda)$:
\begin{align}
    \mathrm{Id}_\# \mu (A) = \int_\Lambda \delta_\lambda(A) \dup \mu(\lambda) = \mu(A).
\end{align}

Composing a stochastic map~$T\in\mathrm{End}(\Lambda)$ with a deterministic stochastic map~$S\in\mathrm{End}(\Lambda)$, given by a measurable map~$s: \Lambda \to \Lambda$, is the same as composing the stochastic map with the deterministic mapping~$s$:
\begin{align} \notag 
    (T\circ S)(\lambda)
    &= \int_\Lambda T(\lambda') \dup \delta_{s(\lambda)}(\lambda')
    = T(s(\lambda)).
\end{align}

A stochastic map~$T\in\mathrm{End}(\Lambda)$ is \emph{invertible} if there exists a stochastic map~$T' \in \mathrm{End}(\Lambda)$ such that~$T \circ T' = T' \circ T = \mathrm{Id}$.
We use the results in an extensive article by Tobias Fritz~\cite{Fritz_2020} to prove that invertible stochastic maps are deterministic.
\begin{proposition} \label{prop:invertible_stochastic_map_is_deterministic}
    Let~$T\in\mathrm{End}(\Lambda)$ be an invertible stochastic map. Then~$T$ is deterministic in the sense that there exists a measurable bijection~$\tau: \Lambda \to \Lambda$ (with measurable inverse) such that~$T(A|\lambda) = \1_A(\tau(\lambda))$ for all~$\lambda \in \Lambda$ and all measurable subsets~$A \subseteq \Lambda$.
\end{proposition}
\begin{proof}
    Taking standard Borel spaces as objects together with stochastic maps as morphisms yields the category \textsc{BorelStoch}~\cite[Chpt.~4]{Fritz_2020}. 
    The isomorphisms in this category,~$T\in \mathrm{End}(\Lambda)$ such that there exists~$T^{-1}\in\mathrm{End}(\Lambda)$ with~$T\circ T^{-1} = T^{-1}\circ T = \mathrm{Id}$, are exactly the invertible stochastic maps.

    The meaning of the following \lq intermediate attributes\rq\ can be ignored, only note that Fritz means something slightly weaker by \lq deterministic\rq\ compared to what we need: \textsc{BorelStoch} has conditionals~\cite[Ex.~11.7]{Fritz_2020}, hence it is positive~\cite[Lem.~11.24]{Fritz_2020}, hence every isomorphism is deterministic~\cite[Rem.~11.28]{Fritz_2020}. Deterministic morphisms have image in~$\{0, 1\}$-valued measures~\cite[Ex.~10.4]{Fritz_2020} and in \textsc{BorelStoch} morphisms with image in~$\{0, 1\}$-valued measures have image in the delta measures~\cite[Ex.~10.5]{Fritz_2020}. 
    
    Taken together, every invertible~$T\in\mathrm{End}(\Lambda)$ has image in the delta measures. That is, for every~$\lambda\in\Lambda$ there is a~$\lambda'\in\Lambda$ such that~$T(\lambda) = \delta_{\lambda'}$. I.e., defining~$\tau(\lambda)\equiv \lambda'$,
    \begin{align}
        T(A|\lambda) = \delta_{\tau(\lambda)}(A) = \1_A(\tau(\lambda))
    \end{align}
    for all~$\lambda\in\Lambda$ and all measurable subsets~$A\subseteq \Lambda$. The mapping~$\tau:\Lambda \to \Lambda$ must be measurable since~$T$ is a stochastic map (see~\cite[Ex.~10.4]{Fritz_2020}).

    The inverse~$T^{-1}$ is also an isomorphism, hence corresponds to a measurable~$\tau':\Lambda \to \Lambda$ with
    \begin{align}
        \1_A(\lambda) = T\circ T^{-1}(A|\lambda) = \1_A(\tau\circ \tau'(\lambda))
    \end{align}
    for all measurable~$A\subseteq \Lambda$. Therefore,~$\tau\circ \tau' = \mathrm{Id}_\Lambda$ (take the closed hence measurable subsets~$A = \{\lambda'\} \subseteq \Lambda$) and similarly~$\tau'\circ\tau = \mathrm{Id}_\Lambda$. Hence,~$\tau$ is bijective. 
\end{proof}

If the underlying space has a product structure,~$\Lambda = \Lambda_1 \times \Lambda_2$, one can take the product of stochastic maps living on the factors~$\Lambda_j$. Let~$T_j: \Lambda_j \to \mathbb{P}(\Lambda_j)$ be stochastic maps, then their product is the point-wise product of measures in~$\mathbb{P}(\Lambda_j)$
\begin{align}
    \left( T_1 \otimes T_2 \right) (\lambda) \equiv T_1(\lambda_1) \otimes T_2(\lambda_2)
\end{align}
where~$\lambda = (\lambda_1, \lambda_2)\in \Lambda_1\times \Lambda_2$. This tensor product is a valid stochastic map (see, e.g.,~\cite{Fritz_2020}).

Any stochastic map~$T: \Lambda \to \mathbb{P}(\Lambda)$ can be viewed as a deterministic map at the cost of adding a random variable to the input. More precisely, there exists a jointly measurable map~$\tau: \Lambda \times (0, 1) \to \Lambda$ such that
\begin{align} \label{eq: determinized_stochastic_map}
    T(A|\lambda) = \int_0^1 \dup c\, \1_A(\tau(\lambda | c))
\end{align}
for all measurable~$A\subseteq \Lambda$ and~$\lambda \in \Lambda$; see~\cite[Thm.~10.7.6]{Bogachev_2007}. In words, the probability~$T(A|\lambda)$ that~$\lambda$ is mapped to something in~$A$ is equal to the probability that~$\lambda' = \tau(\lambda | c)$ lies in~$A$ where~$c$ is uniformly distributed in~$(0, 1)$.

Stochastic maps also act on functions via a pull-back. Let~$T:\Lambda \to \mathbb{P}(\Lambda)$ be a stochastic map and~$f:\Lambda \to \R$ a bounded Borel function. Then
\begin{align}
    T^*\hspace{-2pt}f (\lambda) \equiv \int_\Lambda f(\lambda') T(\dup\lambda' |\lambda)
\end{align}
and the resulting function~$T^*\hspace{-2pt}f$ is again bounded Borel~\cite[Sec.~19.2]{Aliprantis_2006}. 
If we use the \lq determinized version\rq\ from Eq.~\eqref{eq: determinized_stochastic_map} for~$T$ we obtain
\begin{align}
    T^*\hspace{-2pt}f(\lambda) = \int_0^1 \dup c\, f(\tau(\lambda | c)) = \mathbb{E}_{c\sim (0, 1)}[f(\tau(\lambda| c))].
\end{align}
For a probability measure~$\mu$ we have a generalized change of variables formula~\cite[Thm.~19.9]{Aliprantis_2006}
\begin{align}
    \int_\Lambda f(\lambda) \,\dup T_\# \mu(\lambda) = \int_\Lambda   T^*\hspace{-2pt}f(\lambda) \, \dup \mu(\lambda).
\end{align}

Now we fix a topology on~$\Lambda$. Specifically we assume for the remainder of this section that~$\Lambda$ is a manifold\footnote{As explained in the main text, by a \lq manifold\rq\ we mean a topological manifold that may have a boundary and be of non-pure dimension.} equipped with its Borel~$\sigma$-algebra. In particular,~$\Lambda$ is a standard Borel space and a Polish space. 
We use the \emph{topology of weak convergence} on the space of probability measures~$\mathbb{P}(\Lambda)$. Since we are dealing with a metrizable space~$\Lambda$ the weak topology can be characterized by saying a sequence of probability measures~$(\mu_n)_{n\in\N}\subseteq \mathbb{P}(\Lambda)$ converges to some probability measure~$\mu\in\mathbb{P}(\Lambda)$ if and only if expectation values of arbitrary continuous bounded functions~$f\in\mathcal{C}_b(\Lambda)$ converge~\cite[Def.~8.1.1]{Bogachev_2007}, i.e.,
\begin{align}
    \mathbb{E}_{\mu_n}[f] \equiv \int_\Lambda f(\lambda)\dup \mu_n(\lambda) \xrightarrow{n\to \infty} \mathbb{E}_\mu[f].
\end{align}

For a Polish space~$\Lambda$ the weak topology on~$\mathbb{P}(\Lambda)$ is well behaved. Specifically, the weak topology makes~$\mathbb{P}(\Lambda)$ a Polish space, i.e., separable and completely metrizable~\cite[Thm.~15.15]{Aliprantis_2006}.

Also the stochastic maps on~$\Lambda$ behave well: Compositions of continuous stochastic maps as well as their push-forwards are continuous. In the literature a weakly continuous stochastic map is referred to as a stochastic map that has the \lq Feller Property\rq.
\begin{proposition}
    \label{prop: characterizing weakly continuous stochastic maps}
    Let~$T\in\mathrm{End}(\Lambda)$ be a stochastic map. Then the following are equivalent:
    \begin{itemize} \itemsep0em
        \item~$T: \Lambda \to \mathbb{P}(\Lambda)$ is weakly continuous
        \item The push-forward~$T_\#: \mathbb{P}(\Lambda) \to \mathbb{P}(\Lambda)$ is weakly continuous
    \end{itemize}
\end{proposition}
\begin{proof}
    This is the content of~\cite[Thm.~19.14]{Aliprantis_2006}.
\end{proof}

\begin{proposition}
    Compositions of weakly continuous stochastic maps are weakly continuous. More generally, the continuous stochastic maps~$\mathfrak{C}(\Lambda, \mathbb{P}(\Lambda))\subseteq\mathrm{End}(\Lambda)$ form a submonoid. 
\end{proposition}
\begin{proof}
    First, any weakly continuous~$T\in\mathfrak{C}(\Lambda, \mathbb{P}(\Lambda))$ is Borel measurable and therefore automatically a valid stochastic map~\cite[Thm.~19.7]{Aliprantis_2006}, i.e.,~$\mathfrak{C}(\Lambda, \mathbb{P}(\Lambda))\subseteq\mathrm{End}(\Lambda)$.

    Next, compositions of weakly continuous stochastic maps are weakly continuous: Let~$T,\, S \in\mathfrak{C}(\Lambda, \mathbb{P}(\Lambda))$. Then~$T_\#$ and~$S_\#$ are weakly continuous by Prop.~\ref{prop: characterizing weakly continuous stochastic maps}. Using Prop.~\ref{prop: push-forward and composition commute},~$(T\circ S)_\# = T_\# S_\#$ is continuous as a composition of continuous maps. Hence~$T\circ S$ is weakly continuous, again by Prop.~\ref{prop: characterizing weakly continuous stochastic maps}.

    Finally, the identity stochastic map,~$\mathrm{Id}(A|\lambda) = \delta_\lambda(A)$, is weakly continuous,~$\mathrm{Id}\in\mathfrak{C}(\Lambda, \mathbb{P}(\Lambda))$, because its push-forward is continuous as the identity on~$\mathbb{P}(\Lambda)$.
\end{proof}

A similar statement also holds for joint continuity if we \lq append\rq\ some index space (think of the monoid of processes here).
\begin{proposition}
    \label{prop: joint continuity of stochastic maps under composition}
    Let~$\mathcal{E}, \mathcal{F}, \Lambda$ be manifolds. For~$g\in\mathcal{E}$ and~$h \in \mathcal{F}$ let~$T_g, S_h \in \mathrm{End}(\Lambda)$ be stochastic maps such that
    \begin{align} \notag
        \mathcal{E} \times \Lambda \ni (g, \lambda) \mapsto T_g(\lambda) \in \mathbb{P}(\Lambda), \\ 
        \mathcal{F} \times \Lambda \ni (h, \lambda) \mapsto S_h(\lambda) \in \mathbb{P}(\Lambda)
    \end{align}
    are both jointly continuous. Then
    \begin{align}
        \mathcal{E} \times \mathcal{F} \times \Lambda \ni (g, h, \lambda) \mapsto T_g\circ S_h(\lambda) \in \mathbb{P}(\Lambda)
    \end{align}
    is jointly continuous.
\end{proposition}
\begin{proof}
    For~$g\in\mathcal{E},h\in\mathcal{F},\lambda\in\Lambda$ it holds that~$T_g\circ S_h(\lambda) = (T_g)_{\#} S_h(\lambda)$ and
    \begin{align}
        \mathcal{E}\times\mathcal{F}\times \Lambda \ni (g, h,\lambda)\mapsto (g, S_h(\lambda))\in \mathcal{E}\times\mathbb{P}(\Lambda)
    \end{align}
    is jointly continuous. Therefore, it suffices to show that
    \begin{align}
        \mathcal{E}\times\mathbb{P}(\Lambda)\ni(g,\mu)\mapsto (T_g)_{\#}\mu \in\mathbb{P}(\Lambda)
    \end{align}
    is jointly continuous as well. We prove sequential continuity (sufficient because all spaces in question are metrizable): 
    Let~$(g_n, \mu_n)_{n\in \N} \subseteq \mathcal{E}\times \mathbb{P}(\Lambda)$ be a convergent sequence with limit~$(g, \mu)\in \mathcal{E}\times\mathbb{P}(\Lambda)$. This means~$g_n \to g$ in~$\mathcal{E}$ and~$\mu_n\to\mu$ in~$\mathbb{P}(\Lambda)$. In particular,~$\delta_{g_n}\otimes \mu_n \to \delta_g \otimes \mu$ in~$\mathbb{P}(\mathcal{E}\times \Lambda)$~\cite[Ex.~2.1,~Thm.~2.8]{Billingsley_1999}. Now let~$f\in\mathcal{C}_b(\Lambda)$, then
    \begin{align} \notag 
        \mathcal{E} \times \Lambda &\ni (h, \lambda) 
        \mapsto T_h(\lambda) 
        \mapsto \mathbb{E}_{T_h(\lambda)}[f]\\ 
        &= \int_\Lambda f(\lambda') \dup (T_h(\lambda))(\lambda') =  T^*_{h}f(\lambda) \in \R
    \end{align}
    is a continuous function as a composition of continuous maps and bounded because
    \begin{align}
        \sup_{(h, \lambda)\in \mathcal{E}\times \Lambda}\abs{T^*_{h}f(\lambda)}
        &\le \int_\Lambda \norm{f}_\infty \dup(T_h(\lambda))(\lambda') = \norm{f}_\infty.
    \end{align}
    Therefore, we have
    \begin{align} \notag
        \mathbb{E}_{(T_{g_n})_\# \mu_n}[f]
        &= \int_\Lambda f(\lambda) \dup (T_{g_n})_\# \mu_n(\lambda) \\ \notag 
        &= \int_\Lambda T^*_{g_n}f(\lambda) \dup \mu_n(\lambda) \\ \notag 
        &= \int_{\Lambda} \left(\int_\mathcal{E} T^*_{h}f(\lambda)\dup\delta_{g_n}(h) \right) \dup\mu_n(\lambda) \\ \notag
        &= \int_{\mathcal{E}\times\Lambda} T^*_h f(\lambda) \dup (\delta_{g_n}\otimes \mu_n)(h, \lambda) \\ \notag 
        &= \mathbb{E}_{\delta_{g_n}\otimes \mu_n}[(h, \lambda)\mapsto T^*_{h}f(\lambda)] \\ \notag 
        &\xrightarrow{n\to \infty} \mathbb{E}_{\delta_{g}\otimes \mu}[(h, \lambda)\mapsto T^*_{h}f(\lambda)] \\ 
        &= \ldots 
        = \mathbb{E}_{(T_{g})_\# \mu}[f],
    \end{align}
    by the same steps in reverse. Hence,~$(T_{g_n})_\# \mu_n \to (T_g)_\# \mu$ in~$\mathbb{P}(\Lambda)$ as claimed.
\end{proof}

\section{Hilbert's Fifth Problem}
\label{app: lie group from monoid}

We show that for continuous reversible LHV dynamics the monoid of processes~$\Exp$ is a Lie group. By assumption~$\Exp$ is a manifold and a group that is also a topological monoid (meaning the group composition is continuous).

According to the solution to Hilbert's Fifth Problem any locally Euclidean topological group is (isomorphic to) a Lie group~\cite[Chpt.~4.10]{Montgomery_Zippin_1955}. Note that here \lq locally Euclidean\rq\ is defined slightly differently to our own convention (no boundary is allowed and the dimension must be pure). Also, the group inversion must be continuous for a topological group. Hence, we require a variant of that theorem. 
\begin{theorem}
    \label{thm: lie group from monoid}
    Let~$\Exp$ be a group that is also a topological monoid and a topological manifold (as always, in principle allowing for a boundary and non-pure, potentially unbounded, dimension). Then~$\Exp$ is a Lie group.
\end{theorem}
\begin{proof}
    First,~$\Exp$ is algebraically a group and the group multiplication is continuous because~$\Exp$ is a topological monoid. Moreover, as a manifold~$\Exp$ is Hausdorff and locally Euclidean, in particular locally compact since the half spaces~$\R^n_+$ are locally compact. Thus, it follows from Ellis' Theorem~\cite{Ellis_1957} that~$\Exp$ is a topological group (meaning that also the group inversion is continuous).

    Since~$\Exp$ is locally Euclidean the identity~$\mathrm{Id}\in\Exp$ has an open neighborhood~$U\subseteq \Exp$ that is homeomorphic to an open subset of a half space~$\R^n_+$ for some~$n\in\N_0$. In particular, there exists a continuous injective mapping~$\phi: U \to \R^n_+ \subseteq \R^n$. Additionally, again because~$\Exp$ is locally Euclidean, it is locally arc-wise connected since~$\R^n_+$ is (this is a local property inherited by open subsets and preserved under homeomorphisms). It follows that~$\Exp$ is a Lie group~\cite[Cor.~7.4]{Gleason_Palais_1957}.
\end{proof}

Note that~$\Exp$ being a group and also a topological monoid does in general not imply that the group inversion is continuous, this conclusion requires the manifold assumption. Also, while we in principle allow for manifolds with boundary and non-pure unbounded dimension, the conclusion, that~$\Exp$ is a Lie group, means that~$\Exp$ is a smooth manifold without boundary and of constant finite dimension.

\section{Noninteracting Time Evolution}
\label{app: action of noninteracting unitaries}

In this appendix we give the technical proofs for Proposition~\ref{prop: local_unitaries} and Theorem~\ref{thm: local_unitaries}, and present an additional construction of local LHV dynamics for noninteracting unitaries based on the general class of LHV models from App.~\ref{app: general lhv models}.

\begin{proof}[Proof of Prop.~\ref{prop: local_unitaries}]
    The (Lie) group of noninteracting unitaries~$G = \mathrm{PU}(\Hil_1)^N$ together with the subgroups~$G_\pi = \mathrm{PU}(\Hil_1)^{\# \pi} \subseteq \mathrm{PU}(\Hil_\pi)$ for subsets~$\pi \subseteq [N]$ forms a local net of processes, which we take to be our process monoid. 
    For the quantum dynamics we can simply choose the identity,
    \begin{align}
        \Phi: G \to \mathrm{PU}(\Hil_1)^N,\ U\mapsto \Phi_U \equiv U.
    \end{align}
    Then locality (QM) is trivially satisfied and the LHV dynamics will be microscopic (since~$\Phi$ is faithful). We also have local realizability: If a noninteracting unitary factorizes,~$U = \bigotimes_{j=1}^J \Psi_j \in G$ for some partition~$(\pi_j)_{j=1}^J \in \mathfrak{p}([N])$ of the particles and local channels~$\Psi_j \in \mathrm{C}(\Hil_{\pi_j})$, then each local channel~$\Psi_j$ must be unitary since their tensor product is unitary, and factorize since the tensor product factorizes over individual particles. That is,~$\Psi_j \in G_{\pi_j}$. Therefore,~$\Psi_j = \Phi_{\Psi_j} \in \Phi(G_{\pi_j})$ as claimed.

    We extend~$T_U:\Lambda_1 \to \Lambda_1$, defined for~$U\in\mathrm{PU}(\Hil_1)$, to maps~$T_U \equiv \bigtimes_{k=1}^N T_{U_k}: \Lambda \to \Lambda$, defined for noninteracting unitaries~$U = \bigotimes_{k=1}^N U_k \in G$. 
    This is local (LHV): If~$U\in G_\pi$ for a subset~$\pi\subseteq [N]$, then~$U = \bigotimes_{k\in \pi} U_k \otimes \bigotimes_{k \in\pi^c} \1_{\Hil_1}$. Since~$\mathrm{PU}(\Hil_1) \ni U \mapsto T_U$ is a group homomorphism we have~$T_{\1_{\Hil_1}} = \mathrm{Id}_{\Lambda_1}$. Therefore,
    \begin{align} \notag 
        T_U 
        &= \bigtimes_{k\in \pi} T_{U_k} \times \bigtimes_{k\in\pi^c} T_{\1_{\Hil_1}} \\
        &= \bigtimes_{k\in \pi} T_{U_k} \times \mathrm{Id}_{\Lambda_{\pi^c}} 
        \in \mathrm{End}(\Lambda_\pi).
    \end{align}

    To show that the measurement statistics transform correctly, we first lift the condition \lq$L^{a_k}_{U_k^\dagger[x_k]} = T^{-1}_{U_k}(L^{a_k}_{x_k})$\rq\ for~$a_k \in \outcomes_1$,~$x_k\in\measurements_1$ and~$U_k\in\mathrm{PU}(\Hil_1)$ to all noninteracting unitaries: For~$m\in\mathbf{M}$ and a noninteracting unitary~$U = \bigotimes_{k=1}^N U_k \in G$ define
    \begin{align}
        U^\dagger[m] \equiv (a_1,\ldots, a_N | U^\dagger_1[x_1], \ldots, U^\dagger_N[x_N]) \in \mathbf{M}.
    \end{align}
    This is a valid measurement event since~$\measurements$ is invariant under the action of noninteracting unitaries. We have
    \begin{align}
        \mathbf{P}(m|U[\rho]) = \mathbf{P}(U^\dagger[m]|\rho).
    \end{align}
    Now calculate
    \begin{align}
     \notag
        L_{U^\dagger[m]} 
        &=\bigtimes_{k=1}^N L^{a_k}_{U^\dagger_k[x_k]} 
        = \bigtimes_{k=1}^N T_{U_k}^{-1}(L^{a_k}_{x_k}) \\
        &= \left(\bigtimes_{k=1}^N T_{U_k}\right)^{-1} \left(\bigtimes_{k'=1}^N L^{a_{k'}}_{x_{k'}}\right)
        = T_U^{-1}(L_m).
    \end{align}
    From this it follows that the measurement statistics transform correctly, for any initially valid hidden-variable distribution: Let~$\rho\in \States$ be a state,~$U\in G$ a process (i.e., a noninteracting unitary), and~$m\in \mathbf{M}$ a measurement event. Further, let~$\mu'_\rho$ be \emph{any} hidden-variable distribution that is valid for~$\rho$ and the measurements in~$\measurements$. Then we have
    \begin{align} \notag 
        &\mathbf{P}(m| \Phi_U[\rho])
        = \mathbf{P}(m|U[\rho]) 
        = \mathbf{P}(U^\dagger[m]|\rho) \\
        &= \mu'_\rho(L_{U^\dagger[m]})
        = \mu'_\rho(T^{-1}_U(L_m)) = (T_U)_\# \mu'_\rho(L_m).
    \end{align}
    In summary, we have a dynamical LHV model for~$\States,\measurements, \mathrm{PU}(\Hil_1)^N$ with local deterministic microscopic LHV dynamics with free gauge. 

    The process monoid~$G$ is a Lie group, in particular a smooth manifold. If moreover~$\Lambda_1$ is a finite-dimensional (smooth) manifold and the mapping~$\mathrm{PU}(\Hil_1) \times \Lambda_1 \ni (U, \lambda) \mapsto T_U(\lambda) \in \Lambda_1$ is continuous (smooth), then also their product
    \begin{align} \notag 
        G \times \Lambda \ni \, & (U, \lambda) = \left(\bigotimes_{k=1}^N U_k, \lambda\right) \\ \notag 
        &\widehat=\, (U_1,\ldots, U_N, \lambda_1, \ldots, \lambda_N) \\ 
        &\mapsto \bigtimes_{k=1}^N T_{U_k}(\lambda_k) = T_U(\lambda) \in \Lambda
    \end{align}
    is continuous (smooth). Hence, in this case we have continuous (smooth) LHV dynamics.
\end{proof}

\begin{proof}[Proof of Thm.~\ref{thm: local_unitaries}]
    We apply Proposition~\ref{prop: local_unitaries}. The set of all local projective measurements is invariant under the action of noninteracting unitaries. The set of all separable states is invariant as well.
    
    In terms of the measurement rule~$q(a|x, \lambda) = \1_{L^a_x}(\lambda)$ the required properties from the proposition translate to the existence of smooth invertible maps~$T_U:\Lambda_1\to \Lambda_1$ (depending smoothly on~$U\in\mathrm{PU}(\Hil_1)$) such that
    \begin{align}
        q(a|U^\dagger[x], \lambda) = q(a|x, T_U(\lambda)).
    \end{align}
    If all states in~$\States$ are separable, then the separable decomposition of states~$\rho\in\States$ yields a valid LHV model with pure single-particle states as the hidden variables (a smooth manifold) and Born's rule for single particles as the measurement rule:
    \begin{align}
        q(a|x, \ketbra{\psi}) \equiv \Tr(\ketbra{\psi}M^a_x), \quad \ketbra{\psi} \in \mathbb{D}(\Hil_1).
    \end{align}
    In this case we can choose~$T_U(\ketbra{\psi}) \equiv U\ketbra{\psi}U^\dagger$ which is a smooth group action, so the dynamics is smooth and microscopic, and it satisfies the condition above
    \begin{align} \notag
        &q(a|U^\dagger[x], \ketbra{\psi}) 
        = \Tr(\ket{\psi}\hspace{-3pt}\bra{\psi}M^a_{U^\dagger[x]}) \\ \notag
        &= \Tr(\ket{\psi}\hspace{-3pt}\bra{\psi}U^\dagger M^a_x U)
        = \Tr(U\ket{\psi}\hspace{-3pt}\bra{\psi} U^\dagger M^a_x) \\
        &= q(a|x, T_U(\ketbra{\psi})).
    \end{align}
    One may make this measurement rule deterministic and that process retains this property. The dynamics is reversible because it is microscopic. 
    None of the single-particle objects in the construction explicitly depend on the total number of particles~$N$.
\end{proof}

As a second example, we consider the general measurement rules, introduced in App.~\ref{app: general lhv models} for the purpose of showing that expressive LHV models for sets of states and continuous sets of measurements exist. We show that these also satisfy the required condition. In fact, it turns out that the maps~$T_{U}$ again define a smooth group action of~$\mathrm{PU}(\Hil_1)$ via linear bijections on~$\Lambda_1\cong \R^n$, hence the dynamics is again reversible and microscopic.
\begin{proof}
    For any single-particle unitary~$U\in\mathrm{U}(\Hil_1)$ we can expand the function~$\measurements_1 \ni x \mapsto \vec B(U^\dagger[x]) \in\R^K$ into the same basis functions~$B_n$
    \begin{align}
        \label{eq: define d matrix}
        \vec B(U^\dagger[x]) = \sum_{n=1}^K \vec d_n(U^\dagger)B_n(x) = \mathbf{d}_{U^\dagger} \vec B(x),
    \end{align}
    where~$\vec d_n(U^\dagger) \in \R^K$ and~$\mathbf{d}_{U^\dagger}\in \R^{K\times K}$. Since the measurement rule~$\vec q(x, \boldsymbol{\lambda})$ is a function of~$\boldsymbol{\lambda}\vec B(x)$ only, we can read off the (single-particle) hidden-variable transformations:
    \begin{align}
        \boldsymbol{\lambda}\vec B(U^\dagger[x]) = \boldsymbol{\lambda}\mathbf{d}_{U^\dagger} \vec B(x) = T_U(\boldsymbol{\lambda}) \vec B(x)
    \end{align}
    with the linear transformations~$T_U(\boldsymbol{\lambda}) \equiv \boldsymbol{\lambda}\mathbf{d}_{U^\dagger}$. This mapping~$U\mapsto T_U$ is a group action, that is, this dynamics is in particular reversible and microscopic: For~$U=\1$ we clearly have~$\mathbf{d}_{\1^\dagger} = \1_{K\times K}$ so~$T_\1 = \mathrm{Id}$. And for unitaries~$U_1,U_2$ we have
    \begin{align}
        \notag
        \mathbf{d}_{U_1U_2}\vec B(x) 
        &=\vec B(U_1 U_2 [x])  \\ \notag
        &= \sum_{n=1}^K \vec d_n(U_1)\left(\sum_{l=1}^K \vec d_l(U_2) B_l(x)\right)_n \\
        &= \mathbf{d}_{U_1}\mathbf{d}_{U_2}\vec B(x).
    \end{align}
    Therefore,~$\mathbf{d}_{U_1U_2} = \mathbf{d}_{U_1}\mathbf{d}_{U_2}$ and thus~$T_{U_1U_2} = T_{U_1}\circ T_{U_2}$.

    There is one subtlety: We implicitly assumed that the functions~$B_n(U[x])$ can be expanded into the same~$K<\infty$ basis functions that we fixed for the measurement rule. Of course, this is always possible in an approximate sense, however, here we require exact equality. Without a cutoff,~$K=\infty$, this clearly works. For a finite cutoff~$K < \infty$ it turns out that this is also possible, at least for projective measurements. 

    Before coming to the general analysis we consider qubits. In this example case, we expect the required mathematical objects to be familiar to most readers. Projective qubit measurements are parameterized by normal vectors~$\hat n \in S^2 \cong \measurements_1$ indicating the direction of the spin measurement. An orthonormal basis for~$L^2(S^2)$ is given by the (real-valued) spherical harmonics~$Y_{l, m}$. The (non-faithful) action of unitaries~$U\in \mathrm{U}(2)$ on qubit measurements corresponds to rotations~$U[\hat n] = R_U\hat n,\ R_U \in SO(3)$. Now we have the property that rotating the argument of a spherical harmonic does not change the degree~$l$, that is for any rotation~$R\in SO(3)$
    \begin{align}
        Y_{l, m}\circ R \in \mathrm{span}_\R \{Y_{l, m'}\}_{m'=-l}^l.
    \end{align}
    Thus, there are no problems with the cutoff as long as we include all orders~$m=-l,\ldots,l$ for degrees~$l \le l_{\max}$. In this case, the matrices~$\mathbf{d}_U$ are related to the Wigner D-matrix for the rotation matrices~$R_U$. From this, it also follows that the mapping~$\mathrm{U}(\Hil_1)\times \Lambda_1 \ni (U, \boldsymbol{\lambda}) \mapsto T_U(\boldsymbol{\lambda}) \in \Lambda_1$ is smooth.

    We can generalize all of this to qudits of arbitrary dimension~$D$. Indeed the space of (single-particle) projective measurements is isomorphic to the unitary group~$\mathrm{U}(\Hil_1)=\mathrm{U}(D)$ modulo local phases~$\mathrm{U}(1)^D$. The reason is that a projective measurement is specified by the projections onto an orthonormal basis. A fixed reference basis can be transformed into \emph{any} basis via some unitary and unitaries that differ not only by a diagonal unitary lead to different bases. Therefore, for projective measurements the relevant function space is~$L^2(\measurements_1, \mu)$, where~$\measurements_1 = \mathrm{U}(\Hil_1)/\mathrm{U}(1)^D$ (for qubits~$\mathrm{U}(2)/\mathrm{U}(1)^2 \cong S^2$) and~$\mu$ is the~$\mathrm{U}(\Hil_1)$-invariant probability measure on~$\measurements_1$ (derived from the Haar measure). In this picture, the action of unitaries on measurements is given by matrix multiplication
    \begin{align}
        U \mathrm{U}(1)^D = x \in \measurements_1 \ \Rightarrow U'[x] = (U'U)\mathrm{U}(1)^D.
    \end{align}
    The \emph{Peter-Weyl Theorem for homogeneous spaces of compact groups}~\cite{Farashahi_2017} implies that there exists an orthonormal basis of~$L^2(\measurements_1, \mu)$ of the form (for qubits these are exactly the spherical harmonics)
    \begin{align}
        \bigsqcup_{n\in\N} \{B_{n, j}: \measurements_1 \to \C\}_{j=1}^{J_n} \subseteq L^2(\measurements_1, \mu)
    \end{align}
    with the property that~$J_n < \infty$ for all~$n\in \N$ and transforming the argument with a unitary does not mix different~$n$, i.e.,
    \begin{align}
        \left(\measurements_1 \ni x \mapsto B_{n, j}(U[x])\right) \in \mathrm{span}_\C \{B_{n, j'}\}_{j'=1}^{J_n}
    \end{align}
    for all unitaries~$U$. These basis functions are continuous and may be complex-valued. However, we can pass to real-valued basis functions with the same properties. In conclusion, as long as the cutoff~$K$ is a cutoff for~$n\in\N$ and we always include all (finitely many)~$B_{n,j}$ for~$j\in\{1,\ldots,J_n\}$ everything is consistent. Finally, also for this general construction the coefficient matrices~$\mathbf{d}_U$ depend smoothly on~$U\in \mathrm{U}(\Hil_1)$: It follows from eq.~(\ref{eq: define d matrix}) that~$U\mapsto\mathbf{d}_U$ is a continuous homomorphism from the Lie group~$\mathrm{U}(\Hil_1)$ into the Lie group of invertible~$K$-by-$K$ matrices and hence this mapping is automatically smooth.
\end{proof}

\section{Bell LHV Counterexample}
\label{app: Bell LHV counterexample}

In this appendix, we construct hidden-variable distributions for noisy, separable two-qubit states using the Bell LHV base model. Then, we show that their time evolution under a Heisenberg interaction is incompatible with a state-independent velocity field on the hidden-variable space. This proves Theorem~\ref{thm: heisenberg_counterexample}.

For a single-qubit state~$\rho(\vec r)$, specified by the Bloch vector~$\vec r$, a valid choice for the hidden-variable distribution is~\cite{von_Selzam_Marquardt_2025}
\begin{align}
    4\pi \, p_{\vec r}(\hat \lambda) = 4\relu(\vec r \cdot \hat \lambda) + 1 - r,
\end{align}
where~$r = \norm{\vec r}_2$ and~$\relu(x) = x\Theta(x)$. Now, consider a generic two-qubit density matrix~$\rho$. It can be expanded into the Pauli matrices~$\sigma_j$
\begin{align}
    \notag
    &\rho(\vec a, \vec b, T) = \frac{1}{4}\Big( \1\otimes\1 + \vec a\cdot\vec\sigma\otimes\1 + \1\otimes \vec b\cdot \vec \sigma \\
    &+ \sum_{j, k = 1}^3 T_{jk} \sigma_j\otimes\sigma_k \Big),\ \vec a, \vec b \in \R^3,\ T\in\R^{3\times 3}.
\end{align}
We require the singular-value decomposition for the matrix~$T$ in terms of orthonormal bases~$\{\hat u_j\}_{j=1}^3$ and~$\{\hat v_j\}_{j=1}^3$ of~$\R^3$ and singular values~$S_j \ge 0$
\begin{align}
    T = \sum_{j=1}^3 S_j \hat u_j \hat v_j^{\top}.
\end{align}
Then we have the identity~\cite{Ben-Aryeh_Mann_2015}
\begin{align}
    \notag
    &\rho(\vec a, \vec b, T) \\ \notag
    &=  \sum_{j=1}^3 \frac{S_j}{2}\left(\rho(\hat u_j)\otimes \rho(\hat v_j) + \rho(-\hat u_j)\otimes \rho(-\hat v_j)\right) \\ \notag
    &\hspace{10pt}+ a \rho(\hat a)\otimes \rho(\vec 0) + b\rho(\vec 0)\otimes \rho(\hat b) \\
    &\hspace{10pt}+ \left(1-a-b-\sum_{j=1}^3 S_j\right)\rho(\vec 0) \otimes \rho(\vec 0).
\end{align}
Whenever~$a+b+\sum_{j=1}^3 S_j \le 1$, the decomposition above yields a valid separable decomposition from which we can read off a valid hidden-variable distribution
\begin{align}
    \notag
    &(4\pi)^2\, p_{\vec a, \vec b, T}(\hat \lambda_1, \hat \lambda_2) 
    = \sum_{j=1}^3 \frac{S_j}{2} \Big(4\relu(\hat u_j\cdot \hat \lambda_1)4\relu(\hat v_j \cdot \hat \lambda_2) \\ \notag
    &\hspace{20pt}+4\relu(-\hat u_j\cdot \hat \lambda_1)4\relu(-\hat v_j \cdot \hat \lambda_2) \Big) 
    + 4a \relu(\hat a \cdot \hat \lambda_1) \\ \notag 
    &\hspace{20pt}+ 4b \relu(\hat b\cdot \hat \lambda_2) 
    + 1-a-b-\sum_{j=1}^3 S_j \\ \notag 
    &=1-a-b-\sum_{j=1}^3 S_j
    + 4\relu(\vec a \cdot \hat \lambda_1) + 4\relu(\vec b \cdot \hat \lambda_2) \\
    &\hspace{20pt}+ 8 \sum_{j=1}^3 S_j \relu((\hat u_j \cdot \hat \lambda_1) (\hat v_j \cdot \hat \lambda_2)).
\end{align}
Here, we used
\begin{align}
    \notag
    \relu(\alpha x) &= \alpha \relu(x)\ \forall x \in \R,\ \forall \alpha \ge 0,\\
    \relu(xy) &= \relu(x)\relu(y)+\relu(-x)\relu(-y)\ \forall x, y\in \R. 
\end{align}

For the Heisenberg Hamiltonian~$H=\frac{\omega}{4}\sum_{j=1}^3 \sigma_j\otimes\sigma_j$ the von Neumann equation of motion is given by
\begin{align}
    \notag
    &\dot \rho(\vec a, \vec b, T) = \frac{1}{4}\Big(\dot{\vec a} \cdot \vec \sigma\otimes\1 + \1\otimes \dot{\vec b}\cdot \vec \sigma + \sum_{j,k} \dot T_{jk} \sigma_j\otimes\sigma_k\Big), \\
    &\dot{\vec a} = -\dot{\vec b} = \omega \vec z\hspace{-2pt}\left(\frac{T-T^{\top}}{2}\right),\ \dot T = -\omega A\hspace{-2pt}\left(\frac{\vec a - \vec b}{2}\right).
\end{align}
Here,~$\vec z(A)\in\R^3$ is the unique vector such that~$\vec z(A)\cross\vec v = A\vec v\ \forall \vec v \in \R^3$ for an anti-symmetric matrix~$A$. Vice versa,~$A(\vec z)\in\R^{3\times 3}$ is the unique anti-symmetric matrix such that~$A(\vec z)\vec v = \vec z\cross\vec v\ \forall \vec v \in \R^3$ for a vector~$\vec z$.

We obtain valid hidden-variable distributions for all times
\begin{align}
    p_{\vec a, \vec b, T, t}(\hat\lambda_1, \hat \lambda_2) \equiv p_{\vec a(t), \vec b(t), T(t)}(\hat \lambda_1, \hat \lambda_2).
\end{align}
This is a particular choice of a dynamical gauge, however, for the following we only need to assume this for a neighborhood of~$t=0$. These distributions depend continuously on time~$t$ and the hidden variables~$\hat\lambda_1, \hat\lambda_2$. The dependence is smooth almost everywhere ($\relu(x)$ has a kink at~$x=0$). 

Now, suppose there exists a state-independent velocity field~$V(\hat \lambda_1, \hat\lambda_2, t)$ (even allowing time dependence, even though the Hamiltonian is time-independent) for which these distributions solve the continuity equation\footnote{Since the hidden-variable distributions are not everywhere differentiable, derivatives should formally be understood in a weak sense. The velocity field is only assumed to be measurable and bounded. In particular, continuity would suffice for the contradiction.}. In fact, it suffices to assume that this holds at the initial time~$t=0$, i.e.,
\begin{align}
    \label{eq: continuity equation at t=0}
    \partial_t\big|_{t=0} p_{\vec a(t), \vec b(t), T(t)} = -\nabla\cdot\left( p_{\vec a, \vec b, T} V \right),
\end{align}
where~$\nabla$ denotes the divergence with respect to~$(\hat \lambda_1, \hat \lambda_2)$, we have dropped the argument for notational simplicity, and~$V, \vec a, \vec b, T$ correspond to their values at~$t=0$. 

Since the velocity field does not depend on the state, that is, on~$\vec a, \vec b, T$, we are free to choose these parameters as long as they satisfy the condition~$a+b+\sum_{j=1}^3 S_j \le 1$ for all times. This can always be guaranteed by uniformly scaling down~$\vec a, \vec b, T$, or equivalently by considering states~$\rho_0\in\States_v$ for sufficiently small visibility~$v>0$. The strategy in the following is to consider different states and infer what this implies for the (same) velocity field~$V$. Eventually, we will arrive at contradicting properties.

First, note that there are valid initial states that do evolve in time since the Hamiltonian is nontrivial. Thus, also the distribution~$p_{\vec a(t), \vec b(t), T(t)}$ changes in time because it accounts for the change in measurement statistics. In these cases, the left-hand side of Eq.~(\ref{eq: continuity equation at t=0}) is nonzero. Therefore, the velocity field cannot be zero almost everywhere,~$V\not\equiv 0$, otherwise the measurement statistics could not change. 

Now, we also show~\lq$V\equiv 0$ almost everywhere\rq, yielding the contradiction. From here on consider symmetric~$T$ and~$\vec a = \vec b$. In this case,~$\dot{\vec a} = \dot{\vec b} = \vec 0, \dot T = 0$ so the left-hand side of Eq.~(\ref{eq: continuity equation at t=0}) vanishes and we obtain~$\nabla\cdot (p_{\vec a, \vec b, T} V) = 0$. From~$\vec a = \vec b = \vec 0,\ T= 0$ (the maximally mixed state) we infer that the velocity field must be divergence-free since~$p_{\vec a, \vec b, T} = 1/(4\pi)^2$ is a (nonzero) constant. Thus, again for general~$\vec a = \vec b,\ T=T^\top$, the condition becomes
\begin{align}
    0 = 
    V \cdot \nabla p_{\vec a, \vec b, T}.
\end{align}
Now, choose~$\vec a = \vec b = \vec 0,\ T = \pm \epsilon \hat u \hat u^\top$ for~$\epsilon >0$ arbitrarily small and any~$\hat u\in S^2$. The hidden-variable distribution reduces to
\begin{align}
    \notag
    &(4\pi)^2 p_{\pm, \epsilon, \hat u}(\hat \lambda_1, \hat\lambda_2) 
    = 1-\epsilon + 8\epsilon  \relu(\pm (\hat u\cdot \hat\lambda_1)(\hat u\cdot \hat \lambda_2)),\\ \notag
    &(4\pi)^2 \nabla_{\lambda_j} p_{\pm,\epsilon, \hat u}(\hat \lambda_1, \hat\lambda_2)  \\
    &= \pm 8\epsilon \Theta(\pm (\hat u\cdot \hat\lambda_1)(\hat u\cdot \hat \lambda_2)) (\hat u \cdot \hat \lambda_{\neg j}) \hat u,
\end{align}
where we used~$\partial_x\relu(x) = \Theta(x)$ and~$j\in \{1, 2\},\ \neg 1 = 2,\neg 2 = 1$. Writing~$V = (\vec V_1, \vec V_2)$ we can make the condition explicit
\begin{align}
    \notag
    0 &= \Theta\hspace{-2pt}\left( \pm (\hat u\cdot \hat \lambda_1)(\hat u\cdot \hat \lambda_2) \right) \\ 
    &\hspace{0pt}\times \Big[ (\hat \lambda_2 \cdot \hat u)\vec V_1(\hat \lambda_1,\hat \lambda_2) \cdot \hat u + (\hat \lambda_1 \cdot \hat u) \vec V_2(\hat \lambda_1,\hat \lambda_2)\cdot \hat u \Big].
\end{align}
Since this holds for both signs we obtain (after rescaling)
\begin{align}
    \label{eq: important for contradiction}
    0 = (\hat \lambda_2 \cdot \vec u)\vec V_1(\hat \lambda_1,\hat \lambda_2) \cdot \vec u + (\hat \lambda_1 \cdot \vec u) \vec V_2(\hat \lambda_1,\hat \lambda_2)\cdot \vec u 
\end{align}
for all~$\vec u \in \R^3$ and all~$\hat \lambda_1, \hat\lambda_2 \in S^2$\footnote{This equation for fixed~$\vec u \in \R^3$ holds for almost all~$\hat\lambda_1,\hat\lambda_2\in S^2$. In principle, the null-set where it fails could depend on~$\vec u$. However, since both sides of the equation are continuous in~$\vec u$ (for fixed~$\hat\lambda_1, \hat\lambda_2$) it follows that the equation holds for almost all~$\hat\lambda_1,\hat\lambda_2\in S^2$ independent of~$\vec u \in \R^3$, i.e., the null-set where the equation fails does not depend on~$\vec u$. This is enough for the following.}. Note that the velocity is tangential to the hidden-variable space~$\vec V_j \cdot \hat \lambda_j = 0$. So, from the choice~$\vec u = \hat\lambda_j$ we obtain~$\vec V_j \cdot \hat \lambda_{\neg j} = 0$. That is, the velocity field components~$\vec V_1, \vec V_2$ are individually orthogonal to both~$\hat\lambda_1$ and~$\hat\lambda_2$. 
For~$\hat\lambda_1 \neq \pm \hat \lambda_2$ this means
\begin{align}
    \vec V_j = v_j \hat V,\quad \hat V = \frac{\hat \lambda_1\times\hat\lambda_2}{\norm*{\hat \lambda_1\times\hat\lambda_2}}, \quad \norm*{\vec V_j} = \abs{v_j},
\end{align}
where the~$v_j$ are scalar functions of~$\hat\lambda_1$ and~$\hat\lambda_2$. Next, from the choices~$\vec u = \hat \lambda_j + \vec V_j$ in Eq.~(\ref{eq: important for contradiction}) we obtain
\begin{align}
    \label{eq: important for contradiction 2}
    0 = \vec V_1\cdot \vec V_2 + (\hat \lambda_1\cdot \hat\lambda_2)(\vec V_j\cdot \vec V_j) = v_1v_2 + v_j^2 \,\hat\lambda_1\cdot \hat\lambda_2.
\end{align}
That is, for~$\hat\lambda_1\cdot\hat\lambda_2 \neq 0$ we have~$v_1^2 = v_2^2$. In other words
\begin{align}
    \vec V_1 = v\hat V,\qquad \vec V_2 = \pm v \hat V.
\end{align}
Inserting this back into Eq.~(\ref{eq: important for contradiction 2}) leads to
\begin{align}
    0 = v^2 \left(\hat \lambda_1\cdot\hat\lambda_2\pm 1\right).
\end{align}
For~$\hat\lambda_1\neq \pm \hat \lambda_2$ this implies~$\norm*{\vec V_j} = v = 0$. Putting everything together,~$V\equiv0$ for almost all~$ \hat\lambda_1,\hat\lambda_2\in S^2 \ \text{with} \ \hat\lambda_1\cdot\hat\lambda_2 \not \in \{-1, 0, 1\}$. Hence,~$V(\hat \lambda_1, \hat\lambda_2)=0$ almost everywhere and we get a contradiction as claimed.

\section{Nonlocal Cheat LHV Dynamics}
\label{app: cheat dynamics}

In this appendix we provide the remaining details for the proof of Theorem~\ref{thm: cheat LHV dynamics}.
\begin{proof}
    We defined the new hidden-variable space~$\tilde \Lambda_1 \equiv \Lambda_1 \times \States$, the new measurement rule~$\tilde L^a_x \equiv L^a_x \times \States$ for~$x\in \measurements_1$,~$a\in\outcomes_1$, and the new hidden-variable distributions~$\tilde \mu_\rho \equiv \mu_\rho \otimes \delta_\rho^{\otimes N} \in \mathbb{P}(\Tilde{\Lambda})$. The set of states~$\States$ forms a standard Borel space as a measurable subset of the standard Borel space~$\mathbb{D}(\Hil)$. Hence~$\tilde\Lambda_1$ is a standard Borel space as well, i.e., a valid hidden-variable space. The sets~$\tilde L^a_x$ correctly form measurable partitions of~$\tilde\Lambda_1$. 
    
    We model the set of physical processes~$\Exp$ by the free monoid generated by~$\channels$. The quantum dynamics simply reintroduces the algebraic relations: For~$\Psi \in \channels$ define~$\Phi_{\Psi} \equiv \Psi$ and then extend to a monoid homomorphism. In particular, we have~$\Phi(\Exp) = \channels$. Let~$\tilde\lambda = (\lambda, \sigma) \in \Tilde{\Lambda}$ and let~$g=\Psi\in\channels\subseteq \Exp$. Recall that we defined the dynamics as
    \begin{align} \notag
        T_\Psi(\tilde\lambda) 
        &= T_\Psi(\lambda_1,\ldots,\lambda_N, \sigma_1,\ldots, \sigma_N) \\
        &\equiv \mu_{\Psi[\sigma_1]}\otimes \delta_{\Psi[\sigma_1]}^{\otimes N} \in\mathbb{P}(\Lambda \times \States^N).
    \end{align}
    This is well-defined since~$\Psi[\sigma_1] \in \channels[\States] \subseteq \States$. And the~$T_\Psi$ are valid stochastic maps: It holds that~$\sigma_1 \mapsto \Psi[\sigma_1]\mapsto \mu_{\Psi[\sigma_1]}(A)$ is measurable for any fixed measurable~$A\subseteq \Lambda$ (by assumption). Hence,~$\States \ni \sigma \mapsto \mu_{\Psi[\sigma]} \in \mathbb{P}(\Lambda)$ is a stochastic map, denoted by~$\mu_{\Psi[\cdot]}$, and similarly~$\States \ni \sigma \mapsto \delta_{\Psi[\sigma]} \in \mathbb{P}(\States)$ is also a stochastic map, denoted by~$\delta_{\Psi[\cdot]}$. Consider the measurable map~$f: \tilde\Lambda \to \States^{N+1},(\lambda_1,\ldots,\lambda_N, \sigma_1,\ldots,\sigma_N) \mapsto (\sigma_1,\ldots,\sigma_1)$. Then~$T_\Psi = (\mu_{\Psi[\cdot]}\otimes \delta^{\otimes N}_{\Psi[\cdot]})\circ f$ is a valid stochastic map because tensor products and compositions of stochastic maps are stochastic maps,~$f$ being deterministic (see App.~\ref{app: stochastic_maps}).
    
    On~$\Exp$ we define~$T$ by extending it to a monoid homomorphism. Then, for~$g = \Psi_K\circ \cdots \circ \Psi_1 \in \Exp$ and~$\rho\in\States$ it holds that
    \begin{align}
        (T_g)_\# \tilde\mu_\rho
        = \tilde\mu_{\Phi_g[\rho]}.
    \end{align}
    This can be proved by induction on the word length~$K\in\N_0$. For the empty word~$g=\mathrm{Id}$ this is clear since~$(T_{\mathrm{Id}})_\#$ acts as the identity on measures and~$\Phi_{\mathrm{Id}}$ is the identity on states. Now, assume the statement is true for some word length~$K\in\N_0$ and consider a process~$g\in\Exp$ of word length~$K+1$. Then~$g = \Psi \circ g' \in \Exp$ where~$g'\in \Exp$ has word length~$K$ and~$\Psi\in\channels$. We compute
    \begin{align} \notag 
        (T_g)_\# \tilde\mu_\rho
        &= (T_{\Psi}\circ T_{g'})_\# \tilde\mu_\rho 
        = (T_\Psi)_\# (T_{g'})_\# \tilde\mu_\rho \\
        &= (T_\Psi)_\# \tilde\mu_{\Phi_{g'}[\rho]}.
    \end{align}
    We need to show that this is equal to
    \begin{align} 
        \tilde\mu_{\Phi_g[\rho]}
        &= \tilde\mu_{\Phi_{\Psi\circ g'}[\rho]}
        = \tilde\mu_{\Phi_{\Psi}\circ \Phi_{g'}[\rho]} 
        = \tilde\mu_{\Psi[\Phi_{g'}[\rho]]}.
    \end{align}
    Since~$\Phi_{g'}[\rho]\in\States$ it suffices to show that~$(T_\Psi)_\# \tilde\mu_{\rho} = \tilde\mu_{\Psi[\rho]}$ for all~$\rho\in\States$. Indeed,~$T_\Psi$ was defined to make this true:
    \begin{align} \notag
        (T_\Psi)_\# \tilde\mu_\rho 
        &= \int_{\Lambda\times \States^N} T_\Psi(\lambda, \sigma) \dup \hspace{-2pt}\left(\mu_\rho \otimes \delta_\rho^{\otimes N}\right)(\lambda, \sigma) \\ \notag
        &= \int_\Lambda T_\Psi(\lambda, (\rho, \ldots, \rho)) \dup \mu_\rho(\lambda) \\ \notag
        &= \int_\Lambda \underbrace{\mu_{\Psi[\rho]} \otimes \delta_{\Psi[\rho]}^{\otimes N}}_{\substack{\text{indep. of }\lambda}} \,\dup \mu_\rho(\lambda) \\
        &= \mu_{\Psi[\rho]} \otimes \delta_{\Psi[\rho]}^{\otimes N}
        = \tilde\mu_{\Psi[\rho]}.
    \end{align}
    We conclude that the dynamics has a consistent gauge. 
    It immediately follows that the measurement statistics are correct: For~$g \in \Exp$,~$\rho\in\States$ and~$m\in\mathbf{M}$ we have
    \begin{align}
        \notag
        (T_g)_\# \tilde\mu_\rho(\tilde L_m)
        &= \tilde\mu_{\Phi_g[\rho]}(\tilde L_m)
        = \mu_{\Phi_g[\rho]} \otimes \delta_{\Phi_g[\rho]}^{\otimes N}(L_m \times \States^N)  \\
        &= \mu_{\Phi_g[\rho]}(L_m)
        = \mathbf{P}(m|\Phi_g[\rho]).
    \end{align}
    
    Finally assume~$\Lambda_1$,~$\States$, and~$\channels$ to be (finite-dimensional) manifolds and~$\States \ni \rho \mapsto \mu_\rho \in \mathbb{P}(\Lambda)$ to be continuous. Then~$\tilde \Lambda_1$ is a finite-dimensional manifold as the product of finitely many finite-dimensional manifolds, and~$\Exp\cong \bigsqcup_{K\in\N_0}\channels^K$ is a manifold as the countable disjoint union of products of manifolds. Moreover, we have that
    \begin{align} \notag
        \channels \times \Lambda \times \States^N \ni\,  &(\Psi, \lambda, \sigma) \mapsto (\Psi, \sigma_1) \\
        &\mapsto \Psi[\sigma_1] \mapsto  \mu_{\Psi[\sigma_1]} \in \mathbb{P}(\Lambda)
    \end{align}
    is continuous as a composition of continuous maps. Similarly~$(\Psi, \lambda, \sigma) \mapsto \delta_{\Psi[\sigma_1]} \in \mathbb{P}(\States)$ is continuous since~$\States \ni \rho \mapsto \delta_\rho \in\mathbb{P}(\States)$ is~\cite[Ex.~2.1]{Billingsley_1999}. Therefore,~$(\Psi, \lambda, \sigma) \mapsto T_\Psi(\lambda, \sigma)$ is continuous as the tensor product of continuous maps~\cite[Thm.~2.8]{Billingsley_1999}. Then also its extension from~$\channels$ to all of~$\Exp$ is (jointly) continuous (using Prop.~\ref{prop: joint continuity of stochastic maps under composition} in App.~\ref{app: stochastic_maps}). 
    Finally,~$\channels \ni \Psi \mapsto \Phi_\Psi = \Psi$ is trivially continuous and this extends from~$\channels$ to all of~$\mathcal{E}$ because composition of quantum channels is continuous (it corresponds to matrix multiplication).
\end{proof}

\section{Dimensionality Constraint}
\label{app: dimensionality constraint}

In this Appendix we prove Theorem~\ref{thm: dim-constraint-maintext}, the mathematical basis for the dimensionality constraint (Thm.~\ref{thm: dim-constraint for no go}) and the no-go result it implies (Cor.~\ref{cor:physical_nogo}). We recall the precise statement:

\textit{
    Let~$\Exp$ be a Lie group that acts continuously on a finite-dimensional manifold~$\Lambda$ via~$T: \Exp \to \mathrm{Homeo}(\Lambda)$. Let~$G$ be a compact simple Lie group and assume there exists a surjective Lie group homomorphism~$\Phi: \Exp \to G$ with~$\ker T \subseteq \ker \Phi$. Then
    \begin{align}
        \dim G \le \frac{\dim \Lambda (\dim \Lambda + 1)}{2}.
    \end{align}
}

Recall that by a manifold we mean a topological manifold that may have a boundary and be of non-pure dimension. 
The idea of the proof is to use the mapping~$\Phi$ to construct a compact Lie group, with dimension at least that of~$G$, that acts effectively on a connected component of~$\Lambda$. Then the result follows from a theorem by Montgomery and Zippin~\cite[Chpt.~6.3, Thm.~2]{Montgomery_Zippin_1955}:
\begin{proposition}\label{prop:montgomery_and_zippin}
    Let~$H$ be a compact connected Lie group which acts effectively on a connected space~$\Lambda$ that is locally homeomorphic to~$\R^n$ and let~$d$ be the highest dimension of any orbit. Then~$\dim H \le d(d+1)/2$.
\end{proposition}
Our definition of a locally Euclidean space in the main text allows for a boundary and non-pure dimension, which is more general than being \lq locally homeomorphic to~$\R^n$ (with fixed~$n\in\N_0$)\rq. Note that, in contrast to our own convention, in~\cite{Montgomery_Zippin_1955} the latter stronger condition is referred to as \lq locally Euclidean\rq.  

Below we recall some definitions and basic facts about Lie groups and Lie algebras and then prove the above theorem.  As references we use the books by Montgomery and Zippin~\cite{Montgomery_Zippin_1955}, Knapp~\cite{Knapp_1996}, Duistermaat and Kolk~\cite{Duistermaat_Kolk_2000}, Hilgert and Neeb~\cite{Hilgert_Neeb_2012} and Lee~\cite{Lee_2012}.
A Lie group is a group that is also a smooth manifold such that group composition and inversion are smooth maps. For a Lie group~$G$ we write~$\mathfrak{g} = \mathrm{Lie}(G)$ for its Lie algebra,~$G_0$ for its identity component, and~$e\in G$ for the identity element. We have~$\dim G = \dim \mathfrak{g}$~\cite[Thm.~8.37]{Lee_2012}.

The Lie algebra of a Lie group can be identified with its tangent space at the identity and a Lie group homomorphism~$T: H \to G$ lifts to a Lie algebra homomorphism~$T_*:\mathfrak{h}\to \mathfrak{g}$ corresponding to the differential at the identity~$T_* = dT_e$. A Lie algebra~$\mathfrak{g}$ is \emph{simple} if it is nonabelian and has no ideals other than~$0$ and~$\mathfrak{g}$, and \emph{semisimple} if it has no nonzero solvable ideals (equivalently~$\mathrm{rad}(\mathfrak{g})=0$). A Lie group is \emph{simple} (\emph{semisimple}) if it is connected and its Lie algebra is simple (semisimple). We use~\cite[Prop.~12.1.4]{Hilgert_Neeb_2012} to say a Lie algebra~$\mathfrak{g}$ is \emph{compact} if there exists a compact Lie group whose Lie algebra is isomorphic to~$\mathfrak{g}$. In particular, the Lie algebra of a compact Lie group is compact. 
A \emph{continuous Lie group action} of a Lie group~$G$ on a manifold~$\Lambda$ is given by a group homomorphism~$T: G \to \mathrm{Homeo}(\Lambda)$ such that the map~$G\times \Lambda \ni (g, \lambda) \mapsto T_g(\lambda) \in \Lambda$ is jointly continuous.

In the following we write~$H$ and~$G$ for Lie groups with Lie algebras~$\mathfrak{h}$ and~$\mathfrak{g}$, respectively. In the end,~$H$ will correspond to the process monoid~$\Exp$ and~$G$ will correspond to a subgroup of the unitary channels. 

First, we show how to construct a simple compact subgroup~$U\subseteq H$ from a surjective Lie group homomorphism~$J: H \to G$ such that~$J(U) = G$. The idea is to prove a corresponding statement on the Lie algebra level and then integrate back to a Lie group.
\begin{lemma}\label{lem:submersion}
    Let~$J: H\to G$ be a surjective Lie group homomorphism. Then the induced Lie algebra homomorphism~$J_*: \mathfrak{h}\to\mathfrak{g}$ is surjective.
\end{lemma}
\begin{proof}
    By~\cite[Thm.~7.5]{Lee_2012} every Lie group homomorphism has constant rank. Being surjective and of constant rank,~$J$ is a smooth submersion by the Global Rank Theorem~\cite[Thm.~4.14(a)]{Lee_2012}. In particular its differential at the identity~$dJ_e: \mathfrak{h}\to\mathfrak{g}$ is surjective. The induced map~$J_*: \mathfrak{h}\to\mathfrak{g}$ is~$dJ_e$ (potentially up to identifying the tangent space at the identity with the Lie algebra, depending on the definition), hence surjective.
\end{proof}

\begin{lemma}\label{lem:split}
    Let~$J_*: \mathfrak{h}\to\mathfrak{g}$ be a surjective Lie algebra homomorphism onto a simple compact Lie algebra~$\mathfrak{g}$. Then there is a subalgebra~$\mathfrak{u}\subseteq\mathfrak{h}$ such that
    \begin{align}
        J_*|_{\mathfrak{u}}: \mathfrak{u}\to\mathfrak{g}
    \end{align}
    is a Lie algebra isomorphism. In particular~$\mathfrak{u}$ is simple compact and~$J_*(\mathfrak{u}) = \mathfrak{g}$.
\end{lemma}
\begin{proof}
    By the Levi decomposition~\cite[Thm.~B.2]{Knapp_1996} there is a semisimple subalgebra~$\mathfrak{s}\subseteq\mathfrak{h}$ such that~$\mathfrak{h} = \mathfrak{s} \ltimes \mathrm{rad}(\mathfrak{h})$ (vector space direct sum, with~$\mathfrak{s}$ a subalgebra and~$\mathrm{rad}(\mathfrak{h})$ an ideal). The image~$J_*(\mathrm{rad}(\mathfrak{h}))$ is solvable as the image of a solvable algebra under a Lie algebra homomorphism. And it is an ideal of~$\mathfrak{g}$ as the image of the ideal~$\mathrm{rad}(\mathfrak{h})$ under a surjective (Lie) algebra homomorphism. A solvable ideal of the semisimple~$\mathfrak{g}$ is~$0$ (simple implies semisimple~\cite[Prop.~1.13]{Knapp_1996}), so~$J_*(\mathrm{rad}(\mathfrak{h}))=0$. This implies~$J_*(\mathfrak{s})=J_*(\mathfrak{s})+J_*(\mathrm{rad}(\mathfrak{h}))=J_*(\mathfrak{h})=\mathfrak{g}$.

    Decompose the semisimple algebra~$\mathfrak{s}$ into simple ideals by~\cite[Thm.~1.51]{Knapp_1996},~$\mathfrak{s}=\mathfrak{s}_1\oplus\cdots\oplus\mathfrak{s}_n$.
    For each~$j\in[n]$, the image~$J_*(\mathfrak{s}_j)$ is an ideal of~$\mathfrak{g}$ (as the image of an ideal under a surjective Lie algebra homomorphism), and~$\mathfrak{g}$ is simple, so~$J_*(\mathfrak{s}_j)\in\{0,\mathfrak{g}\}$. Not all of these can be~$0$, since~$\sum_j J_*(\mathfrak{s}_j)=J_*(\mathfrak{s})=\mathfrak{g}\neq 0$. So we can fix an index~$j_0, \ \mathfrak{u} \equiv \mathfrak{s}_{j_0}$, with~$J_*(\mathfrak{u})=\mathfrak{g}$. The kernel of the restriction~$J_*|_{\mathfrak{u}}: \mathfrak{u}\to\mathfrak{g}$ is an ideal of the simple algebra~$\mathfrak{u}$~\cite[Prob.~8-31.(b)]{Lee_2012}, and it is not all of~$\mathfrak{u}$ because the image is nonzero. By simplicity the kernel is~$0$. Hence~$J_*|_{\mathfrak{u}}$ is injective, hence an isomorphism (surjectivity was established above). Being isomorphic to the simple compact algebra~$\mathfrak{g}$,~$\mathfrak{u}$ is itself simple and compact (both properties are isomorphism invariants).
\end{proof}

\begin{lemma}\label{lem:integrate}
    Let~$H$ be a Lie group and let~$\mathfrak{u}\subseteq\mathfrak{h}$ be a compact simple Lie subalgebra. Then there is a unique connected Lie subgroup~$U\subseteq H$ with~$\mathrm{Lie}(U)=\mathfrak{u}$, and~$U$ is embedded, compact and simple.
\end{lemma}
\begin{proof}
    By \cite[Thm.~19.26]{Lee_2012} there is a unique connected Lie subgroup~$U\subseteq H$ with~$\mathrm{Lie}(U)=\mathfrak{u}$. By construction~$U$ is a connected Lie group with compact simple Lie algebra~$\mathfrak{u}$, hence simple. In particular,~$U$ is semisimple. By Weyl's Theorem for connected semisimple Lie groups with compact Lie algebra \cite[Thm.~12.1.17]{Hilgert_Neeb_2012},~$U$ is compact. In particular,~$U$ is closed and hence embedded~\cite[Thm.~7.21]{Lee_2012}.
\end{proof}

We can now prove the following proposition.
\begin{proposition}\label{prop:pullback-G}
    Let~$H$ be a Lie group, let~$G$ be a compact simple Lie group, and let~$J:  H\to G$ be a surjective Lie group homomorphism. Then there exists a compact simple embedded Lie subgroup~$U \subseteq H$ with~$J(U) = G$.
\end{proposition}
\begin{proof}
    By Lemma~\ref{lem:submersion} the induced map~$J_*: \mathfrak{h}\to\mathfrak{g}$ is a surjective Lie algebra homomorphism. The Lie algebra~$\mathfrak{g}=\mathrm{Lie}(G)$ is simple and compact because~$G$ is simple and compact. Lemma~\ref{lem:split} yields a simple compact Lie subalgebra~$\mathfrak{u}\subseteq\mathfrak{h}$, with~$J_*|_{\mathfrak{u}}: \mathfrak{u}\to\mathfrak{g}$ an isomorphism. Lemma~\ref{lem:integrate} produces a simple compact embedded Lie subgroup~$U\subseteq H$ with Lie algebra~$\mathfrak{u}$. Finally,~$(J|_U)_* = J_*|_\mathfrak{u}$ is a Lie algebra isomorphism induced by the Lie group homomorphism~$J|_U: U \to G$. Since~$U$ and~$G$ are connected it follows that~$J|_U$ is surjective \cite[Thm.~21.31]{Lee_2012}, i.e.,~$J(U) = G$.
\end{proof}

Next, we show that the dimensionality constraint in Proposition~\ref{prop:montgomery_and_zippin} also holds if one replaces the assumption of an effective action on a connected space that is locally homeomorphic to~$\R^n$ by a nontrivial action of a simple Lie group on a manifold.
\begin{lemma}\label{lem:simple_group_algebraic}
    For a simple Lie group~$G$ the only closed connected normal subgroups are~$\{e\}$ and~$G$.
\end{lemma}
\begin{proof}
    For any closed connected normal subgroup~$N\subseteq G$, we have that~$\mathrm{Lie}(N)$ is an ideal of~$\mathfrak{g}$~\cite[Thm.~20.28]{Lee_2012}, hence~$\mathrm{Lie}(N)=0$ or~$\mathrm{Lie}(N)=\mathfrak{g}$. If~$\mathrm{Lie}(N)=0$ then~$\dim N = \dim \mathrm{Lie}(N) = 0$ and thus~$N$ is a point (as a connected~$0$-dimensional manifold). Hence,~$N=\{e\}$ because it is a group. If~$\mathrm{Lie}(N)=\mathfrak{g}$ then~$\dim N = \dim \mathfrak{g} = \dim G$. This implies~$N = G$: It is a closed full-dimensional subgroup of~$G$. Any such subgroup is an embedded submanifold~\cite[Thm.~20.12]{Lee_2012} of full dimension, hence open in~$G$~\cite[Prop.~5.1]{Lee_2012}. Being nonempty closed and open in the connected space~$G$, it is all of~$G$.
\end{proof}

If the subgroup is not assumed to be connected, we still have the following lemma.
\begin{lemma} \label{lem:zero-dimensional-normal-subgroups}
    Let~$G$ be a simple Lie group and let~$N\subsetneq G$ be a proper closed normal subgroup. Then~$\dim N = 0$.
\end{lemma}
\begin{proof}
    By the closed-subgroup theorem~$N\subseteq G$ is an embedded Lie subgroup~\cite[Thm.~20.12]{Lee_2012}. Its identity component~$N_0$ is an open connected subgroup of~$N$. In particular,~$\dim N = \dim N_0$~\cite[Prop.~7.15]{Lee_2012}. Furthermore,~$N_0$ is closed in~$N$ and hence closed in~$G$ because~$N$ is closed. Also,~$N_0$ is normal in~$G$ because~$N$ is~\cite[Prop.~1.11.7]{Duistermaat_Kolk_2000}.
    Since~$N_0$ is a closed connected normal subgroup of the simple group~$G$, Lemma~\ref{lem:simple_group_algebraic} gives~$N_0 = \{e\}$ or~$N_0 = G$. The latter is not possible since we have the strict inclusion~$N_0 \subseteq N \subsetneq G$ by assumption. Therefore,~$N_0 = \{e\}$. We conclude~$\dim N = \dim N_0 = 0$.
\end{proof}

\begin{lemma} \label{lem:zero-dimensional-kernel}
    Let~$G$ be a simple Lie group acting continuously and nontrivially on a manifold~$\Lambda$. Denote the action by~$T: G\to \mathrm{Homeo}(\Lambda)$. Then~$\ker(T)\subsetneq G$ is a closed normal Lie subgroup and~$\dim \ker(T) = 0$.
\end{lemma}
\begin{proof}
    The kernel of a continuous Lie group action on any Hausdorff space (which~$\Lambda$ is, as a manifold) is a closed subgroup~\cite[Chpt.~1.26]{Montgomery_Zippin_1955} and it is normal as the kernel of a group homomorphism. Hence it is an embedded Lie subgroup~\cite[Thm.~20.12]{Lee_2012}. 
    The kernel is not everything,~$\ker(T) \subsetneq G$, otherwise the action would be trivial contradicting the assumption. 
    Now the claim follows from the previous lemma.
\end{proof}

\begin{lemma}\label{lem:reduce-to-connected-component}
    Let~$T:G\to \mathrm{Homeo}(\Lambda)$ be a nontrivial continuous Lie group action of a connected Lie group~$G$ on a manifold~$\Lambda$. Then there is a connected component~$\Lambda_k\subseteq\Lambda$ such that~$T_k: G \to \mathrm{Homeo}(\Lambda_k)$ is a well-defined nontrivial continuous Lie group action.
\end{lemma}
\begin{proof}
    Let~$\Lambda_k\subseteq \Lambda$ be any connected component (in particular a submanifold of pure dimension). The orbit maps~$a_\lambda: G \to \Lambda,\ g \mapsto T_g(\lambda)$ are continuous for any~$\lambda\in\Lambda$ because the action is continuous. Since~$G$ is connected and~$T_{e} = \mathrm{Id}$ we have~$a_\lambda(G) \subseteq \Lambda_k$ for all~$\lambda\in\Lambda_k$ (a continuous map maps the connected component of a point to the connected component of its image). Thus~$T_g(\Lambda_k) \subseteq \Lambda_k$ for all~$g\in G$. That is,~$T_k: G \to \mathrm{Homeo}(\Lambda_k)$ is a well-defined continuous Lie group action. We can choose a connected component~$\Lambda_k$ such that this action is nontrivial: If it was trivial for all components the original action would be trivial as well contradicting the assumption. 
\end{proof}

\begin{lemma}\label{lem:restrict-to-interior}
    Let~$T:G \to \mathrm{Homeo}(\Lambda)$ be a nontrivial continuous Lie group action of a Lie group~$G$ on a connected manifold~$\Lambda$. Then~$T$ restricts to a nontrivial continuous Lie group action~$\tilde{T}: G \to \mathrm{Homeo}(\Lambda^\circ)$ of~$G$ on~$\Lambda^\circ$, the interior of~$\Lambda$ which is a connected topological manifold without boundary and of pure dimension~$\dim\Lambda^\circ = \dim \Lambda$, in particular, locally homeomorphic to~$\R^n$ for~$n=\dim \Lambda \in \N$.
\end{lemma}
\begin{proof}
    As a connected manifold~$\Lambda$ has pure dimension, i.e., it is a topological manifold with boundary in the classical sense (e.g. as defined in~\cite{Lee_2012}). Therefore, the interior~$\Lambda^\circ$ is a topological manifold without boundary of the same (pure) dimension~\cite[Prop.~1.38]{Lee_2012}, in particular, locally homeomorphic to~$\R^n$ for~$n=\dim \Lambda \in \N$. Also,~$\Lambda^\circ\subseteq \Lambda$ is dense, and connected because~$\Lambda$ is\footnote{Consider a point on the boundary,~$\lambda \in \partial \Lambda$, and let~$U\subseteq \Lambda$ be an open neighborhood of~$\lambda$. As an open subset~$U$ is a submanifold with~$\dim U = \dim \Lambda$. Since the boundary has lower dimension~$\dim \partial \Lambda = \dim \Lambda -1$~\cite[Prop.~1.38]{Lee_2012} it follows that~$U$ cannot be fully contained in the boundary, i.e.,~$U\cap \Lambda^\circ \neq \emptyset$. Hence~$\Lambda^\circ\subseteq\Lambda$ is dense. The argument for why~$\Lambda^\circ$ is connected can be found, for example, in~\cite[Chpt.~4.6]{Gallier_Xu_2013}.}.
    Every point~$\lambda\in\Lambda$ is (exclusively) either an interior point ($\lambda \in \Lambda^\circ$) or a boundary point ($\lambda \in \partial \Lambda$)~\cite[Thm.~2.59]{Lee_2011} and being an interior point is preserved by homeomorphisms. Therefore every homeomorphism~$f \in T(G)\subseteq \mathrm{Homeo}(\Lambda)$ restricts to a homeomorphism~$f|_{\Lambda^\circ}$ of~$\Lambda^{\circ}$. Clearly the mapping~$G\ni g \mapsto (T_g)|_{\Lambda^\circ}\in\mathrm{Homeo}(\Lambda^\circ)$ remains a group homomorphism and~$G\times \Lambda^\circ \ni (g,\lambda) \mapsto T_g(\lambda)\in\Lambda^\circ$ is continuous. Hence,~$G\ni g \mapsto (T_g)|_{\Lambda^\circ}\in\mathrm{Homeo}(\Lambda^\circ)$ defines a continuous group action of~$G$ on~$\Lambda^\circ$. It is nontrivial: If it was, then~$T$ would be trivial as well because~$\Lambda^\circ \subseteq \Lambda$ is dense, contradicting the assumption.
\end{proof}

\begin{proposition} \label{prop:montgomery-zippin-variant}
    Let~$G$ be a compact simple Lie group acting continuously and nontrivially on a finite-dimensional manifold~$\Lambda$. Then
    \begin{align}
        \dim G \le \frac{\dim \Lambda (\dim \Lambda + 1)}{2}.
    \end{align}
\end{proposition}
\begin{proof}
    Denote the action by~$T:G\to \mathrm{Homeo}(\Lambda)$. Since~$G$ is connected there is a connected component~$\Lambda_k \subseteq \Lambda$ such that~$T_k: G \to \mathrm{Homeo}(\Lambda_k)$ is a well-defined nontrivial continuous Lie group action (Lemma~\ref{lem:reduce-to-connected-component}). Since~$\Lambda_k$ is a connected manifold~$T_k$ restricts to a nontrivial continuous Lie group action~$\tilde{T}_k: G \to \mathrm{Homeo}(\Lambda_k^\circ)$ on the interior~$\Lambda_k^\circ$ of~$\Lambda_k$ by the previous lemma.
    By Lemma~\ref{lem:zero-dimensional-kernel} the kernel~$\ker(\tilde T_k) \subsetneq G$ is a closed normal Lie subgroup of dimension zero. We obtain a continuous effective Lie group action~$\tilde T'_k: G'\to \mathrm{Homeo}(\Lambda^\circ_k)$ of the quotient~$G'\equiv G/\ker(\tilde T_k)$~\cite[Chpt.~1.26]{Montgomery_Zippin_1955} with
    \begin{align}
        \dim G' = \dim G - \dim \ker(\tilde T_k) = \dim G.
    \end{align}
    The quotient map~$G \to G'$ is a surjective Lie group homomorphism~\cite[Thm.~21.26]{Lee_2012}. Thus~$G'$ is compact and connected as the image of the compact and connected~$G$ under a continuous mapping. Now Proposition~\ref{prop:montgomery_and_zippin} implies that~$\dim G' \le d(d+1)/2$ where~$d$ is the highest dimension of any orbit of the action. These orbits are subsets of~$\Lambda^\circ_k$ hence~$d \le \dim \Lambda^\circ_k = \dim \Lambda_k$ and
    \begin{align}
        \dim G' \le \frac{\dim \Lambda_k (\dim \Lambda_k + 1)}{2}.
    \end{align}
    Finally use that~$\dim\Lambda_k \le \dim \Lambda$ (by definition~$\dim \Lambda$ is the largest dimension of any connected component) and use that~$\dim G = \dim G'$.
\end{proof}

Finally, we are in a position to prove Theorem~\ref{thm: dim-constraint-maintext}.
\begin{proof}
    By assumption~$G$ is a compact simple Lie group and~$\Phi: \Exp\to G$ is a surjective Lie group homomorphism. By Proposition~\ref{prop:pullback-G} there exists a compact simple embedded Lie subgroup~$U\subseteq \Exp$ with~$\Phi(U) = G$. Since~$T:\Exp\to \mathrm{Homeo}(\Lambda)$ is a continuous Lie group action we can restrict it to~$U \subseteq \Exp$ and obtain a continuous Lie group action~$T|_U: U \to \mathrm{Homeo}(\Lambda)$. This action is nontrivial: If it was trivial, then~$U \subseteq \ker T \subseteq \ker \Phi$ contradicting~$\Phi(U) = G \neq \{e\}$. Hence Proposition~\ref{prop:montgomery-zippin-variant} applies and~$\dim U \le \dim \Lambda (\dim \Lambda+1)/2$. Finally,~$\dim G = \dim \Phi|_U(U) = \dim U - \dim \ker \Phi|_U \le \dim U$~\cite[Thm.~21.27]{Lee_2012}, finishing the proof.
\end{proof}

\section{Non-Continuous LHV Dynamics}
\label{app: reduce Lambda}

We provide the remaining details for the proof of Theorem~\ref{thm: simple hidden-variable space}. Recall that for a state~$\rho \in \States$, a process~$g\in\Exp$, a single-particle measurement event~$(a|x)\in\mathbf{M}_1$, a hidden variable~$\tilde\lambda \in (0, 1)^N \equiv \tilde\Lambda$, and a measurable subset of hidden variables~$\tilde{A}\subseteq \tilde{\Lambda}$, we defined
    \begin{align} \notag
        \tilde L^{a}_x &\equiv J(L^a_x), \\ \notag 
        \tilde \mu_\rho(\tilde A) &\equiv \mu_\rho(\boldsymbol{J}^{-1}(\tilde A)), \\ 
        \tilde T_g(\tilde A |\tilde\lambda) &\equiv T_g\hspace{-2pt}\left(\boldsymbol{J}^{-1}(\tilde A) \big| \boldsymbol{J}^{-1}(\tilde\lambda)\right).
    \end{align}
\begin{proof}
    The space~$(0, 1)$ is a valid hidden-variable space (as a finite-dimensional manifold), and~$\tilde{L}: \measurements_1 \to \mathfrak{p}((0, 1))$ is a valid local measurement rule on~$(0, 1)$ since~$J$ maps measurable partitions to measurable partitions. We have valid probability measures~$\tilde \mu_\rho = \boldsymbol{J}_\# \mu_\rho \in \mathbb{P}((0, 1)^N)$. The maps~$\tilde T_g: (0, 1)^N \to \mathbb{P}((0, 1)^N)$ are valid stochastic maps since~$\tilde T_g(\tilde\lambda) = \boldsymbol{J}_\# T_g(\boldsymbol{J}^{-1}(\tilde\lambda))$, i.e.,~$\tilde T_g=\boldsymbol{J}\circ T_g \circ \boldsymbol{J}^{-1}\in\mathrm{End}(\tilde{\Lambda})$. 
    
    The mapping~$\Exp\ni g \mapsto \tilde T_g \in \mathrm{End}((0, 1)^N)$ is also a monoid homomorphism: For~$\tilde \lambda \in (0, 1)^N$,~$\tilde A \subseteq (0, 1)^N$, and the trivial process~$\mathrm{Id}\in \Exp$ we have
    \begin{align} \notag 
        \tilde T_{\mathrm{Id}}(\tilde A|\tilde \lambda) 
        &=  T_{\mathrm{Id}}(\boldsymbol{J}^{-1}(\tilde A)|\boldsymbol{J}^{-1}(\tilde\lambda)) \\
        &= \delta_{\boldsymbol{J}^{-1}(\tilde \lambda)}(\boldsymbol{J}^{-1}(\tilde A))
        = \delta_{\tilde\lambda}(\tilde A),
    \end{align}
    hence~$\tilde T_{\mathrm{Id}} = \mathrm{Id}_{(0, 1)^N}$ is the identity map. And for~$g, g' \in \Exp$ we have
    \begin{align} \notag 
        &\tilde T_g \circ \tilde T_{g'}(\tilde A |\tilde \lambda)
        = \int_{(0, 1)^N} \tilde T_g(\tilde A|\tilde \lambda') \tilde T_{g'}(\dup \tilde \lambda'|\tilde \lambda) \\ \notag 
        &= \int_{(0, 1)^N} T_g(\boldsymbol{J}^{-1}(\tilde A)| \boldsymbol{J}^{-1}(\tilde\lambda')) \dup\hspace{-2pt} \left(\boldsymbol{J}_\#T_{g'}(\boldsymbol{J}^{-1}(\tilde \lambda))\right)(\tilde \lambda') \\ \notag 
        &= \int_{\Lambda} T_g(\boldsymbol{J}^{-1}(\tilde A)| \lambda') T_{g'}(\dup \lambda'|\boldsymbol{J}^{-1}(\tilde\lambda)) \\ \notag 
        &= T_g\circ T_{g'}(\boldsymbol{J}^{-1}(\tilde A)| \boldsymbol{J}^{-1}(\tilde \lambda)) \\ 
        &= T_{g\circ g'}(\boldsymbol{J}^{-1}(\tilde A)| \boldsymbol{J}^{-1}(\tilde \lambda)) 
        = \tilde T_{g\circ g'}(\tilde A|\tilde \lambda),
    \end{align}
    hence~$\tilde T_g \circ \tilde T_{g'} = \tilde T_{g\circ g'}$ as claimed.
    
    It remains to check the measurement statistics. Let~$m\in\mathbf{M}$,~$\rho\in\States$ and~$g\in\Exp$. Note that
    \begin{align} \notag 
        \boldsymbol{J}(L_m) 
        &= J^{\times N}\hspace{-2pt}\left(\bigtimes_{k=1}^N L^{a_k}_{x_k}\right)
        = \bigtimes_{k=1}^N J(L^{a_k}_{x_k}) \\
        &= \bigtimes_{k=1}^N \tilde L^{a_k}_{x_k}
        = \tilde L_m,
    \end{align}
    hence~$\boldsymbol{J}^{-1}(\tilde L_m) = L_m$. Therefore the quantum measurement statistics are correctly reproduced
    \begin{align} \notag
        &(\tilde T_g)_\#\tilde\mu_{\rho}(\tilde L_m)
        = \int_{(0, 1)^N} \dup \tilde \mu_\rho(\tilde\lambda) \tilde T_g(\tilde L_m|\tilde\lambda) \\ \notag 
        &= \int_{(0, 1)^N}\dup (\boldsymbol{J}_\#\mu_\rho)(\tilde \lambda) T_g(\boldsymbol{J}^{-1}(\tilde L_m)| \boldsymbol{J}^{-1}(\tilde \lambda)) \\ \notag 
        &= \int_{\Lambda} \dup \mu_\rho(\lambda) T_g(L_m | \lambda)
        = (T_g)_\# \mu_\rho(L_m) \\ 
        &= \mathbf{P}(m| \Phi_g[\rho]).
    \end{align}

    If~$\mathfrak{L}$ has local dynamics, then so does~$\tilde{\mathfrak{L}}$ because the quantum dynamics remains unaffected and~$\boldsymbol{J}=J^{\times N}$ is a product over the individual particles. 
    
    If~$\mathfrak{L}$ has deterministic dynamics,~$T_g(\lambda) = \delta_{\tau_g(\lambda)}$, then so does~$\tilde{\mathfrak{L}}$ because
    \begin{align} \notag
        \tilde T_g(\tilde A|\tilde \lambda) 
        &= \delta_{\tau_g(\boldsymbol{J}^{-1} (\tilde\lambda))}(\boldsymbol{J}^{-1}(\tilde A)) \\
        &= \delta_{\boldsymbol{J}\circ \tau_g \circ \boldsymbol{J}^{-1}(\tilde\lambda)}(\tilde A),
    \end{align}
    and~$\boldsymbol{J}\circ \tau_g \circ \boldsymbol{J}^{-1}: (0, 1)^N \to (0, 1)^N$ is a measurable map.
    
    If~$\mathfrak{L}$ has reversible dynamics, then the dynamics is in particular deterministic with bijective~$\tau_g: \Lambda \to \Lambda$. Hence~$\tilde{\mathfrak{L}}$ has reversible dynamics because the maps~$\boldsymbol{J}\circ \tau_g \circ \boldsymbol{J}^{-1}: (0, 1)^N \to (0, 1)^N$ are bijective as well.
    
    Microscopic LHV dynamics is a statement about the unchanged quantum dynamics~$\Phi:\Exp \to \mathrm{C}(\Hil)$, thus~$\mathfrak{L}$ has microscopic dynamics if and only if~$\tilde{\mathfrak{L}}$ does.
    
    If~$\mathfrak{L}$ has dynamics with consistent gauge, then so does~$\tilde{\mathfrak{L}}$ because in that case, similar to an earlier calculation but now for arbitrary measurable~$\tilde A\subseteq (0, 1)^N$
    \begin{align} \notag 
        (\tilde T_g)_\#\tilde\mu_{\rho}(\tilde A)
        &= (T_g)_\# \mu_\rho(\boldsymbol{J}^{-1}(\tilde A)) \\
        &= \boldsymbol{J}_\# (T_g)_\#\mu_\rho(\tilde A),
    \end{align}
    that is
    \begin{align}
        (\tilde T_g)_\#\tilde\mu_{\rho} 
        &= \boldsymbol{J}_\# (T_g)_\#\mu_\rho
        = \boldsymbol{J}_\# \mu_{\Phi_g[\rho]}
        = \tilde\mu_{\Phi_g[\rho]}.
    \end{align}
    
    If~$\mathfrak{L}$ has dynamics with free gauge, then so does~$\tilde{\mathfrak{L}}$. To see this let~$\tilde \mu'\in \mathbb{P}((0, 1)^N)$ be any probability measure that reproduces the measurement statistics of~$\rho\in\States$, i.e.
    \begin{align}
        \tilde \mu'(\tilde L_m) = \mathbf{P}(m|\rho)\quad \forall m\in\mathbf{M}.
    \end{align}
    Then~$\mu'\equiv \boldsymbol{J}^{-1}_\#\tilde \mu' \in \mathbb{P}(\Lambda)$ is valid for~$\rho$ as well
    \begin{align}
        \boldsymbol{J}^{-1}_\#\tilde \mu'(L_m) = \tilde\mu'(\boldsymbol{J}(L_m))
        = \tilde \mu'(\tilde L_m)
        = \mathbf{P}(m| \rho).
    \end{align}
    Then, since~$\mathfrak{L}$ has dynamics with free gauge, and because~$\boldsymbol{J}_\# \mu' = \boldsymbol{J}_\# \boldsymbol{J}^{-1}_\# \tilde \mu' = \tilde \mu'$, we have
    \begin{align}
        (\tilde T_g)_\# \tilde\mu'(\tilde L_m)
        &= (T_g)_\# \mu'(L_m) = \mathbf{P}(m|\Phi_g[\rho]).
    \end{align}

    All these implications also hold in the reverse direction since we can turn~$\tilde{\mathfrak{L}}$ into~$\mathfrak{L}$ in the same way using the isomorphism~$J^{-1}$.
\end{proof}

\section{Local LHV Dynamics}
\label{app: local LHV dynamics}

We prove Theorem~\ref{thm: local LHV dynamics consequences}.
\begin{proof}
    Assume the set of channels~$\channels$ has local LHV dynamics for~$\measurements$ and~$\States$. This means there exists a dynamical LHV model~$(\Lambda_1, L, \mu, \Exp, T, \Phi)$ such that
    \begin{align}
        \mathbf{P}(a| x, \Phi_g[\rho]) = \int_{\Lambda_1^N} \dup(T_g)_\#\mu_\rho(\lambda)\, Q_{a|x}(\lambda)
    \end{align}
    for all~$\rho\in\States$,~$g\in\Exp$,~$x\in\measurements$ and~$a\in\outcomes$. Here~$Q$ is the product of local response functions~$q$, derived from~$L$. Furthermore, the process monoid~$\Exp$ is a local net of processes for~$N$ particles, the LHV dynamics~$T$ is local, and the quantum dynamics~$\Phi$ is local and satisfies local realizability. For~$N=1$ there is nothing to show, so we assume~$N\ge 2$.

    We fix a partition~$\pi=(\pi_j)_{j=1}^J \in \mathfrak{p}([N])$. Let~$\rho\in\States$ be a state. We construct an LHV model for~$\rho$ with respect to the~$J$-particle measurements in~$\measurements[\pi] = \bigtimes_{j=1}^J \channels^\dagger_{\pi_j}[\measurements_{\pi_j}]$.
    To this end, let~$x=(x_{\pi_1}, \ldots, x_{\pi_J})\in\measurements [\pi]$ be such a measurement. 
    First we construct a measurement~$y\in\measurements$ and a process~$g\in\Exp$ such that~$\Phi_g^\dagger[y]=x$: For each~$j\in [J]$ we have~$x_{\pi_j} \in \channels^\dagger_{\pi_j}[\measurements_{\pi_j}]$. That is, there is a channel~$\Psi_{\pi_j} \in \channels_{\pi_j}$ and a measurement~$y_{\pi_j} \in \measurements_{\pi_j}$ such that~$\Psi_{\pi_j}^\dagger[y_{\pi_j}]=x_{\pi_j}$. Local realizability implies that there is a local process~$g_{\pi_j}\in\Exp_{\pi_j}$ with~$\Phi_{g_{\pi_j}} = \Psi_{\pi_j}$. By construction both~$y_{\pi_j}\in \measurements_{\pi_j}$ as well as~$g_{\pi_j}\in\Exp_{\pi_j}$ only depend on the local measurement~$x_{\pi_j}\in \channels^\dagger_{\pi_j}[\measurements_{\pi_j}]$ (and not on the other measurements).

    We set~$g\equiv g_{\pi_1} \circ \cdots \circ g_{\pi_J} \in \Exp$ (the order does not matter due to local commutativity) and~$y\equiv (y_{\pi_1}, \ldots, y_{\pi_J}) \in \bigtimes_{j=1}^J \measurements_{\pi_j} \cong \measurements$. Since the quantum dynamics is local we have~$\Phi_{g_{\pi_j}} \in \mathrm{C}(\Hil_{\pi_j})$ and hence
    \begin{align}
        \Phi_g = \Phi_{g_{\pi_1}} \circ \cdots \circ \Phi_{g_{\pi_J}}
        = \Psi_{\pi_1}\otimes \cdots \otimes \Psi_{\pi_J}.
    \end{align}
    It follows that
    \begin{align} \notag 
        \Phi_g^\dagger[y]
        &= \left( \Psi_{\pi_1}^\dagger[y_{\pi_1}], \ldots, \Psi_{\pi_J}^\dagger[y_{\pi_J}] \right) \\
        &= (x_{\pi_1}, \ldots, x_{\pi_J}) = x.
    \end{align}
    Similarly, since the LHV dynamics is local it holds that~$T_{g_{\pi_j}} \in \mathrm{End}(\Lambda_{\pi_j})$ and hence
    \begin{align}
        T_g 
        = T_{g_{\pi_1}} \circ \cdots \circ T_{g_{\pi_J}}
        = T_{g_{\pi_1}} \otimes \cdots \otimes T_{g_{\pi_J}}. 
    \end{align}
    It follows that~$T_g$ acting on the product of local measurement rules~$Q_{a|y}$ factorizes: Let~$\lambda = (\lambda_{\pi_1},\ldots,\lambda_{\pi_J}) \in \Lambda$, then
    \begin{align}
        \notag 
        &T^*_g Q_{a|y}(\lambda)
        = \int_\Lambda Q_{a|y}(\lambda') T_g(\dup\lambda'|\lambda) \\ \notag
        &= \int_{\Lambda} \prod_{j=1}^J \underbrace{\prod_{k \in \pi_j} q(a_{k} |y_{k}, \lambda'_{k})}_{\substack{\equiv \ q_{\pi_j}(a_{\pi_j}| y_{\pi_j},\lambda'_{\pi_j})}} \bigotimes_{j=1}^J T_{g_{\pi_j}}(\dup \lambda' | \lambda) \\ \notag
        &= \prod_{j=1}^J \int_{\Lambda_1^{\#\pi_j}} q_{\pi_j}(a_{\pi_j} | y_{\pi_j}, \lambda'_{\pi_j}) T_{g_{\pi_j}}(\dup \lambda'
        _{\pi_j} |\lambda_{\pi_j}) \\
        &\equiv \prod_{j=1}^J \Tilde{Q}^{(\pi_j)}_{g_{\pi_j}}(a_{\pi_j}|y_{\pi_j}, \lambda_{\pi_j}).
    \end{align}
    Now, we can see that~$\rho$ is~$\pi$-local via direct computation
    \begin{align}
        \notag 
        &\mathbf{P}(a| x, \rho)
        = \mathbf{P}(a| \Phi^\dagger_g[y], \rho) 
        = \Tr(\rho M^a_{\Phi^\dagger_g[y]}) \\ \notag
        &= \Tr(\rho \Phi^\dagger_g\left[M^a_{y}\right])
        = \Tr(\Phi_g[\rho] M^a_{y})
        = \mathbf{P}(a | y, \Phi_g[\rho]) \\ \notag
        &= \int_\Lambda \dup (T_g)_\# \mu_\rho(\lambda)\, Q_{a|y}(\lambda)
        = \int_\Lambda \dup \mu_\rho(\lambda)\, T_g^* Q_{a|y}(\lambda) \\
        &= \int_\Lambda \dup \mu_\rho(\lambda) \prod_{j=1}^J \Tilde{Q}^{(\pi_j)}_{g_{\pi_j}}(a_{\pi_j}|y_{\pi_j}, \lambda_{\pi_j}).
    \end{align}
    This is a local model with the hidden-variable distribution~$\mu_\rho$ (that is even independent of the partition~$\pi$) and local response functions
    \begin{align}
        \Tilde{q}_j(a_{\pi_j}|x_{\pi_j}, \lambda) = \Tilde{Q}^{(\pi_j)}_{g_{\pi_j}(x_{\pi_j})}(a_{\pi_j} |y_{\pi_j}(x_{\pi_j}), \lambda_{\pi_j}).
    \end{align}
\end{proof}

\section{No Free Gauge}
\label{app: no free gauge}

We prove Proposition~\ref{prop: extend LHV dynamics to representable states}, Lemma~\ref{lem: noisy state cooling} and Theorem~\ref{thm: no free gauge LHV dynamics} which together prove that LHV dynamics with free gauge for all unitaries cannot exist for~$N\ge 3$ particles.

\begin{proof}[Proof of Prop.~\ref{prop: extend LHV dynamics to representable states}]
    Let~$\mu'_{\rho_1}\in\mathbb{P}(\Lambda)$ be a valid hidden-variable distribution for~$\rho_1\in\mathcal{R}$. Then~$t \mu'_{\rho_1} + (1-t) \mu_{\rho_2}$ is a valid hidden-variable distribution for~$\rho$ because~$\rho = t \rho_1 + (1-t) \rho_2$ and the quantum measurement statistics are linear
    \begin{align} \notag 
        &(t \mu'_{\rho_1} + (1-t) \mu_{\rho_2})(L_m)
        = t\mu'_{\rho_1}(L_m) + (1-t)\mu_{\rho_2}(L_m) \\ \notag 
        &= t \mathbf{P}(m|\rho_1) + (1-t)\mathbf{P}(m|\rho_2)
        = \mathbf{P}(m|t\rho_1 + (1-t)\rho_2) \\
        &= \mathbf{P}(m|\rho).
    \end{align}
    Thus, the assumption of free gauge and linearity of the push-forward imply
    \begin{align} \notag 
        &\mathbf{P}(m|\Phi_g[\rho])
        = (T_g)_\# [t \mu'_{\rho_1} + (1-t) \mu_{\rho_2}](L_m) \\ \notag 
        &= t (T_g)_\# \mu'_{\rho_1}(L_m) + (1-t) (T_g)_\# \mu_{\rho_2}(L_m) \\ 
        &= t (T_g)_\# \mu'_{\rho_1}(L_m) + (1-t) \mathbf{P}(m|\Phi_g[\rho_2]).
    \end{align}
    Rearranging and using linearity of quantum channels and the measurement statistics leads to
    \begin{align} \notag
        &(T_g)_\# \mu'_{\rho_1}(L_m)
        = \frac{1}{t}\mathbf{P}(m|\Phi_g[\rho]) - \frac{1-t}{t}\mathbf{P}(m|\Phi_g[\rho_2]) \\
        &= \mathbf{P}\hspace{-2pt}\left(m\Big|\Phi_g\hspace{-2pt}\left[\frac{1}{t}\rho - \frac{1-t}{t}\rho_2\right]\right)
        = \mathbf{P}(m|\Phi_g[\rho_1]).
    \end{align}
    I.e.,~$(T_g)_\# \mu'_{\rho_1} \in \mathbb{P}(\Lambda)$ is a valid hidden-variable distribution for the state~$\Phi_g[\rho_1]$. In particular,~$\Phi_g[\rho_1] \in \mathcal{R}$.
\end{proof}

\begin{proof}[Proof of Lem.~\ref{lem: noisy state cooling}]
    The idea is to permute the basis vectors to shift weight onto a single basis vector for the first particle. 

    Write~$\omega = \ketbra{\psi}$ for some pure state~$\ket{\psi}\in \Hil_1$. Then extend~$\ket{0}\equiv \ket{\psi}$ to an orthonormal basis~$\{\ket{a}\}_{a=0}^{D-1}$ of~$\Hil_1$. Define the unitary~$U$ as a permutation~$\pi$ of the corresponding product basis~$\{\ket{abc}\}_{a, b, c = 0}^{D-1}$ of~$\Hil_1^{\otimes 3}$. Specifically we swap the basis vectors~$\ket{a00}$ with~$\ket{0aa}$:
    \begin{align}\notag 
        \pi(a00)&\equiv(0aa),\ \pi(0aa)\equiv(a00),\ \text{else } \pi(abc)\equiv(abc) \\
        U \ket{abc} &\equiv \ket{\pi(abc)}.
    \end{align}
    We have~$\omega_v = \sum_a A_a \ketbra{a}$ where~$A_0 = v + \tfrac{1-v}{D}$ and~$A_a = \tfrac{1-v}{D}$ for~$a\neq 0$. Set~$\alpha(abc)\equiv A_a A_b A_c$. Then
    \begin{align}
        \omega_v^{\otimes3} 
        &= \left(\sum_{a=0}^{D-1} A_a \ketbra{a}\right)^{\otimes 3} 
        = \sum_{a, b, c} \alpha(abc) \ketbra{abc}.
    \end{align}
    Using that the permutation~$\pi$ is its own inverse,~$\pi^2 = \mathrm{Id}$, leads to
    \begin{align} \notag
        U \omega_v^{\otimes3}  U^\dagger 
        &= \sum_{a, b, c} \alpha(abc) \ketbra{\pi(abc)} \\  
        &= \sum_{a, b, c} \alpha(\pi(abc)) \ketbra{abc}.
    \end{align}
    We arrive at
    \begin{align} \notag 
        &\Tr_{2, 3}(U \omega_v^{\otimes3} U^\dagger )
        = \sum_{a, b, c} \alpha(\pi(abc)) \ketbra{a} \\
        &= \underbrace{\sum_{b, c} \alpha(\pi(0bc))}_{\substack{\equiv C_0}} \ketbra{0} + \sum_{a \neq 0} \underbrace{\sum_{b, c} \alpha(\pi(abc))}_{\substack{\equiv C_a}} \ketbra{a}.
    \end{align}
    We evaluate the coefficients, using that~$A_a = A_1$ for~$a\neq 0$ does not depend on~$a$.
    \begin{align} \notag
        C_0
        &= \sum_{b}\alpha(\pi(0bb)) + \sum_{b\neq c} \alpha(\pi(0bc)) \\ \notag 
        &= \sum_{b} \alpha(b00) + \sum_{b\neq c} \alpha(0bc) \\ \notag 
        &= A_0^3 + \sum_{b\neq 0} A_b A_0^2 \\ \notag 
        &\hspace{20pt} + \sum_{c\neq 0} A_0^2 A_c + \sum_{b\neq 0} A_0 A_b A_0 + \sum_{0\neq b\neq c \neq 0} A_0 A_b A_c \\
        &= A_0^3 + 3(D-1) A_0^2 A_1 + (D-1)(D-2)A_0 A_1^2, 
    \end{align}
    and for~$a\neq 0$
    \begin{align} \notag 
        C_a
        &= \alpha(\pi(a00)) + \sum_{b, c| \neg(b=c=0)} \alpha(\pi(abc)) \\ \notag 
        &= \alpha(0aa) + \sum_{b, c| \neg(b=c=0)} \alpha(abc) \\ \notag 
        &= A_0 A_a^2 + \sum_{b\neq 0} A_a A_b A_0  \\ \notag 
        &\hspace{20pt} + \sum_{c\neq 0} A_a A_0 A_c + \sum_{b\neq0, c\neq 0} A_a A_b A_c \\ \notag
        &= A_0 A_1^2 + 2(D-1) A_0 A_1^2 + (D-1)^2 A_1^3 \\ 
        &= (2D-1)A_0 A_1^2 + (D-1)^2 A_1^3.
    \end{align}
    Note that the result is independent of~$a$. In particular
    \begin{align} \notag 
        1 &= \Tr(\Tr_{2, 3}(U \omega_v^{\otimes3} U^\dagger )) 
        =  \sum_a C_a \\
        &= C_0 + (D-1) C_1,
    \end{align}
    hence~$1-(C_0-C_1) = D C_1$. Thus
    \begin{align} \notag 
        &\Tr_{2, 3}(U \omega_v^{\otimes3} U^\dagger )
        = C_0 \ketbra{0} + C_1 \sum_{a\neq 0} \ketbra{a} \\  \notag
        &= \underbrace{(C_0 - C_1)}_{\substack{\equiv v'}} \ketbra{0} + \underbrace{D C_1}_{\substack{= 1-v'}} \frac{1}{D}\sum_{a=0}^{D-1}\ketbra{a} \\ 
        &= v' \omega + (1-v') \rho_0
    \end{align}
    as claimed. It remains to evaluate~$v'$. Straightforward algebra (or a CAS), and reinserting~$A_0 = v + \tfrac{1-v}{D}$ and~$A_1 = \tfrac{1-v}{D}$ leads to
    \begin{align}
        v' 
        &= v \left( 1 + v(1-v) + \frac{(1-v)^2}{D} \right)
        \equiv \varphi(v).
    \end{align}
    Clearly~$\varphi$ is continuous and~$\varphi(v)\ge v$. Also, for~$v\in [0, 1]$
    \begin{align} \notag
        1-\varphi(v) 
        &= 1-v - v^2(1-v) - v\frac{(1-v)^2}{D} \\ 
        &= (1-v)\left( 1 - v^2 - \frac{v(1-v)}{D}\right)
    \end{align}
    is nonnegative as the product of two nonnegative factors:~$1-v\ge 0$ and
    \begin{align}
        \notag 
        &1 - v^2 - \frac{v(1-v)}{D} 
        = 1-v \frac{Dv+1-v}{D} \\
        &= 1-v \frac{1+(D-1)v}{D} 
        \ge1-\frac{1+(D-1)}{D}
        = 0.
    \end{align}
    Therefore,~$\varphi(v)\le 1$. Consequently, for any~$v\in (0, 1]$, we have that~$\varphi^n(v)$ is a monotonically increasing sequence bounded by one, hence convergent. The limit is then determined by~$0<v^* = \varphi(v^*)$ which has the unique solution~$v^* = 1$.
\end{proof}

\begin{proof}[Proof of Thm.~\ref{thm: no free gauge LHV dynamics}]
    First, observe that for any state~$\rho\in \mathcal{R}$ its single-particle marginal~$\Tr_{\neg 1}(\rho)\equiv \Tr_{2,\ldots,N}(\rho)$ is representable as well by marginalizing the hidden-variable distribution~$\mu_\rho$: Let~$\pi_1: \Lambda \to \Lambda_1$ be the projection onto the first coordinate. Then we have the valid probability measure~$\mu' \equiv (\pi_1)_\# \mu_\rho \in \mathbb{P}(\Lambda_1)$. Let~$(a_1|x_1) \in \mathbf{M}_1$ be a measurement event and~$x_0 \in \measurements_1$ any measurement. Then
    \begin{align} \notag
        &\mu'(L^{a_1}_{x_1})
        =\mu_\rho\hspace{-2pt}\left(L^{a_1}_{x_1}\times \Lambda_1^{N-1}\right) \\ \notag
        &=\mu_\rho\hspace{-2pt}\left(L^{a_1}_{x_1} \times \left[\bigsqcup_{a\in\outcomes_1} L^{a}_{x_0}\right]^{N-1}\right) \\ \notag 
        &= \sum_{a_2,\ldots,a_N\in\outcomes_1} \mu_\rho(L^{a_1}_{x_1} \times L^{a_2}_{x_0} \times \cdots \times L^{a_N}_{x_0})\\ \notag
        &= \sum_{a_2,\ldots,a_N\in\outcomes_1} \mathbf{P}(a_1, a_2,\ldots, a_N | x_1, x_0,\ldots,x_0, \rho) \\ \notag
        &= \sum_{a_2,\ldots,a_N\in\outcomes_1} \Tr(\rho M^{a_1}_{x_1} \otimes M^{a_2}_{x_0} \otimes \cdots \otimes M^{a_N}_{x_0}) \\  \notag
        &= \Tr(\rho M^{a_1}_{x_1}\otimes \1^{\otimes (N-1)})
        = \Tr(\Tr_{2,\ldots,N}(\rho) M^{a_1}_{x_1}) \\
        &= \mathbf{P}(a_1| x_1, \Tr_{\neg 1}(\rho)).
    \end{align}
    Then, also~$\Tr_{\neg 1}(\rho)^{\otimes N}$ is representable, by taking the corresponding product measure: For~$m\in\mathbf{M}$
    \begin{align} \notag
        &\mu'^{\otimes N}(L_m)
        = \prod_{k=1}^N \mu'(L^{a_k}_{x_k}) 
        = \prod_{k=1}^N \Tr(\Tr_{\neg 1}(\rho) M^{a_k}_{x_k}) \\ \notag 
        &= \Tr(\bigotimes_{k=1}^N \left[\Tr_{\neg 1}(\rho) M^{a_k}_{x_k}\right])
        = \Tr(\Tr_{\neg 1}(\rho)^{\otimes N} M_m) \\
        &= \mathbf{P}\hspace{-2pt}\left(m \big| \Tr_{\neg 1}(\rho)^{\otimes N}\right).
    \end{align}
    
    Now consider the representable state~$\rho \equiv v \omega^{\otimes N} + (1-v) \1/D^N \in\States_v\subseteq \mathcal{R}$ where~$\omega\in\mathbb{D}(\Hil_1)$ is a pure single-particle state. We have~$\Tr_{\neg 1}(\rho) = \omega_v$, so~$\omega_{w_0}$ with~$w_0 =v > 0$ is representable. By Lemma~\ref{lem: noisy state cooling} there exists a three-particle unitary~$U$, depending only on~$\omega$, such that~$\Tr_{\neg 1}(U \omega_{w_0}^{\otimes 3} U^\dagger) = \omega_{\varphi(w_0)}$ where~$\varphi:[0, 1]\to [0, 1]$ is a map with~$w_n\equiv\varphi^{n}(w_0) \to 1$ as~$n\to \infty$. Denote by~$\1_{\neg [3]}$ the identity on the remaining particles. Since~$\Phi(\Exp)$ contains all unitary channels it follows that there exists a process~$g\in \Exp$ with~$\Phi_g = \mathrm{Ad}_{U \otimes \1_{\neg [3]}}$. By Corollary~\ref{cor: invariance of representable states}~$\Phi_g[\omega_{w_0}^{\otimes N}] = U \omega_{w_0}^{\otimes 3} U^\dagger \otimes \omega_{w_0}^{\otimes (N-3)}$ is representable. Then by the above
    \begin{align}
        \Tr_{\neg 1}(U \omega_{w_0}^{\otimes 3} U^\dagger \otimes \omega_{w_0}^{\otimes(N-3)}) = \omega_{\varphi(w_0)}
    \end{align}
    is representable as well. Repeating this process finitely many times demonstrates that~$\omega_{\varphi^n(w_0)} = \omega_{w_n}$ and hence~$\omega_{w_n}^{\otimes N}$ are representable states for all~$n\in\N$. Since~$w_n$ converges to~$1$ it gets arbitrarily close to~$1$ for finite~$n$. 
    Then there is a process~$g\in\Exp$ such that the representable state~$\Phi_g[\omega_{w_n}^{\otimes N}]$ is nonlocal, which is a contradiction. Hence no LHV dynamics with free gauge can exist.
    
    Explicitly, let~$\tau\in \mathbb{D}(\Hil)$ be any pure entangled, hence nonlocal, state. Let~$V$ be a unitary that maps the pure product state~$\omega^{\otimes N}$ to~$\tau$. Then,
    \begin{align} \notag 
        V \omega_{w_n}^{\otimes N} V^\dagger
        &= V \left( w_n^N \omega^{\otimes N} +  (1-w_n^N)\tau'(\omega, w_n) \right)  V^\dagger \\ 
        &= w_n^N \tau + (1-w_n^N) V \tau'(\omega, w_n) V^\dagger,
    \end{align}
    for some state~$\tau'(\omega, w_n)\in \mathbb{D}(\Hil)$. Since the set of nonlocal states is open\footnote{The set of \emph{local} states (w.r.t.~all projective measurements) is closed as an intersection of closed sets, given by states satisfying particular Bell inequalities for particular choices of measurements.} we have that~$w_n^N \tau + (1-w_n^N) \tau'$ is nonlocal for~$w_n^N$ sufficiently close to~$1$ and any state~$\tau'$, because it is uniformly close to~$\tau$:
    \begin{align} \notag 
        \norm{w_n^N \tau + (1-w_n^N) \tau' - \tau}_1
        &= (1-w_n^N)\norm{\tau' - \tau}_1 \\
        &\le 2(1-w_n^N),
    \end{align}
    where~$\norm{\cdot}_1$ denotes the trace norm.  
    Finally pick a process~$g\in\Exp$ with~$\Phi_g = \mathrm{Ad}_V$.
\end{proof}

\bibliography{bibliography}

\end{document}